%% file: main.tex
\documentclass{lmcs}

\input{preamble.tex}

\begin{document}

\title[A programming language combining quantum and classical control]{A programming language combining quantum and classical control}

\author{Kinnari Dave}[a]

\author{Louis Lemonnier}[b]

\author{Romain Péchoux}[c]

\author{Vladimir Zamdzhiev}[d]

\address{Aarhus University, Denmark}

\address{University of Edinburgh, United Kingdom}

\address{Université de Lorraine, CNRS, Inria, LORIA, F-54000 Nancy, France}

\address{Université Paris-Saclay, CNRS, ENS Paris-Saclay,
Inria, Laboratoire Méthodes Formelles, 91190, Gif-sur-Yvette, France}


\begin{abstract}
	The two main notions of control in quantum programming languages are often
	referred to as ``quantum control'' and ``classical control''. With the
	latter, the control flow is based on classical information, potentially
	resulting from a quantum measurement, and this paradigm is well-suited to
	mixed state quantum computation. Whereas with quantum control, we are
	primarily focused on pure quantum computation and there the ``control'' is
	based on superposition. The two paradigms have not mixed well traditionally
	and they are almost always treated separately. In this work, we show that
	the paradigms may be combined within the same system. The key ingredients
	for achieving this are: (1) syntactically: a modality for incorporating
	pure quantum types into a mixed state quantum type system; (2)
	operationally: an adaptation of the notion of ``quantum configuration''
	from quantum lambda-calculi, where the quantum data is replaced with pure
	quantum primitives; (3) denotationally: suitable (sub)categories of Hilbert
	spaces, for pure computation and von Neumann algebras, for mixed state
	computation in the Heisenberg picture of quantum mechanics.
\end{abstract}

\maketitle


\section{Introduction}
\label{sec:intro}

\input{intro-fossacs}

\section{Quantum Control}\label{sec:quantum-control}

\input{quantum-control}

\section{Combining Classical and Quantum Control}\label{sec:combining}

\input{combining_qc}

\section*{Acknowledgements}

The authors would like to thank Benoît Valiron for helpful comments and
discussions.
Louis Lemonnier's research was funded by the Engineering and Physical
Sciences Research Council (EPSRC) under project EP/X025551/1 ``Rubber DUQ:
Flexible Dynamic Universal Quantum programming''.
This work is supported by the the Plan France 2030 through the PEPR
integrated project EPiQ ANR-22-PETQ-0007 and the HQI platform
ANR-22-PNCQ-0002; and by the European Union through the MSCA SE project
QCOMICAL (Grant Agreement ID: 101182520).

\bibliographystyle{alphaurl}
\bibliography{references}

\end{document}

%% file: preamble.tex
\usepackage[utf8]{inputenc}
\usepackage{xcolor}
\usepackage[T1]{fontenc}
\usepackage{bussproofs}
\usepackage{mathrsfs}
\usepackage{amsmath,amssymb}
\usepackage{stmaryrd}
\usepackage{mathtools}
\usepackage{yfonts}
\usepackage{hyperref}
\usepackage{cleveref}
\usepackage{marvosym}
\usepackage{fancyvrb}
\usepackage{tikzit}
\usepackage{xfrac}
\usepackage{pifont}
\usepackage{footnote}
\usepackage{enumitem}
\let\subparagraph\paragraph
\usepackage{multirow}
\usepackage{bussproofs}

\usepackage{float}
\usepackage{stackengine}
\usepackage{cleveref}
\usepackage{caption}
\usepackage{array}
\usepackage{subcaption}
\usepackage{braket}
\usepackage{wrapfig}
\usetikzlibrary{quantikz2}
\usetikzlibrary{shapes}

\usepackage{amsthm}

\usepackage{tcolorbox}

\usepackage{proof}

\renewcommand{\vec}[1]{\vv{#1}}

\tikzstyle{tri} = [draw, isosceles triangle, isosceles triangle apex angle=90,
shape border rotate=-180,inner sep=0pt, minimum size=1cm]

\newenvironment{bprooftree}
  {\leavevmode\hbox\bgroup}
  {\DisplayProof\egroup}

\definecolor{babyblue}{rgb}{0.63, 0.79, 0.95}
\definecolor{blush}{rgb}{0.82, 0.62, 0.91}
\definecolor{columbiablue}{rgb}{0.61, 0.87, 1.0}
\definecolor{darkseagreen}{rgb}{0.56, 0.74, 0.56}
\definecolor{darksalmon}{rgb}{0.91, 0.59, 0.48}
\definecolor{deeppeach}{rgb}{1.0, 0.8, 0.64}

\newcommand{\secref}[1]{\S\ref{#1}}%
\newcommand\cmark{\ding{51}}
\newcommand\xmark{\ding{55}}
\newcommand{\defeq}{\triangleq}
\newcommand{\eqdef}{\defeq}
\newcommand{\brr}{\boldsymbol{\longrightarrow}}

\newcommand\trace[1]{\ensuremath{\mathrm{tr}(#1)}}

\newcommand{\cat}[1]{\mathbf{#1}}
\newcommand{\Hilbc}{\mathbf{Hilb}^{\aleph_o}}
\newcommand{\id}{\mathrm{id}}
\newcommand{\iid}{\mathrm{id}}
\newcommand{\rmsucc}{\mathrm{succ}}
\newcommand{\rmctrl}{\mathrm{crtl}}

\newcommand{\NMIU}{\mathbf{NMIU}}
\newcommand{\NCPSU}{\mathbf{NCPSU}}

\newcommand{\bit}{\textit{bit}}
\newcommand{\qbit}{\textit{qbit}}

\newcommand{\matalg}[1]{\mathrm M_{#1}}

\newcommand{\ov}[1]{\overline{#1}}

\newcommand{\lrb}[1]{{\left\llbracket #1 \right\rrbracket}}

\newcommand\restr[2]{{
  \left.\kern-\nulldelimiterspace 
  #1 
  \vphantom{\big|} 
  \right|_{#2} 
  }}

\renewcommand{\bra}[1]{\ensuremath{\left\langle #1 \right|}}
\renewcommand{\ket}[1]{\ensuremath{\left|  #1 \right\rangle}}
\newcommand{\ketbra}[2]{\ensuremath{\ket{#1}\!\bra{#2}}}
\newcommand{\abs}[1]{\left| #1 \right|}
\newcommand{\norm}[1]{\left\Vert #1 \right\Vert}

\newcommand{\inv}{^{*}}
\newcommand{\R}{\mathbb{R}}
\newcommand{\N}{\mathbb{N}}

\newcommand{\C}{\mathbb{C}}

\renewcommand{\vec}[1]{\mathbf{#1}}

\newcommand{\x}{\vec{x}}

\newcommand{\sset}[1]{\ensuremath{\left\{#1\right\}}}
\newcommand{\match}[3]{\mathrm{match}(#1,#2,#3)}
\newcommand{\ip}[2]{\langle #1 , #2 \rangle}

\newcommand{\psu}[1]{\underline{\textbf{#1}}}
\newcommand{\alt}{\mid}

\newcommand{\iso}{\mapsto}

\newcommand{\inl}[1]{\mathbf{inl}{\;#1}}

\newcommand{\ini}[1]{\mathbf{inx}{\;#1}}
\newcommand{\inx}[1]{\mathbf{inx}{\;#1}}
\newcommand{\inr}[1]{\mathbf{inr}{\;#1}}

\newcommand{\bintensor}[2]{#1 \ptimes #2}

\newcommand{\clause}[2]{\mid  #1 \mapsto #2}
\newcommand{\clauses}[1]{\left\{ ~#1~ \right\}}

\newcommand{\isobasiquel}{\clauses{\clause{v_1}{l_1}\mid\dots\clause{v_n}{l_n}}}
\newcommand\isobasique\isobasiquel

\newcommand{\uctrl}[1]{\mathrm{ctrl}~#1}

\newcommand{\isoterm}{\ensuremath{u}}
\newcommand{\OD}[2]{{\rm ONB}_{#1}{#2}}
\newcommand{\ODe}[2]{{\rm ONB^e}_{#1}{#2}}
\newcommand{\ODval}[2]{{\rm ONB}^{\emph{val}}_{#1}{#2}}

\newcommand{\entail}{\vdash}

\newcommand{\mynl}{\\[1ex]}

\newcommand{\ketz}{\ket{{\ps 0}}}
\newcommand{\keto}{\ket{{\ps 1}}}
\newcommand{\pket}[1]{\ensuremath{\ps{\ket{#1}}}}

\newcommand{\Schrod}{Schr\"odinger}

\newcommand{\qif}[3]{\mathbf{qif}~\ps{#1}~\mathbf{then}~\ps{#1} \ptimes #2~\mathbf{else}~\ps{#1} \ptimes #3}

\newcommand{\tinl}[1]{\textit{inl}(#1)}
\newcommand{\tinr}[1]{\textit{inr}(#1)}
\newcommand{\tcase}[3]{\textit{case}\, #1 \,\textit{of}\, \{ \tinl{x} {\to} #2, \ \tinr{y} {\to} #3\}}

\newcommand{\tcaseof}[1]{\textit{case}\ #1 \textit{ of }}
\newcommand{\tlet}[3]{\textit{let}\ #1 = #2\textit{ in } #3}
\newcommand{\tletpair}[4]{\tlet{#1 \otimes #2}{#3}{#4}}
\newcommand{\tletb}[3]{\tlet{\mathcal B(#1)}{#2}{#3}}
\newcommand{\tletpairb}[4]{\tlet{\mathcal B(#1 \otimes #2)}{#3}{#4}}
\newcommand{\tlift}[1]{\textit{lift}(#1)}
\newcommand{\tforce}[1]{\textit{force}(#1)}
\newcommand{\tzero}{\textit{zero}}
\newcommand{\tsucc}[1]{\textit{succ}(#1)}
\newcommand{\tmatch}[3]{\textit{match}\, #1 \,\textit{with}\, \{ \tzero {\to} #2, \ \tsucc{x} {\to} #3\}}
\newcommand{\tpure}[1]{\textit{pure}(#1)}
\newcommand{\tmeas}[1]{\textit{meas}(#1)}
\newcommand{\tun}[2]{\mathcal B({#1})(#2)}

\newcommand{\skipsum}[1]{%
  \mathop{{\vphantom{\sum}}_{#1}\kern-\scriptspace}\!\sum\nolimits
}
\newcommand{\wvdots}{\ \vdots \ }

\newcommand{\DCPO}{\mathbf{DCPO}}
\newcommand{\Hilb}{\mathbf{Hilb}}

\newcommand{\Coiso}{\mathbf{Coisometry}}
\newcommand{\Contr}{\mathbf{Contr}}
\newcommand{\Isometry}{\Iso}
\newcommand{\Unitary}{\mathbf{Unitary}}

\newcommand{\Iso}{\mathbf{Isometry}}

\newcommand{\Ker}{\mathrm{Ker}}
\newcommand{\iim}{\mathrm{Im}}

\newcommand\nat\Rightarrow

\newcommand{\intf}[1]{\left\llbracket #1 \right\rrbracket} 
\newcommand\interp\intf
\newcommand\sem\intf

\newcommand{\cvariables}{\textsc{Vars}}
\newcommand{\ctypes}{\textsc{Types}}
\newcommand{\cterms}{\textsc{Terms}}
\newcommand{\cvalues}{\textsc{Values}}

\newcommand{\ps}[1]{\boldsymbol{#1}}

\newcommand{\pterms}{\textbf{\textit{Terms}}}%
\newcommand{\pbterms}{\textbf{\textit{Basis Values}}}%
\newcommand{\pvalues}{\textbf{\textit{Values}}}%
\newcommand{\punitaries}{\textbf{\textit{Unitaries}}}%
\newcommand{\pexp}{\textbf{\textit{Expressions}}}%
\newcommand{\ptypes}{\textbf{\textit{Ground quantum types}}}%
\newcommand{\putypes}{\textbf{\textit{Unitary types}}}%
\newcommand{\Nat}{\mathrm{Nat}}
\newcommand{\isotype}[2]{U(#1,#2)}
\newcommand{\ptype}{\boldsymbol Q}%
\newcommand{\ptypeone}{\boldsymbol Q_1}%
\newcommand{\ptypetwo}{\boldsymbol Q_2}%
\newcommand{\pbasic}{\boldsymbol I}%
\newcommand{\pGamma}{\ps \Gamma}%
\newcommand{\pQ}{{\ensuremath{\ps Q}}}%
\newcommand{\pqbit}{\textbf{qbit}}%
\newcommand{\ptqbit}[1]{\textbf{qbit}^{{#1}}}%
\newcommand{\palpha}{\ps \alpha}%
\newcommand{\pbeta}{\ps \beta}%
\newcommand{\entailpure}{\ps \vdash}%
\newcommand{\entailiso}{\Vdash}
\newcommand{\pdash}{\entailpure}%
\newcommand{\udash}{\entailiso}%
\newcommand{\isotypeonetwo}{U(\ptypeone,\ptypetwo)}
\newcommand{\isotypetwoone}{U(\ptypetwo,\ptypeone)}

\newcommand{\ptimes}{\boldsymbol\otimes}
\newcommand{\pplus}{\boldsymbol\oplus}
\newcommand{\pNat}{\textbf{qnat}}
\newcommand{\pnat}{\pNat}%
\newcommand{\phad}{\mathrm{Had}}
\newcommand{\pb}{{\ensuremath{\ps b}}}
\newcommand{\pcnot}{\mathrm{CNOT}}

\newcommand{\pt}{\boldsymbol{t}}
\newcommand{\pv}{\boldsymbol{v}}

\newcommand{\pe}{\boldsymbol{e}}

\newcommand{\psucc}[1]{\textbf{S}{\;#1}}
\newcommand{\pzero}{\textbf{0}}

\newcommand{\unibasique}{\clauses{\clause{\ps b_1}{\ps e_1}\ \dots\clause{\ps b_n}{\ps e_n}}}

\newcommand{\unibasiqueppm}{\clauses{\clause{\ps b'_1}{\ps e'_1}\ \dots\clause{\ps b'_m}{\ps e'_m}}}
\newcommand{\unibasiqueshort}{\{\ldots \mid \! {\ps b_i} \mapsto {\pe_i} \ldots\}}
\newcommand{\unibreduit}{\clauses{\clause{\ps b_i}{\ps e_i}}_i}
\newcommand{\confs}{\texttt{Conf}}
\newcommand{\vconfs}{\texttt{VConf}}
\newcommand{\svconfs}{\texttt{SVConf}}
\newcommand{\wfconfs}[2]{\mathrm{WF}_{#1}(#2)}
\newcommand{\pdot}{\ps \cdot}
\newcommand\psdot\pdot
\newcommand{\pair}[2]{#1 \ptimes #2}
\newtcbox{\rombox}[1][gray]{on line,
arc=7pt,
colback=#1!35!white,
colframe=black,
before upper={\rule[-3pt]{0pt}{10pt}},boxrule=0pt,
boxsep=0pt,left=4pt,right=4pt,top=2pt,bottom=2pt}

\input{tikzit.tikzstyles}
\input{tikzit.tikzdefs}

%% file: tikzit.tikzstyles
\tikzstyle{empty}=[shape=circle, tikzit fill={rgb,255: red,191; green,191; blue,191}]
\tikzstyle{dot}=[fill=black, draw=black, shape=circle]
\tikzstyle{gather}=[shading=gatherballshading, outer color={zx_grey_thick}, inner color={zx_grey}, fill={zx_grey}, draw=black, tikzit category=scal, rounded corners=0.8mm, regular polygon, regular polygon sides=3, shape border rotate=-90, inner sep=1.6pt, anchor=center, outer sep=-.1cm, line width=0.75 pt]
\tikzstyle{hook}=[right hook->, draw=black, tikzit draw=magenta]
\tikzstyle{mono}=[draw=black, to reversed->]

\tikzstyle{to}=[draw=black, ->]
\tikzstyle{equal-arrow}=[-, double equal sign distance]
\tikzstyle{dashed-arrow}=[dashed, ->]
\tikzstyle{tirets}=[-, draw=black, dashed]
\tikzstyle{double}=[<->]

%% file: intro-fossacs.tex
There are two important paradigms in the design of quantum programming languages -- ``classical control'' and ``quantum control''.
In the classical control approach (\emph{e.g.},
\cite{selinger2004towards,selinger2009quantum,selinger2006lambda,kenta-bram,qpl-fossacs,popl22})
the control flow of a program is
conditioned on classical information that may result from quantum measurements.
Type systems for quantum programming languages that are based on classical
control are able to represent a variety of effects, \emph{e.g.}, quantum state
preparation, state evolution through the application of unitary operators, and
probabilistic effects induced by quantum measurements. Because of this, it is
natural to conceptualize such lambda-calculi using quantum structures and
models that are suited to describing mixed-state quantum computation and
information (\emph{e.g.}, density matrices, superoperators).
In the quantum control approach (\emph{e.g.}, \cite{sabry2018symmetric,diazcaro2019unitary,diazcaro2022unitary,pablo2022unbounded,valiron2022semantics}), one usually places an emphasis on pure state quantum
computation and more specifically on superposition of terms, pure quantum
states and unitary operators. In approaches that utilize classical control, one
often starts with a selection of constants that represent some unitary
operators (\emph{e.g.}, Hadamard, CNOT) and more complex unitary operations can be
described in a circuit-like fashion by composing the atomic unitaries in a
suitable way. In the quantum control approach, unitary operators are
defined through more fundamental primitives that do not require the programmer
to specify a circuit-like decomposition of the unitary operation. Because of
this, a considerable economy in terms of syntax can be achieved with quantum
control. 
\input{intro-fig.tex}
For example, consider the circuit in Figure~\ref{fig:prog}. It uses the
well-known \textit{Toffoli} gate, which would have to be decomposed as a large
circuit like the one in the figure. However, the program representing this
circuit can instead be written as a simple lambda term in our language given in
Figure~\ref{fig:prog}, where $\textit{Tof} \defeq \mathcal
B(\mathrm{ctrl}~\pcnot)$, $\mathcal B$ being a modality incorporating pure quantum data into mixed quantum data. 
The $\pcnot$ can be defined using symmetric pattern matching
as (see Example~\ref{ex:unitaries} for more details):
\begin{align*}
		\scalebox{1}{$\pcnot$} &\scalebox{1}{$\ \colon \isotype{\ptqbit{2}}{\ptqbit{2}}$}
		\\
		\scalebox{1}{$\pcnot$} &\scalebox{1}{$\defeq
			\left\{ \begin{array}{lcl}
				\mid\bintensor{\keto}{\ketz} & \iso & \bintensor{\keto}{\keto} \\
				\mid \bintensor{\keto}{\keto} & \iso & \bintensor{\keto}{\ketz} \\
				\mid\bintensor{\ketz}{\ket{\ps x}} & \iso & \bintensor{\ketz}{\ket{\ps x}}
		\end{array} \right\}$}
\end{align*}
To highlight this further note that the controlled Hadamard can be defined similarly
to the Toffoli in our language, but a typical classically controlled quantum programming language
would require two $H$ gates, six $T$ gates, and one $CNOT$ gate.
Thus, our language offers an ease of writing programs by abstracting away
quantum circuits. We aim to shift focus from drawing quantum circuits, which
is a more \textit{low-level} approach to writing code, to writing programs
directly using syntax, which is a more \textit{high-level} approach. In
classical (non-quantum) programming, we can use high-level languages to write
algorithms without having to specify boolean circuits or to encode integers
using tuples of bits/booleans. With our paper, we are trying to advance quantum
programming abstractions towards such a direction. For example, an individual
quantum natural number $\ket n$ can be used instead of a tensor of qubits $\ket{q_1q_2...q_k}$
whose binary encoding corresponds to $n$ (see Example~\ref{exa:walk} for a
quantum walk using quantum natural numbers). We view boolean circuits and
binary encoding of integers as low-level and likewise for their quantum
counterparts.  Our aim is to develop higher-level quantum abstractions.

\subsection*{Our Contributions}
\label{sec:contr}
Because of the potential presence of quantum entanglement, it is impossible (in
general) to decompose a quantum state into a non-trivial tensor product of two
smaller quantum states. This poses problems for the design of quantum
programming languages. If we adopt a naive approach, where we assume that we
have a type $\pqbit$ whose closed values are the admissible qubit states, then
following a standard (linear) typing discipline, the closed values of type
$\pqbit \ptimes \pqbit$ would not capture the admissible entangled states. The
problem lies in the fact that in a standard (linear) type system, the closed
values of type $A \otimes B$ are the pairs of closed values of type $A$ and
type $B$. The quantum lambda-calculus (QLC) \cite{quantitative}, which is based
on classical control, overcomes this problem by using \emph{quantum
configurations}. In the QLC, the type $\pqbit$ of qubits is an opaque type in
the sense that there are no closed values of such a type. Instead, program
states are described via quantum configurations which are triples $(\ket \psi,
\ell, M)$, where $\ket \psi \in \mathbb C^{2^n}$ is a \emph{pure} quantum state
(possibly entangled); $M$ is a program possibly containing free variables of
type $\pqbit$; $\ell$ is a \emph{linking} function that maps the free quantum
variables of $M$ to appropriate components of the state $\ket \psi$. It is
possible to also view $\ell$ as a unitary permutation acting on $\ket \psi$ and
this view is important for our development. Such configurations, even though
they are not part of the user-facing syntax, allow us to reason about entangled
states and the operational semantics are described via a relation $(\ket \psi,
\ell, M) \to_p (\ket{\psi'}, \ell', M')$, with $p$ the probability of
reduction.

The main idea that allows us to combine both quantum and classical control is
to modify the quantum configurations described above by replacing the quantum
state $\ket \psi$ with a suitable \emph{pure quantum term} $\ps t$ and to
replace the linking function $\ell$ with a suitable \emph{unitary permutation}
$u_\sigma$, where both $\ps t$ and $u_\sigma$ are syntactic constructs from our
pure quantum language that may be assigned types.

In order to accurately model the situation with the QLC, the term $\pt$ has to
correspond to a \emph{normalized} vector and we replace the unitary constants from the
QLC with programmable unitary constructs, so the pure subsystem also has to
ensure their \emph{unitarity}. Our pure
subsystem is based on ideas described in \cite{phd-kostia,sabry2018symmetric},
but we further build on this by introducing an \emph{equational theory} for our
pure subsystem, which has unique normal forms and we describe a
\emph{denotational semantics} for it that is \emph{sound and complete} with
respect to the equational theory.


Next, we integrate the pure quantum subsystem into our main calculus, which
allows us to describe both quantum and classical information and, therefore also
classical control. This is a variant of the QLC with the 
addition of a modality $\mathcal
B(\pQ)$ which allows us to view pure quantum types $\pQ$ as types of mixed
state quantum operations in the Heisenberg picture. The intuition behind this
is the following: mixed states in the \Schrod{} picture can be seen as Completely Positive Trace Preserving (CPTP) maps $(1 \mapsto \rho) \colon \mathbb C \to \mathcal T(H)$ whereas mixed states in
the Heisenberg picture are given by continuous linear maps of the form $\trace{\rho -}
\colon \mathcal B(H) \to \mathbb C, $
both of which are determined by a choice of density operator $\rho \colon H \to H$
(see \cite[Section 7]{cho2014semantics} for detail on this matter)\footnote{$\mathcal T(H)$ and $\mathcal B(H)$ stand for the trace-class and bounded
linear operators, respectively, on a Hilbert space $H$.}.
This modality allows us to replace the \emph{constants} for state
preparation and unitaries acting on qubits in the QLC with \emph{terms and
expressions} from the pure subsystem acting on more general quantum types
beyond qubits. The operational semantics are described via a suitable adaptation
of the aforementioned quantum configurations. The denotational semantics follows previous work
on quantum programming semantics based on von Neumann algebras
\cite{kenta-bram,qpl-fossacs,popl22} and we show, in addition, that the
assignment $\mathcal B(-)$ may be extended to a strict monoidal functor
(between the relevant categories of Hilbert spaces and von Neumann algebras)
that is crucial for our denotational semantics.

Overall, the two main contributions of our work are: (1) we show that quantum control
and classical control can be combined in a syntactic, operational and
denotational sense by integrating a pure quantum control subsystem as part of
the meta-theory and syntax of a quantum lambda-calculus; (2) an equational
theory for a (sub)system for quantum control with unique normal forms and a
sound and complete denotational semantics. 
 

\subsection*{Related Work}
\label{sec:r-work}
\input{intro-table.tex}

We begin by commenting on semantic approaches to classical control. In \cite{quantitative}, the authors describe a QLC where the
higher-order primitives are interpreted using techniques from the semantics of
quantitative models of linear logic. In
\cite{quantum-game,quantum-game-full-abstraction,marc-phd}, the author(s) work
with QLCs whose denotational semantics is given via game semantics.
More recently, in \cite{enriched-presheaf}, the authors have shown how to
interpret a QLC with recursive types using presheaves that are enriched over a
suitable base category and they also prove full abstraction. The approach that we use in this paper is based on von
Neumann algebras which have been previously used for the semantics of quantum
programming languages \cite{kenta-bram,popl22,qpl-fossacs}. In particular, if
we disregard the interaction between the pure subsystem and the main calculus,
then our denotational model for the main calculus coincides with the one in
\cite{kenta-bram}. The reason that we choose von Neumann algebras over the
other approaches is because we think this gives the model which is the closest
to mathematical physics out of the ones mentioned. Furthermore, our primary
focus is not on developing the semantics of the classical paradigm;
instead it is on the \emph{interaction} between the pure quantum control
subsystem and the main calculus based on classical control. Another
distinguishing feature of our system, compared to the ones above, is the
addition of types such as $\mathcal B({\pnat})$ which allow us to describe
states in infinite-dimensional Hilbert spaces, such as $\ell^2(\mathbb N)$, and
which we do not think can be easily included using the other approaches.
For the quantum control paradigm, relevant work includes
\cite{diazcaro2019unitary,diazcaro2022unitary,pablo2022unbounded,semimodules,lineal}. What all these approaches have in
common is that they have specific terms in their syntax for representing
superposition. However, these approaches do not ensure unitarity of the
relevant function spaces, in general. 

Our approach, instead, is based on related work first described in
\cite{sabry2018symmetric} and then later in \cite{phd-kostia}, known as
\emph{symmetric pattern matching}. One of the most distinctive features of this
approach to quantum control is that the terms that introduce superposition are
subject to strict formation conditions that ensure (linear algebraic)
normalization of the corresponding vectors. Furthermore, this approach to
quantum control also imposes strict admissibility criteria on the unitary
function spaces that can be formed through the type system and this ensures
unitarity of the relevant expressions.  Because of these considerations, we
choose this approach to model quantum control in our pure subsystem. 

However, we have made some changes, in particular, we replace the operational
semantics from these works with an equational theory instead. This works better
for our approach to \emph{combining} quantum and classical control within the
quantum configurations (in QLCs the quantum data is considered modulo
equality). We continue building on this development by describing a
denotational semantics that is sound and complete with respect to the
equational theory, which the original paper on symmetric pattern matching did
not have the right tools to prove. This is the first such result for languages
with quantum control.

Another paper which combines quantum and classical control is
Qunity~\cite{qunity}. However, compared to our work, Qunity does not have an
operational semantics (and therefore, no adequacy statement), but instead there
is a compilation procedure to quantum circuits (which is not the focus of our
paper). Furthermore, the denotational semantics of Qunity: (1) does not ensure
unitarity and normalization of the pure primitives in its language, whereas
ours does; (2) does not ensure trace-preservation of the mixed primitives in
the \Schrod{} picture, whereas ours does indeed ensure unitality of the mixed
primitives in the Heisenberg picture (which corresponds to trace-preservation
under the Heisenberg-\Schrod{} duality). Other differences include that Qunity
is restricted to finite-dimensional quantum data and they do not have
higher-order lambda abstractions for dealing with mixed-state primitives.
Instead, Qunity relies on a try-catch mechanism to combine quantum and
classical control, whereas our approach is based on quantum configurations. 

Proto-Quipper-C~\cite{quipper-rc} also combines the classical control aspect of
Quipper with an additional syntax for quantum control. Their approach and ours
differ in many aspects, since they focus on circuit description. They assume a
set of gates (which could be any), and their quantum control feature introduces
a controlled gate, while our quantum control is introduced via pattern
matching. The latter improves readability and expressivity, for example towards
natural numbers. On the semantics aspect, a sound categorical model is provided
for Proto-Quipper-C, but no adequacy result is stated.

The table in Figure~\ref{tab:comp} gives a summarized comparison of our work
with other quantum programming languages in the literature with similar
properties. \textit{Normalized quantum terms} refers to the property that terms
representing pure/mixed quantum states always correspond to complex vectors
with norm one/density operators with trace one.  \textit{Adequacy} refers to
the property that two observationally distinct programs have distinct
denotational interpretations. Note that a cross in the table means that the
statement is not proven in the relevant papers.

\subsection*{Conference Paper.} 
This is an extended version of a paper \cite{dave2025control} published in the 
proceedings of FoSSaCS'25. Compared to the conference paper, the syntax of
quantum control is more detailed, with full typing rules and definition of
substitution. We also added details about the denotational semantics as
unitaries and isometries between separable Hilbert spaces, which is now fully
worked out with all the proofs. All the elements around the equational theory
are also added. Finally, the proof of completeness is incorporated into the 
text. Furthermore, similarly to quantum control, the proofs for the 
fully-fledged syntax have been added, with a more detailed discussion on the 
denotational semantics in von Neumann algebras.

%% file: intro-fig.tex
\begin{figure}[t]
	\begin{subfigure}[b]{0.2\textwidth}
		\centering
		\begin{align*}
			\scalebox{1}{$\lambda y. \textit{let}\ \mathcal B(x \otimes xs)=\textit{Tof}(y)$} \\
			\scalebox{1}{$\textit{in case}\ \tmeas{x} \ \textit{ of }$} \\
			\scalebox{1}{$\{ \tinl{\ast} \to \tun{\pcnot}{xs},$} \\
			\scalebox{1}{$\tinr{\ast} \to xs \}$} 
		\end{align*}
	\end{subfigure}
	\hspace{0.2em}
	\begin{subfigure}[b]{0.6\textwidth}
		\centering
		\raisebox{3em}{
			\scalebox{0.47}{
				\begin{quantikz}
					&&\gategroup[3, steps=14, style={dashed,rounded corners, fill=gray!10, inner xsep=2pt}, background, label style={label position=below, anchor=north, yshift=-0.2cm}]{\textit{Toffoli}} 
					&&&&\ctrl{2}&&&&\ctrl{2}&&\ctrl{1}&\gate{T}&\ctrl{1}&& \meter{} &&\ctrl[vertical wire=c]{1} \\
					&&&&\ctrl{1}&&&&\ctrl{1}&&&\gate{T}&\targ{}&\gate{T^\dagger}&\targ{}&&&&\gate[2]{\textit{CNOT}}&&\\
					&&&\gate{H}&\targ{}&\gate{T^\dagger}&\targ{}&\gate{T}&\targ{}&\gate{T^\dagger}&\targ{}&\gate{T}&\gate{H}&&&&&&&&
				\end{quantikz}}}
	\end{subfigure}
	\caption{A program and the corresponding circuit in Clifford+T.}
	\label{fig:prog}
\end{figure}

%% file: intro-table.tex
\begin{figure}[t]
	\centering
	\scalebox{0.8}{
		\begin{tabular}{|| m{7.3em} 
			| >{\centering\arraybackslash}m{2.3em} 
			| >{\centering\arraybackslash}m{3.4em} 
			| >{\centering\arraybackslash}m{4.2em} 
			| >{\centering\arraybackslash}m{2.7em} 
			| >{\centering\arraybackslash}m{3.6em} 
			| >{\centering\arraybackslash}m{5em} 
			| >{\centering\arraybackslash}m{4em} 
			| >{\centering\arraybackslash}m{3.2em} 
			| >{\centering\arraybackslash}m{4em} ||}
			\hline
			& 
			\textbf{Our work}
			&  
			\multicolumn{3}{c|}{\textbf{Q.}$\lambda$-\textbf{calculus}}
			&
			\textbf{QWIRE} 
			& 
			$ \lambda\textbf{-}\mathcal{S}_1$ 
			&
			\textbf{Qunity} 
			&
			\textbf{Sym} 
			&
			\textbf{Proto-Quipper-C} 
			\\[0.5ex]
			\cline{3-5}
			&&\raggedright 
			\cite{quantitative}
			&
			\cite{quantum-game,quantum-game-full-abstraction,marc-phd}
			&
			\cite{enriched-presheaf}
			& 
			\cite{paykin2017qwire}
			& 
			\cite{diazcaro2019unitary,diazcaro2022unitary,pablo2022unbounded,semimodules,lineal}
			&
			\cite{qunity}
			& 
			\raggedright\cite{sabry2018symmetric}
			& 
			\cite{quipper-rc}
			\\
			\hline \hline
			Quantum control & \cmark & \xmark & \xmark& \xmark & \xmark & \cmark & \cmark & \cmark & \cmark \\
			\hline
			Classical control & \cmark & \cmark & \cmark& \cmark & \cmark & \xmark & \cmark & \xmark & \cmark\\
			\hline
			\raggedright Infinite dimensional quantum type & \cmark & \cmark & \cmark & \cmark & \xmark & \xmark & \xmark &\cmark & \xmark \\
			\hline
			Unitary function & & & & & & & & & \\
			spaces & \cmark & \xmark & \xmark & \xmark & \cmark & \xmark & \xmark &\cmark & \cmark\\
			\hline
			Normalised \qquad  quantum terms & \cmark & \cmark & \cmark & \cmark & \cmark & \xmark & \xmark &\xmark & \cmark\\
			\hline
			Denotational & & & & & & & & & \\
			semantics & \cmark & \cmark & \cmark & \cmark & \cmark & \cmark & \cmark &\xmark & \cmark\\
			\hline
			Adequacy & \cmark & \cmark & \cmark & \cmark & \cmark& \cmark & \xmark & \xmark & \xmark\\
			\hline
		\end{tabular}
		}
		\caption{A comparison of some quantum programming languages.}
		\label{tab:comp}
\end{figure}

%% file: quantum-control.tex
This section mainly consists of work completed in one of the authors' PhD
thesis~\cite[Chapter 3]{louis-thesis}, with notations adapted to this paper. 

\subsection{The Language of Quantum Control}
\label{sec:pure-language}

We start by describing in details the language of (simply-typed) quantum
control. We first lay out the syntax~\secref{sub:simple-syntax}, then set forth
the associated typing rules~\secref{sub:pure-formation}, before detailing out
how substitution works within this language~\secref{sub:qua-substitution}.

\subsubsection{Syntax}
\label{sub:simple-syntax}

\begin{figure}
	\begin{alignat*}{10}
		&\text{({\ptypes})}\quad & \pQ_1, \pQ_2 &~~&&::= ~&\quad& \pbasic \alt
		\pQ_1 \pplus \pQ_2 \alt \pQ_1 \ptimes \pQ_2 \alt \pNat
		\\
		&\text{({\putypes})} & U &&& ::= && U(\pQ_1, \pQ_2) \\[1ex]
		&\text{(\pbterms)} & \pb &&&::=&& \ps * \alt \ps x \alt \inl{\ps b} \alt \inr{\pb} \alt
		\pb \ptimes \pb \alt \pzero \alt \psucc \pb \\
		&\text{(\pvalues)} & \ps v &&& ::= && \Sigma_{i \in I} (\palpha_i \pdot \pb_i) \\
		&\text{(\pexp)} & \ps e &&& ::= && \ps * \alt \ps x \alt \inl{\ps e} \alt \inr{\ps e} \alt
		\pair{\ps e}{\ps e} \alt \pzero \alt \psucc{\ps e} \\
		&&&&&&&  \alt \Sigma_{i \in I} (\palpha_i \pdot \ps e_i)
		\\
		&\text{(\punitaries)} & \isoterm &&&::=&& \unibasique \\ 
		&&&&&&& \alt \isoterm \oplus \isoterm  \alt \isoterm \otimes \isoterm \alt
		\isoterm \circ \isoterm \alt \isoterm\inv \alt \uctrl \isoterm \\
		&\text{(\pterms)} & \ps t &&& ::= && \ps * \alt \ps x \alt \inl{\ps t} \alt \inr{\ps t} \alt
		\pair{\ps t}{\ps t} \alt \pzero \alt \psucc t \\
		&&&&&&& \alt u~\ps t \alt \Sigma_{i \in I} (\palpha_i \pdot \ps t_i)
	\end{alignat*}
	\caption{Syntax of quantum control.}
	\label{fig:qu-syntax}
\end{figure}

The syntax of the programming language studied in this section is described in
a usual way, with grammars given in Figure~\ref{fig:qu-syntax}. The ground
types are given by a unit type $\pbasic$ and the usual constructors $\pplus$ and
$\ptimes$, which are respectively called \emph{direct sum} and \emph{tensor
product}. We also have the inductive type $\pNat$ of quantum natural numbers, as a witness that it is
possible to work with non finite data types, and thus non finite-dimensional
spaces in the model. We interpret every quantum type as a separable Hilbert
space: the unit type is interpreted as $\lrb \pbasic \defeq \mathbb C$; pair
types are interpreted as tensor products; (quantum) sum types are interpreted
as direct sums; finally, $\lrb \pnat \defeq \ell^2(\mathbb N)$.  The type of
qubits may be defined as $\pqbit \eqdef \pbasic \pplus \pbasic.$ We write
$\ptqbit{n}$ for the $n$-fold tensor product $\pqbit \ptimes \cdots \ptimes
\pqbit$. The Hilbert space $\ell^2(\mathbb N)$ allows us to form superpositions
with respect to a countable orthonormal basis (\emph{e.g.}, $\sfrac{1}{\sqrt 2} \ket 3
+ \sfrac{1}{\sqrt 2} \ket 7$). We may think of $\ell^2(\mathbb N)$ as a quantum
analogue of the natural numbers, to abstract from qubit-level computation.  We
equip functions, called \emph{unitaries} in this section, with a separate type,
written $\isotypeonetwo$ when $\pQ_1$ and $\pQ_2$ are two ground types. They
are not interpreted as Hilbert spaces, but as sets of unitary operators:
$\lrb{\isotypeonetwo} \defeq \{ u \colon \lrb \ptypeone \to \lrb \ptypetwo \ |\
u \emph{ is a unitary operator} \}$.

The terms of the language are given as follows:
\begin{itemize}
	\item a term $\ps *$ called the \emph{unit} corresponding to the unit type $\pbasic$;
	\item variables $\ps x, \ps y, \ps z,\dots$, given as elements of a set of variables
		$\mathbf{Var}$, assumed totally ordered;
	\item usual constructors for the direct sum, $\inl\!$ and $\inr\!$ which are
		respectively called the \emph{left injector} and \emph{right injector};
	\item a constructor corresponding to the tensor product, that is also
		written $\ptimes$;
	\item terms for natural numbers, $\pzero$ and the connective $\psucc\!$ that
		gives the successor of a term;
	\item the application of a unitary to a term, written $u~\ps t$ when
		$u$ is a unitary and $\ps t$ a term;
	\item finally, given a set of indices $\pbasic$, which is assumed totally
		ordered, a family of complex numbers $(\palpha_i)_{i \in I}$ and a
		family of terms $(\ps t_i)_{i \in I}$, one can form the term $\sum_{i
		\in I} (\palpha_i \pdot \ps t_i)$, representing the linear combination
		of the terms with complex scalars. This last term construction embodies
		the quantum effect of the language.
\end{itemize}
In quantum theory, a linear combination of vectors $\sum_{i \in I} \palpha_i
\ket{x_i}$ is normalized if $\sum_{i \in I} \abs{\palpha_i}^2 = 1$. The family
of real numbers $(\abs{\palpha_i}^2)_{i \in I}$ is then seen as a probability
distribution. This work does not focus on the probabilistic aspect of quantum
theory; however, we want to work with well-formed states, and thus normalized
states. This is why we ensure later that a linear combination of terms is
normalized. Throughout the section, a term $\Sigma_{i \in \set{1,2}} (\palpha_i
\pdot \ps t_i)$ might be written $\palpha_1 \pdot \ps t_1 + \palpha_2 \pdot \ps
t_2$, regarded as syntactic sugar. In some examples, a term $\Sigma_{i \in \set
*} 1 \pdot \ps t$ can be written $1 \pdot \ps t$ or even $\ps t$ for
readability; but note that $\Sigma_{i \in \set *} 1 \pdot \ps t$ and $\ps t$
are different terms in the syntax.

We call \emph{basis values}, the terms that are unitary-free -- in the sense
that they do not contain any function application -- and that do not involve
linear combinations. These terms are naturally classical (in the sense that
they are not quantum). Their name comes from the fact that we use them as a
syntactic representation of the canonical basis of a Hilbert space. Note that
basis values can be totally ordered.

Values, on the other hand, are linear combinations of basis values. In a value
$\sum_{i \in I} (\palpha_i \pdot \ps b_i)$, we assume that the family $(\ps
b_i)_{i \in I}$ is an increasing sequence, and that none of the scalars is
equal to $0$; this allows us to work with unique normal forms later in the
section. Since this definition is restrictive, we need to introduce a more
general piece of syntax, which still does not include unitary application,
called an \emph{expression}. They are used as outputs of functions, that we
introduce
below.

\begin{exa}
	\label{ex:basis-terms}
	The basis term $\ps \ast \colon \ps I$ represents the complex scalar $1 \in
	\mathbb C$. The basis terms, that represent the computational basis for
	qubits, can be defined by $\ps{\ket 0} \eqdef \inl{\ps \ast} \colon \pqbit$
	and $\ps{\ket 1} \eqdef \inr{\ps \ast} \colon \pqbit.$ Other basis terms
	that can be defined include $\pket{00} \eqdef \pket 0 \ptimes \pket 0
	\colon \ptqbit{2}$ and $\pket{11} \eqdef \pket 1 \ptimes \pket 1 \colon
	\ptqbit{2}$. If $\pb$ is a basis term, we often write it as $\pket{b}$ to
	distinguish it from other (non-basis) terms, \emph{e.g.}, $\pket{\psu{2}} \colon
	\pnat$ represents the qnat $\ket 2 \in \ell^2(\mathbb N).$ Single-qubit
	states may be defined by $\pket{\pm} \eqdef \left( \sfrac{1}{\sqrt 2} \pdot
	\pket{0} \pm \sfrac{1}{\sqrt 2} \pdot \pket{1} \right) \colon \pqbit$ %
	and (Entangled) quantum states can also be defined, \emph{e.g.}, $\textbf{Bell} \eqdef
	\left( \sfrac{1}{\sqrt 2} \pdot \pket{00} + \sfrac{1}{\sqrt 2} \pdot
	\pket{11} \right) \colon \ptqbit{2}$.  Linear combinations of qnats can
	also be written in the language, \emph{i.e.}, $\left( \sfrac{1}{\sqrt 2}
	\pdot \pket{\psu{\pzero}} + \sfrac{1}{\sqrt 6} \pdot \pket{\psu{1}} +
	\sfrac{1}{\sqrt 3} \pdot \pket{\psu{2}} \right) \colon \pnat .$
\end{exa}

Next, we describe the unitaries $\isoterm$ of Figure \ref{fig:qu-syntax}.
\emph{Unitary pattern-matching}, whose syntax is given by $\unibasique$, allows
us to build unitary maps out of basis terms and values: every \emph{closed}
basis term $\pb_i$ is mapped to the closed value $\ps v_i$; basis terms with
free variables $\pb_i$ determine a mapping with the use of substitution with
respect to the corresponding output expression $\ps e_i$, wherein both $\pb_i$ and $\ps
e_i$ have the same free variables occurring within them, and the induced
mapping is determined by performing all possible substitutions of the free
variables on both sides with the same closed basis terms. The formation
conditions (see \secref{sub:pure-formation}) ensure that the basis terms $\pb_i$
(resp. values $\ps v_i$) determine an OrthoNormal Basis (ONB) for the type $\pQ_1$ (resp.
$\pQ_2$) and therefore $\unibasique$ determines a unitary map. We introduce
syntactic sugar for quantum if statements, written $\textbf{qif}$.  If $u_1$ is
of the form $\unibasique$ and $u_2$ of the form $\unibasiqueppm$ then
$\qif{x}{u_2}{u_1}$ is short for the unitary:
\[
	\scalebox{1}{$
	\left\{
		\begin{array}{lcl}
			\ps{\ket 0} \ptimes \ps b_1 &\mapsto & \ps{\ket 0} \ptimes \ps e_1 \\
			& \vdots & \\
			\ps{\ket 0} \ptimes \ps b_n &\mapsto & \ps{\ket 0} \ptimes \ps e_n \\
		\end{array}
		\quad,\quad
	\right.
	\left.
		\begin{array}{lcl}
			\ps{\ket 1} \ptimes \ps b'_1 &\mapsto & \ps{\ket 1} \ptimes \ps e'_1 \\
			& \vdots & \\
			\ps{\ket 1} \ptimes \ps b'_m &\mapsto & \ps{\ket 1} \ptimes \ps e'_m
		\end{array}
	\right\}$}
\]

\begin{exa}
 \label{ex:unitaries}
	The unitaries $\phad$ and $\pcnot$ below encode the standard Hadamard and
	CNOT gates, respectively. $\pcnot$ makes use of some simple pattern-matching
	on its last line: it can match either with $\bintensor{\ketz}{\ketz}$ or
	with $\bintensor{\ketz}{\keto}$. The unitary $\phad_{\pnat} $ defines an
	extension of the Hadamard gate on the space $\ell^2(\mathbb N)$.

\noindent
	\begin{minipage}{0.4\textwidth}
		\begin{align*}
			\scalebox{1}{$\phad$}&\scalebox{1}{$\ \colon \isotype{\textbf{qbit}}{\textbf{qbit}}$}
			\\
			\scalebox{1}{$\phad$} &\scalebox{1}{$\defeq
			\left\{ \begin{array}{lcl}
				\mid \ketz & \iso & \sfrac{1}{\sqrt 2} \pdot\ketz + \sfrac{1}{\sqrt 2} \pdot \keto \\
				\mid \keto & \iso & \sfrac{1}{\sqrt 2} \pdot\ketz - \sfrac{1}{\sqrt 2} \pdot \keto
			\end{array} \right\} \qquad\qquad$}
		\end{align*}
	\end{minipage}
	\hspace*{-1cm}
	\begin{minipage}{0.4\textwidth}
		\begin{align*}
			\scalebox{1}{$\pcnot$} &\scalebox{1}{$\ \colon \isotype{\ptqbit{2}}{\ptqbit{2}}$}
			\\
			\scalebox{1}{$\pcnot$} &\scalebox{1}{$\defeq
			\left\{ \begin{array}{lcl}
				\mid\bintensor{\keto}{\ketz} & \iso & \bintensor{\keto}{\keto} \\
				\mid \bintensor{\keto}{\keto} & \iso & \bintensor{\keto}{\ketz} \\
				\mid\bintensor{\ketz}{\ket{\ps x}} & \iso & \bintensor{\ketz}{\ket{\ps x}}
			\end{array} \right\}$}
		\end{align*}
	\end{minipage}
	\begin{align*}
		\scalebox{1}{$\phad_{\pnat}$}  &\scalebox{1}{$\ \colon \isotype{\pnat}{\pnat}$} \\
		\scalebox{1}{$\phad_{\pnat}$} & \scalebox{1}{$\defeq
		\left\{ \begin{array}{lcl}
			\mid \ket{\psu{\pzero}} & \iso & \sfrac{1}{\sqrt 2} \pdot\ket{\psu{\pzero}} + \sfrac{1}{\sqrt 2}\pdot \ket{\psu{1}} \\
			\mid \ket{\psu{1}} & \iso & \sfrac{1}{\sqrt 2} \pdot\ket{\psu{0}} - \sfrac{1}{\sqrt 2} \pdot\ket{\psu{1}} \\
			\mid \ket{\psu{$\ps x$+2}} & \iso & \ket{\psu{$\ps x$+2}}
		\end{array} \right\}$}
	\end{align*}

\end{exa}

The remaining expressions for forming unitaries are easy to understand:
$\isoterm_1 \oplus \isoterm_2$ and
$\isoterm_1 \otimes \isoterm_2$ represent the direct sum and the tensor product of unitaries,
$\isoterm_2 \circ \isoterm_1$ represents the composition of unitaries, $\isoterm^*$
represents the adjoint of a unitary, and, finally, $\mathrm{ctrl}\ \isoterm$ represents a qubit-controlled
unitary operator. For simplicity, we sometimes write $\unibasiqueshort$ as a
shorthand notation for the unitary $\unibasique$, when $n$ is clear from
context or unimportant. We also point out that unitaries depend only on values
and basis terms and that general pure terms cannot be used for the construction
of unitaries (in particular, a unitary does not depend on other unitaries).

\subsubsection{Types and Formation Rules}
\label{sub:pure-formation}

As usual, a \emph{typing context} consists of a set of pairs of a variable and a
type, written $\pGamma$ and generated over $\pGamma ::= \emptyset \alt \set{\ps
x \colon \pQ} \cup \pGamma$. A comma between two contexts represents the union
of the contexts, \emph{i.e.}, $\pGamma, \pGamma' \triangleq \pGamma \cup \pGamma'$.  We
have two levels of judgments: the one for terms, where sequents are written
$\pGamma \vdash \ps t \colon \pQ$, and a typing judgment for unitaries, that
we write $\entailiso u \colon \isotypeonetwo$. Variables in a context $\pGamma$
are strictly linear: given $\pGamma \entail \ps t \colon \pQ$, every element
of $\pGamma$ has to occur \emph{exactly once} in the term $\ps t$.

Values (resp. expressions, terms) in the language are, in general, normalized
linear combinations of basis values (resp. expressions, terms). However, not
all basis values (resp. expressions, terms) can be summed in a normalized
combination, as in the following example.

\begin{exa}
	$\frac{1}{\sqrt 2} \pdot (\inl{\ps *}) + \frac{1}{\sqrt 2} \pdot (\inl{\ps
	*})$ should not be considered a quantum state: its norm is $\sqrt 2$, not
	$1$. Its interpretation in Hilbert spaces is the following: $\frac{1}{\sqrt
	2} \ket 0 + \frac{1}{\sqrt 2} \ket 0 = \frac{2}{\sqrt 2} \ket 0 = \sqrt 2
	\ket 0$. This is not a normalized state.
\end{exa}

The next definition introduces the predicate of orthogonality, written $\bot$.
Given $\ps t_1 ~\bot~\ps t_2$, we are allowed to form a normalized linear combination
of the two terms.

\begin{defi}[Orthogonality]
	\label{def:orthogonality}
	We introduce a symmetric binary relation $\bot$ on terms. Given two terms
	$\ps t_1, \ps t_2$, the predicate $\ps t_1~\bot~\ps t_2$ holds if it can be derived
	inductively with the rules below; we say that $\ps t_1$ and $\ps t_2$ are
	orthogonal. The relation $\bot$ is defined as the smallest symmetric
	relation such that:
 \[
  \begin{array}{c}
   \infer{\inl{\ps t_1}~\bot~\inr{\ps t_2}}{}
   \qquad
   \infer{\inl{\ps t_1}~\bot~\inl{\ps t_2}}{\ps t_1~\bot~\ps t_2}
   \qquad
   \infer{\inr{\ps t_1}~\bot~\inr{\ps t_2}}{\ps t_1~\bot~\ps t_2}
    \\[1.5ex]
   \infer{\pair{\ps t}{\ps t_1}~\bot~\pair{\ps t'}{\ps t_2}}{\ps t_1~\bot~\ps t_2}
   \qquad
   \infer{\pair{\ps t_1}{\ps t}~\bot~\pair{\ps t_2}{\ps t'}}{\ps t_1~\bot~\ps t_2}
   \qquad
   \infer{\pzero~\bot~\psucc t}{}
   \qquad
   \infer{\psucc{\ps t_1}~\bot~\psucc{\ps t_2}}{\ps t_1~\bot~\ps t_2}
   \\[1.5ex]
   \infer{u~\ps t_1~\bot~u~\ps t_2}{\ps t_1~\bot~\ps t_2}
   \qquad
   \infer{
				\ps t~\bot~\Sigma_{i \in I} (\palpha_i \pdot \ps t_i)
 			}{
 				\forall i \in I, \ps t~\bot~\ps t_i
 			}
   \qquad
 			\infer{
				\ps t~\bot~\Sigma_{i \in I \cup \set *} (\palpha_i \pdot \ps t_i)
 			}{
 				\forall i \in I, \ps t~\bot~\ps t_i
				&
				\palpha_* = 0
				&
				\ps t_* = \ps t
 			}
  \\[1.5ex]
			\infer{\Sigma_{j \in J} (\palpha_j \pdot \ps t_j)~\bot~\Sigma_{k \in K}
			(\pbeta_k \pdot \ps t_k)}
				{\forall i\neq j \in I, \ps t_i~\bot~\ps t_j
				&
				J,K \subseteq I
				&
				\sum_{i \in J \cap K} \bar{\palpha_i}\pbeta_i = 0
				}
  \end{array}
 \]
\end{defi}

\begin{rem}[Orthogonality of variables]
	Given two variables $\ps x$ and $\ps y$, the terms $\ps x$ and $\ps y$ are
	not orthogonal.  The main reason is that they could be instantiated with
	the same value. On the other hand, $\inl{\ps x}$ and $\inr{\ps y}$ are
	orthogonal, for example.
\end{rem}

\begin{rem}
	Consider the following unitary:
	\[
		u \triangleq \sset{\begin{array}{lcl}
			\mid \inl{} \ps x & \iso & \inr{} \ps x \\
			\mid \inr{} \ps y & \iso & \inl{} \ps y
		\end{array}
		}
	\]
	which swaps the two sides of a direct sum $\pplus$. Applying $u$ to $\inl{\ps *}$ for example, we get $\inr{\ps *}$, which
	is orthogonal to $\inl{\ps *}$.  However, we cannot derive that $u(\inl{\ps
	*})$ and $\inl{\ps *}$ are orthogonal with the rules given above. This is
	because we wish to be able to derive orthogonal \emph{statically}.
\end{rem}

Note that orthogonality holds without any formation rule or notion of type.
Figure~\ref{fig:typing-values-simple} introduces the formation rules for
expressions, and therefore for basis values and values, thanks to the notion of
orthogonality.

\begin{figure}
	\[\begin{array}{c}
		\infer{
			\emptyset \entail \ps * \colon \pbasic,
		}
		{}
		\qquad
		\infer{
			\ps x \colon \pQ \entail \ps x \colon \pQ,
		}{}
	\qquad
		\infer{
			\pGamma\entail \inl{\ps e} \colon \pQ_1 \pplus \pQ_2,
		}{
			\pGamma\entail \ps e \colon \pQ_1
		}
		\qquad
		\infer{
			\pGamma\entail \inr{\ps e} \colon \pQ_1 \pplus \pQ_2,
		}{
			\pGamma\entail \ps e \colon \pQ_2
		}
		\\[1.5ex]
			\qquad
		\infer{
			\pGamma_1,\pGamma_2\entail \pair{\ps e_1}{\ps e_2} \colon \pQ_1 \ptimes \pQ_2,
		}{
			\pGamma_1\entail \ps e_1 \colon \pQ_1
			&
			\pGamma_2\entail \ps e_2 \colon \pQ_2
		}
		\qquad
		\infer{\entail \pzero \colon \pNat,}{}
		\qquad
		\infer{
			\pGamma \entail \psucc{\ps e} \colon \pNat,
		}{
			\pGamma \entail \ps e \colon \pNat
		}
		\\[1.5ex]
		\infer{
			\pGamma \entail \Sigma_i (\palpha_i \pdot \ps e_i) \colon \pQ.
		}{
			\pGamma \entail \ps e_i \colon \pQ
			&
			\Sigma_i \abs{\palpha_i}^2 = 1
			&
			\forall i\not= j, \ps e_i~\bot~\ps e_j
		}
	\end{array}
	\]
	\caption{Formation rules of (basis) values and expressions.}
	\label{fig:typing-values-simple}
\end{figure}

Once we know how basis values are formed in a certain type, we can discuss some
orthogonality properties among specific types. In linear algebra, given a
vector space and an orthogonal basis of that space, if two elements of the
basis are not orthogonal, then they are equal. Our case is more subtle,
because, for example, $\inr{} (\inl{\ps *})$ and $\inr{\ps x}$ are not equal,
but also not orthogonal. We will see, later in this section, that they are
linked by substitution.

\begin{lem}
	\label{lem:typed-orthogonality}
	Given three well-formed basis values $\pGamma_1 \entail \ps b_1 \colon \pQ_1$,
	$\pGamma_2 \entail \ps b_2 \colon \pQ_1$, and $\pGamma_3 \entail \ps b_3
	\colon \pQ_1$, if $\ps b_1~\bot~\ps b_2$ and $\neg(\ps b_2~\bot~\ps b_3)$,
	then $\ps b_1~\bot~\ps b_3$.
\end{lem}

Of course, this lemma holds for basis values because they belong to a common
basis, that can be seen as the \emph{canonical basis}. Nothing can be said
about the orthogonality of values that do not belong to a common basis.

\begin{exa}
	The values $\inl{\ps *}$ and $\inr{\ps *}$ are orthogonal, but of course,
	they are both not orthogonal to $\frac{1}{\sqrt 2} \pdot (\inl{\ps *}) -
	\frac{1}{\sqrt 2} \pdot (\inr{\ps *})$.
\end{exa}

This motivates a syntactic definition for a basis, containing values of the
language. Given a set of orthogonal values, we want to ensure that this set
parses the whole type, so that it can be seen as an orthonormal basis. This is
given by the notion of \emph{orthogonal decomposition}. First, given a set $S =
\set{\pair{\ps b_1}{\ps b_1'}, \dots, \pair{\ps b_n}{\ps b_m'}}$, we define
$\pi_1(S) \triangleq \set{\ps b_1, \dots, \ps b_n}$ and $\pi_2(S) \triangleq \set{\ps b_1',
\dots, \ps b_m'}$. Finally, we define $S_\pb^1$ and $S_\pb^2$ respectively as
$\set{\ps b' \alt \pair{\ps b}{\ps b'} \in S}$ and $\set{\ps b' \alt \pair{\ps
b'}{\ps b} \in S}$.

\begin{defi}[Orthonormal Basis]
	\label{def:OD}
	We introduce a predicate $\OD{\pQ}{(-)}$ on finite sets of basis values. Given
	a finite set of values $S$, 
	$\OD{\pQ}{(S)}$ holds if it can be derived with the
	following rules. 
	\[
  \begin{array}{c}
		\infer{\OD{\pbasic}{(\set{\ps *})}}{} 
			\qquad
	      \infer{\OD{\pQ}{(\set{\ps x})}}{} 
			\qquad 
  	\infer{
				\OD{\pQ_1 \pplus \pQ_2}{(\{\inl{\ps b} \alt \ps b \in S \} \cup
				\{\inr{\ps b} \alt \ps b\in T \})}
			}{
				\OD{\pQ_1}{(S)} \qquad \OD{\pQ_2}{(T)}
			} 
  	\mynl
  	\infer{
				\OD{\pQ_1\otimes \pQ_2}{(S)} 
			}{
				\begin{array}{c}
					\OD{\pQ_k}{(\pi_k(S))} \qquad  \forall
					\ps b \in \pi_k(S), \OD{\pQ_{3-k}}{(S^k_\pb)}
				\end{array}
				\quad k\in \{1,2\}
			}			
	\mynl
	\infer{
				\OD{\pNat}{(\set{\pzero} \cup \set{\psucc{\ps v} \alt \ps v \in S})}
			}{
				\OD{\pNat}{(S)}
			}
		\end{array}
	\]
	The existence of this predicate is ensured by induction on the cardinality of $S$.
\end{defi}

To simplify the notations in later proofs, we will write $S \boxplus T$ for the
set $\{\inl{\ps b} \alt \ps b \in S \} \cup \{\inr{\ps b} \alt \ps b \in T \}$,
and $S^{\oplus 0}$ for the set $\set{\pzero} \cup \set{\psucc{\ps v} \alt \ps v
\in S}$.
We call ${\rm ONB}$ the predicate introduced above in general, without specifying the 
type, to facilitate the later discussions.

\begin{rem}
	Note that the precondition to derive $\OD{\pQ_1 \ptimes \pQ_2}{(S)}$ cannot
	be simplified: there are sets $S$ such that $\OD{\pQ_1}{(\pi_1(S))}$ and
	for all $\ps b \in \pi_1(S)$, $\OD{\pQ_2}{(S_\pb^1)}$ but not all $\ps b \in
	\pi_2(S)$ is such that $\OD{\pQ_2}{(S^2_\pb)}$, \emph{e.g.}, $S =
	\set{(\inl{\ps *}) \ptimes \ps y, (\inr{\ps x}) \ptimes (\inl{\ps *}),
	(\inr{\ps x}) \ptimes (\inr{\ps *})}$ with the type $(\pbasic \pplus
	(\pbasic \pplus \pbasic)) \ptimes (\pbasic \ptimes \pbasic)$.
\end{rem}

The predicate ${\rm ONB}_{\pQ}$ defined above ensures that a finite set $S$ of
values represents an orthonormal basis of $\pQ$, that we can view as the
canonical basis. Note that even for $\pNat$, the set representing the basis is
finite, \emph{e.g.}, $\set{\pzero, \psucc \pzero, \psucc (\psucc x)}$. With knowledge
of linear algebra, one can say that the bases secured by ${\rm ONB}$ are not all
the possible orthonormal bases. A change of basis with a unitary matrix also
provides an orthonormal basis. This is the purpose of the next definition.

\begin{defi}
	\label{def:od-ext}
	We extend the previous definition to general values. The statement
	$\ODe{\pQ}{(-)}$ is a predicate on finite sets of values, and
	$\ODe{\pQ}{(S)}$ holds if it can be derived from the following rules.
	\[ \begin{array}{c}
		\infer{\ODe{I}{(\set{\ps *})}}{} 
		\qquad
			\infer{\ODe{\pQ_1}{(\set{\ps x})}}{} 
		\qquad 
		\infer{
			\ODe{\pQ_1\oplus \pQ_2}{(\{\inl{\ps e} \alt \ps e \in S \} \cup
			\{\inr{\ps e} \alt \ps e \in T \})}
		}{
			\ODe{\pQ_1}{(S)} \qquad \ODe{\pQ_2}{(T)}
		} 
		\mynl
		\infer{
			\ODe{\pQ_1 \ptimes \pQ_2}{(S)} 
		}{
			\begin{array}{c}
				\ODe{\pQ_k}{(\pi_k(S))} \qquad \forall
				\ps e \in \pi_k(S), \ODe{\pQ_{3-k}}{(S^k_{\ps e})}
			\end{array}
			\quad k\in\{1,2\}
		}		
		\mynl
		\infer{
			\ODe{\pNat}{(\set{\pzero} \cup \set{\psucc{\ps v} \alt \ps v \in S})}
		}{
			\ODe{\pNat}{(S)}
		}
		\qquad
		\infer{
			\ODe{\pQ_1}{(\set{\Sigma_{\ps e'\in S} (\palpha_{\ps e, \ps e'}
			\pdot \ps e') \alt \ps e \in S})}
		}{
			\ODe{\pQ_1}{(S)}
			& (\palpha_{\ps e, \ps e'})_{(\ps e,\ps e') \in S \times S}\text{ is a unitary matrix} 
		}
	\end{array} \]
	Given $\ODe{\pQ}{(S)}$, we say that $S$ is an \emph{orthogonal decomposition} of type $\pQ$.
\end{defi}

To simplify the notations, we will write $S^\alpha$ the set
$\set{\Sigma_{\ps b'\in S} (\palpha_{\ps b, \ps b'} \pdot \ps b') \alt \ps b \in S}$.

\begin{exa}
 \label{ex:ONBs}
	For $\pqbit $, the rules from Definition~\ref{def:OD} and Definition~\ref{def:od-ext} show that
	$\OD{\pqbit}{(\{ \pket{0}, \pket{1} \})}$ and 
	$\ODe{\pqbit}{(\{\pket{+}, \pket{-} \})}$ hold (with $\pket{\pm}$ defined
	in Example \ref{ex:basis-terms}). We also have that $\OD{\pqbit}{(\{ \ps x
	\} )}$, where $\ps x$ is a variable. In other words, we determine an ONB
	for the type $\pqbit$ by substituting $\ps x$ with all possible closed
	computational \emph{basis} terms of type $\pqbit$ (see Figure
	\ref{fig:qu-syntax}). For $\pnat$, any set $S$ with the property
	that $\OD{\pnat}{(S)}$ holds, must contain a term that is not closed. We
	have that $\OD{\pnat}{(\{ \pket{\psu x}\} )}$, $\OD{\pnat}{(\{
		\pket{\psu{\pzero}}, \pket{\psu{$\ps x$+1}}\}) }$, and $\OD{\pnat}{(\{
		\pket{\psu{\pzero}}, \pket{\psu{1}}, \pket{\psu{$\ps x$+2}} \} )}$ hold
	true and determine the same ONB for the type $\pnat$ (it represents the
	computational basis of $\ell^2(\mathbb N)$). The set used in the last
	predicate allows us to conclude that
	$
  \ODval{\pnat}{\left(\left\{
   \sfrac{1}{\sqrt 2} \psdot\pket{\psu{\pzero}} + \sfrac{1}{\sqrt 2}\psdot \pket{\psu{1}},
   \sfrac{1}{\sqrt 2} \psdot\pket{\psu{\pzero}} - \sfrac{1}{\sqrt 2} \psdot\pket{\psu{1}},
			\pket{\psu{$\ps x$+2}}
  \right\}\right)}
 $
	holds. This last set of values is used to type $\phad_{\pnat}$ from Example
	\ref{ex:unitaries}.
\end{exa}

We define the predicate ${\rm ONB}$ above without the help of the orthogonality
predicate $\bot$, but they are not unrelated: the elements of an orthogonal
decomposition, are in particular pairwise orthogonal.

\begin{lem}[$\mathrm{ONB}$ implies $\bot$]
	\label{lem:od-imp-bot}
	If $\ODe{\pQ}{(S)}$ holds, then for all $\ps t_1 \neq \ps t_2 \in S$, $\ps
	t_1~\bot~\ps t_2$ holds.
\end{lem}
\begin{proof}
	The proof is done by induction on ${\rm ONB}$.
	\begin{itemize}
		\item $\ODe{\pbasic}{(\set{\ps *})}$. There is no pair of different terms in $\set{\ps *}$.
		\item $\ODe{\pQ}{(\set{\ps x})}$. There is no pair of different terms in $\set{\ps x}$.
		\item $\ODe{\pQ_1 \pplus \pQ_2}{(S \boxplus T)}$. There are several cases: either both
			terms are of the form $\inl -$, namely $\ps t_1 = \inl{\ps t'_1}$ and $\ps t_2 = \inl{\ps
			t'_2}$, with $\ps t'_1$ and $\ps t'_2$ in $S$, in which case the induction
			hypothesis on $\ODe{\pQ_1}{(S)}$ gives that $\ps t'_1~\bot~\ps t'_2$, and thus $\inl{\ps t'_1}~\bot~\inl{\ps t'_2}$; or both are of the $\inr -$, the case is similar
			with the induction hypothesis on $\ODe{\pQ_2}{(T)}$; or $\ps t_1 = \inl{\ps t'_1}$ and $\ps t_2
			= \inr{\ps t'_2}$, then we have directly $\inl{\ps t'_1}~\bot~\inr{\ps t'_2}$.
		\item $\ODe{\pQ_1 \ptimes \pQ_2}{(S)}$. Both cases are similar. Suppose
			that $\ODe{\pQ_1}{(\pi_1(S))}$. We know that $\ps t_1 = \ps t'_1
			\ptimes \ps t''_1$ and $\ps t'_2 \ptimes \ps t''_2$. The induction
			hypothesis on $\ODe{\pQ_1}{(\pi_1(S))}$ gives that $\ps t'_1~\bot~\ps
			t'_2$ and thus $\ps t'_1 \ptimes \ps t''_1~\bot~\ps t'_2 \ptimes
			\ps t''_2$.
		\item $\ODe{\pNat}{(S^{\oplus 0})}$. If one of $\ps t_1$ or $\ps t_2$ is
			$\pzero$, the conclusion is direct; else, $\ps t_1 = \psucc{\ps
			t'_1}$ and $\ps t_2 = \psucc{\ps t'_2}$ and the induction
			hypothesis on $\ODe{\pNat}{(S)}$ gives that $\ps t'_1~\bot~\ps t'_2$
			and thus $\psucc{\ps t'_1}~\bot~\psucc{\ps t'_2}$.
		\item $\ODe{\pQ_1}{(S^\alpha)}$. By induction hypothesis, all the terms in $S$
			are pairwise orthogonal; and the matrix $\palpha$ is unitary, which means
			that the inner product of its columns is zero when the columns are
			different, which ensures that two different terms in $S^\alpha$ are
			orthogonal. \qedhere
	\end{itemize}
\end{proof}

Note that while orthogonality ensures non-overlapping, it does not
ensure exhaustivity, only ${\rm ONB}_\pQ$ does. The next lemma details
the exhaustivity of ${\rm ONB}_\pQ$. It is a consequence of a later result,
that uses the notion of substitution (see \secref{sub:qua-substitution}
and Lemma~\ref{lem:od-output-exhaustive}).

\begin{prop}
	\label{lem:od-exhaustive}
	If $\ODe{\pQ}{(S)}$ and $\pGamma \vdash \ps e \colon \pQ$ both hold, then there
	exists $\ps e' \in S$ such that $\neg(\ps e~\bot~\ps e')$ holds.
\end{prop}

This notion of orthogonal decomposition allows us to introduce \emph{unitary}
abstractions in our syntax. A basic unitary has the form $\unibasique$ and is
well-formed only if $(\ps b_i)_i$ forms a basis and $(\ps v_i)_i$ also forms a
basis; we then retrieve the intuition that a unitary should be equivalent to a
change of orthonormal basis. The typing rules for deriving unitaries and
unitary operations are detailed in Figure~\ref{fig:typisos}.

\begin{figure}
	\[
		\begin{array}{c}
			\infer{
				\entailiso \unibasique \colon \isotypeonetwo,
			}{
				\begin{array}{@{}l@{}}
					\pGamma_i \entail \ps b_i \colon \pQ_1 \\
					\pGamma_i \entail \ps e_i \colon \pQ_2
				\end{array}
				&
				\begin{array}{@{}l@{}}
					\OD{\pQ_1}{(\{\ps b_1, \dots, \ps b_n\})} \\
					\ODe{\pQ_2}{(\{\ps e_1, \dots, \ps e_n\})}
				\end{array}
			}
			\qquad
			\infer{
				\entailiso u_1 \oplus u_2 \colon \isotype{\pQ_1 \pplus \pQ_3}{\pQ_2 \pplus \pQ_4}, 
			}{
				\entailiso u_1 \colon \isotype{\pQ_1}{\pQ_2} 
				&
				\entailiso u_2 \colon \isotype{\pQ_3}{\pQ_4} 
			}	
			\mynl\\
			\infer{
				\entailiso u_1 \otimes u_2 \colon \isotype{\pQ_1 \ptimes \pQ_3}{\pQ_2 \ptimes \pQ_4}, 
			}{
				\entailiso u_1 \colon \isotype{\pQ_1}{\pQ_2} 
				&
				\entailiso u_2 \colon \isotype{\pQ_3}{\pQ_4} 
			}
			\qquad
			\infer{
				\entailiso u_2 \circ u_1 \colon \isotype{\pQ_1}{\pQ_3},
			}{
				\entailiso u_1 \colon \isotype{\pQ_1}{\pQ_2}
				&
				\entailiso u_2 \colon \isotype{\pQ_2}{\pQ_3}
			}
			\mynl\\
			\infer{
				\entailiso u\inv \colon \isotypetwoone,
			}{
				\entailiso u \colon \isotypeonetwo
			}		
			\qquad
			\infer{
				\entailiso \uctrl \isoterm \colon \isotype{
				\pqbit \ptimes \pQ}{\pqbit \ptimes \pQ}.
			}{
				\entailiso \isoterm \colon \isotype{\pQ}{\pQ}
			}
		\end{array}
	\]
	\caption{Formation rules of unitaries.}
	\label{fig:typisos}
\end{figure}

Terms in our syntax are either expressions or an application of a unitary to a
term, in a similar style to the $\lambda$-calculus; however, abstractions are
considered separately in the grammar. The formation rules are the same as the
one for values given in Figure~\ref{fig:typing-values-simple}, with the
addition of a rule that enables the application of a unitary to a term. The
details are in Figure~\ref{fig:typterms}. 

\begin{figure}
	\[\begin{array}{c}
		\infer{
			\emptyset \entail \ps * \colon \pbasic,
		}
		{}
		\qquad
		\infer{
			\ps x \colon \pQ \entail \ps x \colon \pQ,
		}{}
		\qquad
			\infer{
			\pGamma\entail \inl{\ps t} \colon \pQ_1 \pplus \pQ_2,
		}{
			\pGamma\entail \ps t \colon \pQ_1
		}
		\qquad
		\infer{
			\pGamma\entail \inr{\ps t} \colon \pQ_1 \pplus \pQ_2,
		}{
			\pGamma\entail \ps t \colon \pQ_2
		}
		\\[2ex]
		\infer{
			\pGamma_1,\pGamma_2 \entail \pair{\ps t_1}{\ps t_2} \colon \pQ_1 \ptimes \pQ_2,
		}{
			\pGamma_1 \entail \ps t_1 \colon \pQ_1
			&
			\pGamma_2 \entail \ps t_2 \colon \pQ_2
		}
\qquad
		\infer{\vdash \pzero \colon \pNat,}{}
		\qquad
		\infer{
			\pGamma \vdash \psucc{\ps t} \colon \pNat,
		}{
			\pGamma \vdash \ps t \colon \pNat
		}
		\\[2ex]
				\infer{
			\pGamma \entail \isoterm~\ps t \colon \pQ_2.
		}{
			\entailiso \isoterm \colon \isotype{\pQ_1}{\pQ_2}
			&
			\pGamma \entail \ps t \colon \pQ_1
		}
		\qquad 
		\infer{
			\pGamma \entail \Sigma_i (\palpha_i \pdot \ps t_i) \colon \pQ,
		}{
			\pGamma \entail \ps t_i \colon \pQ
			&
			\sum_i \abs{\palpha_i}^2 = 1
			&
			\forall i\not= j, \ps t_i~\bot~\ps t_j
		}
	\end{array}
	\]
	\caption{Formation rules of terms.}
	\label{fig:typterms}	
\end{figure}

The application of a unitary to a term is what carries the computational power
of the language. In the $\lambda$-calculus, the $\beta$-reduction reduces an
application to a term where a substitution is performed. A similar mechanism is
at play in this syntax; however, the left-hand side in a unitary abstraction
can contain several variables. Hence the introduction of \emph{valuations},
which we use to perform substitution and to define our equivalent of
$\beta$-reduction.

\subsubsection{Valuations and Substitution}
\label{sub:qua-substitution}

We recall the formalization proposed in~\cite{sabry2018symmetric}, with the
notion of valuation: a partial map from a finite set of variables (the support)
to a set of values. Given two basis values $\ps b$ and $\ps b'$, we build the
smallest valuation $\sigma$ such that the patterns of $\ps b$ and $\ps b'$
match and that the application of the substitution to $\ps b$, written
$\sigma(\ps b)$, is equal to $\ps b'$. We denote the matching of a basis value
$\ps b'$ against a pattern $\ps b$ and its associated valuation $\sigma$ as
$\match{\sigma}{\ps b}{\ps b'}$. Thus, $\match{\sigma}{\ps b}{\ps b'}$ means
that $\ps b'$ matches with $\ps b$ and gives a smallest valuation $\sigma$,
while $\sigma(\ps b)$ is the substitution performed. The predicate
$\match{\sigma}{\ps b}{\ps b'}$ it is defined as follows, with $\ini{\ps b}$
being either $\inl{\ps b}$ or $\inr{\ps b}$.
\[
	\begin{array}{c}
		\infer{\match{\sigma}{\ps *}{\ps *}}{}
		\qquad
		\infer{\match{\sigma}{\ps x}{\ps b'}}{\sigma = \{ \ps x \mapsto \ps b'\}}
		\qquad
		\infer{\match{\sigma}{\ini{\ps b}}{\ini{\ps b'}}}{\match{\sigma}{\ps b}{\ps b'}}
		\mynl
		\infer{
			\match{\sigma}{\ps b_1 \ptimes \ps b_2}{\ps b'_1 \ptimes \ps b'_2}
		}{
			\match{\sigma}{\ps b_1}{\ps b'_1}
			&
			\match{\sigma}{\ps b_2}{\ps b'_2}
			&
			\text{supp}(\sigma_1) \cap \text{supp}(\sigma_2) = \emptyset
			&
			\sigma = \sigma_1\cup\sigma_2
		}
		\mynl
		\infer{\match{\sigma}{\pzero}{\pzero}}{}
		\qquad
		\infer{\match{\sigma}{\psucc{\ps b}}{\psucc{\ps b'}}}{\match{\sigma}{\ps b}{\ps b'}}
	\end{array}
\]
Besides basis values, we authorize valuations to replace variables with any
expression, \emph{e.g.}, $\set{\ps x \mapsto \ps e}$. Whenever $\sigma$ is a
valuation whose support contains the variables of $\ps t$, we write $\sigma(\ps
t)$ for the value where the variables of $\ps t$ have been replaced with the
corresponding terms in $\sigma$, as follows: 
\begin{align*}
\sigma(\ps *) &\triangleq \ps *,\\
\sigma(\ps x) &\triangleq \ps e \text{ if }\{\ps x\mapsto \ps e\}\subseteq \sigma,\\ 
\sigma(\inl{\ps t}) &\triangleq \inl{\sigma(\ps t)},\\
\sigma(\inr{\ps t}) &\triangleq \inr{\sigma(\ps t)},\\ 
\sigma(\pair{\ps t_1}{\ps t_2}) &\triangleq \pair{\sigma(\ps t_1)}{\sigma(\ps t_2)},\\ 
\sigma(\psucc{\ps t}) &\triangleq \psucc \sigma(\ps t),\\
\sigma(u~\ps t) &\triangleq u~\sigma(\ps t),\\
\sigma(\Sigma_i (\palpha_i \pdot \ps t_i)) &\triangleq \Sigma_i (\palpha_i
		\pdot \sigma(\ps t_i)).
\end{align*}

\begin{rem}
	If $\match{\sigma}{\ps b}{\ps b'}$, then $\sigma(\ps b) = \ps b'$.
\end{rem}

We can now show the soundness of orthogonality with regard to pattern-matching:
in other words, orthogonality is stable by substitution, and thus the previous
remark ensures there cannot be any match between two basis values if they are
orthogonal.

\begin{lem}
	\label{lem:ortho-subst-equiv}
	Given two terms $\ps t_1$ and $\ps t_2$, if $\ps t_1~\bot~\ps t_2$ holds, then,
	for all valuations $\sigma_1$ and $\sigma_2$, $\sigma_1(\ps
	t_1)~\bot~\sigma_2(\ps t_2)$.
\end{lem}
\begin{proof}
	Observe that $\sigma_1(\inl{\ps t_1}) = \inl{\sigma_1(\ps t_1)}$ and
	$\sigma_2(\inr{\ps t_2}) = \inr{\sigma_2(\ps t_2)}$ and thus, whatever the
	valuations are, those two terms are orthogonal. The rest of the proof
	falls directly by induction on the definition of $\bot$.
\end{proof}

If one the basis values is closed, we observe that there is an equivalence
between matching the pattern and not being orthogonal; which also implies that
different patterns are orthogonal.

\begin{prop}
	\label{prop:patterns-are-orthogonal}
	Given two well-formed basis values $\pGamma \entail b \colon \pQ$ and $~\entail
	b' \colon \pQ$, there exists a valuation $\sigma$ such that
	$\match{\sigma}{\ps b}{\ps b'}$ iff $\neg(\ps b~\bot~\ps b')$.
\end{prop}
\begin{proof}
	This is proven by a direct induction on $\pGamma \entail \ps b \colon \pQ$.
\end{proof}

The next two lemmas provide a strong link between orthogonal decompositions and
substitutions.

\begin{lem}[Exhaustivity and non-overlapping]
	\label{lem:od-substitution}
	Assume $\OD{\pQ}{(S)}$ holds. Then for all closed basis values $\entail \ps b' \colon
	\pQ$, there exist a unique $\ps b \in S$ and a unique $\sigma$ such that
	$\match{\sigma}{\ps b}{\ps b'}$ holds.
\end{lem}
\begin{proof}
	This is proven by induction on the derivation of $\OD{\pQ}{(S)}$.
	\begin{itemize}
		\item If $\OD{\pbasic}{(\set{\ps *})}$, we have $\ps b' = \ps *$ and
			there is nothing to do.
		\item If $\OD{\pQ}{(\set{\ps x})}$. There is only $\ps x$ in $S$ and
			$\set{\ps x \mapsto \ps b'}$ is the only possible substitution.
		\item If $\OD{\pQ_1 \pplus \pQ_2}{(S \boxplus T)}$, there are two cases:
			\begin{itemize}
				\item either $\ps b' = \inl{\ps b'_A}$, in which case the
					induction hypothesis gives a unique $\ps b_A \in S$ and a
					unique $\sigma$ such that $\match{\sigma}{\ps b_A}{\ps
					b'_A}$, and thus $\match{\sigma}{\inl{\ps b_A}}{\ps b}$ in
					a unique way,
				\item or $\ps b' = \inr{\ps b'_B}$, and a similar argument
					gives a unique match $\match{\sigma}{\inr{\ps b_B}}{\ps
					b'}$.
			\end{itemize}
		\item If $\OD{\pQ_1 \ptimes \pQ_2}{(S)}$, $\ps b' = \ps b'_A \ptimes \ps
			b'_B$, in both cases to derive ${\rm ONB}$, we get unique $\ps b_A$,
			$\ps b_B$, $\sigma_A$ and $\sigma_B$ such that $\match{\sigma_A
			\cup \sigma_B}{\ps b_A \ptimes \ps b_B}{\ps b'}$.
		\item If $\OD{\pNat}{(S^{\oplus 0})}$, there are two cases: either $\ps
			b' = \pzero$, in which case there is nothing to do, or $\ps b' =
			\psucc{\ps b''}$, and the induction hypothesis gives a unique $\ps
			b$ and a unique $\sigma$ such that $\match{\sigma}{\ps b}{\ps b''}$
			and thus $\match{\sigma}{\psucc{\ps b}}{\ps b'}$.\qedhere
	\end{itemize}
\end{proof}

Observe that some of the results in this section do not involve linear
combinations -- except when the substitution is applied, of course. This is
thanks to the choice of basis values, that allows to write unitary abstractions
as $\unibasique$, and to focus the pattern-matching on basis values only.

The definition of valuation $\sigma$ does not involve any condition on types or
type judgments. However, we need this sort of condition to formulate a
substitution lemma, hence the next definition.

\begin{defi}
	\label{def:well-valuation}
	A valuation $\sigma$ is said to be well-formed with regard to a context
	$\pGamma$ if for all $(\ps x_i \colon \pQ_i) \in \pGamma$, we have
	$\set{\ps x_i \mapsto \ps e_i} \subseteq \sigma$ and $~\entail \ps e_i
	\colon \pQ_i$. We write $\pGamma \Vdash \sigma$.
\end{defi}

\begin{rem}
	The valuation $\sigma$ obtained in Lemma~\ref{lem:od-substitution} is
	well-formed iff $\ps b$ is.
\end{rem}

\begin{lem}
	\label{lem:qua-subst-type}
	If $\pGamma \entail \ps t \colon \pQ$ is a  well-formed term and $\pGamma \Vdash \sigma$ is a
	well-formed valuation, then $\entail
	\sigma(\ps t) \colon \pQ$ holds.
\end{lem}
\begin{proof}
	The proof is done by induction on $\pGamma \entail \ps t \colon \pQ$.
	\begin{itemize}
		\item $~\entail \ps * \colon \pbasic$. Direct.
		\item $\ps x \colon \pQ \entail \ps x \colon \pQ$. Since $\ps x \colon
			\pQ \Vdash \sigma$, there is a well-formed term $\ps t$ such that
			$\set{\ps x \mapsto \ps e} \subseteq \sigma$ and thus $\sigma(\ps
			x)$ is well-formed.
		\item $\pGamma \entail \ini{\ps t} \colon \pQ_1 \pplus \pQ_2$.  By
			induction hypothesis, $\sigma(\ps t)$ is well-formed, and thus
			$\ini{\sigma(\ps t)} = \sigma(\ini{\ps t})$ is also well-formed.			
		\item $\pGamma_1, \pGamma_2 \entail \ps t_1 \ptimes \ps t_2 \colon
			\pQ_1 \ptimes \pQ_2$. By induction hypothesis, $\sigma(\ps t_1)$
			and $\sigma(\ps t_2)$ are well-formed, and thus $\sigma(\ps t_1
			\ptimes \ps t_2) = \sigma(\ps t_1) \ptimes \sigma(\ps t_2)$ is also
			well-formed.
		\item $~\entail \pzero \colon \pNat$. Direct.
		\item $\pGamma \entail \psucc{\ps t} \colon \pNat$.
			By induction hypothesis, $\sigma(\ps t)$ is well-formed, and thus
			$\psucc{\sigma(\ps t)} = \sigma(\psucc{\ps t})$ is also well-formed.
			\item $\pGamma \entail u~\ps t \colon \pQ_2$.  
			By induction hypothesis, $\sigma(\ps t)$ is well-formed, and thus
			$u~\sigma(\ps t) = \sigma(u~\ps t)$ is also well-formed.
	\item $\pGamma \entail \Sigma_i (\palpha_i \pdot \ps t_i) \colon \pQ$.
			By induction hypothesis, $\sigma(\ps t_i)$ is well-formed for all
			$i$, and thus $\Sigma_i (\palpha_i \pdot \sigma(\ps t_i)) =
			\sigma(\Sigma_i (\palpha_i \pdot \ps t_i))$ is also well-formed,
			thanks to Lemma~\ref{lem:ortho-subst-equiv}. \qedhere

	\end{itemize}
\end{proof}

Substitutions, now properly defined, help us formalize how the language handles
operations such as the $\beta$-reduction.

\subsection{Hilbert Spaces}
\label{sec:back-hilbert}

Some of the main focus in the study of mathematical quantum mechanics is Hilbert
spaces. We assume basic knowledge of linear algebra, such as vectors, linear
maps, bases, kernels, \emph{etc.} 

\subsubsection{Introductory Definitions}
\label{sub:hilb-intro}

A Hilbert space is a vector space equipped with an inner product,
written $\ip - -$, and this inner product induces a complete metric space. The
inner product is used to calculate probabilities of measurement outcomes in
quantum theory. Note that, given a complex number $\alpha$, we write
$\overline\alpha$ its conjugate.

\begin{defi}[Inner product]
	An inner product on a complex vector space $V$ is an operator $\ip - - \colon
	V \times V \to \C$, such that:
	\begin{itemize}
		\item for all $x,y \in V$, $\ip x y = \overline{\ip y x}$;
		\item for all $x,y,z \in V$ and $\alpha \in \C$,
			\[
				\ip{x}{\alpha y} = \alpha \ip x y,
				\qquad
				\ip{x}{y + z} = \ip x y + \ip x z;
			\]
		\item for all $x \in V$, $\ip x x \geq 0$, and if $\ip x x = 0$, then
			$x = 0$.
	\end{itemize}
\end{defi}

\begin{defi}[Norm]
	Given a vector space with an inner product, the norm of a vector $x$
	is defined as $\Vert x \Vert \defeq \sqrt{ \ip  x x }$.
\end{defi}

This is enough to state the definition of a Hilbert space.

\begin{defi}[Hilbert space]
	A Hilbert space is a vector space $H$ equipped with a inner product such
	that the induced norm is \emph{complete}. By complete, we mean that: if a
	sequence of vectors $(v_i)_{i \in \N}$ is such that $\sum_{i = 0}^\infty
	\Vert v_i \Vert ^2 < \infty$, then there exists a vector $v \in H$ such
	that $\Vert v - \sum_{i = 0}^n v_i \Vert$ tends to zero as $n$ goes to the
	infinity. In this situation, the vector $v$ is unique and we write $v =
	\sum_{i = 0}^\infty v_i$.
\end{defi}

It is interesting to observe that all finite-dimensional vector spaces with an
inner product are complete. Moreover, any vector space with an inner product
can be \emph{completed}, by adding the adequate limit vectors.

\begin{rem}
	A basis in a Hilbert space is not exactly defined the same way as a basis
	in a vector space. A basis in a Hilbert space is such that any vector
	is \emph{limit} of a linear combination of the elements of the basis.
\end{rem}

\begin{defi}[Bounded linear map]
	Given two Hilbert spaces $X$ and $Y$, a linear map $f \colon X \to Y$
	is \emph{bounded} if there exists $\alpha \in \R$ such that $\Vert f x \Vert
	\leq \alpha \Vert x \Vert$ for all $x \in X$.
\end{defi}

\subsubsection{Additional Structure}
\label{sub:hilb-additional}

We make use of different kinds of constructions for Hilbert spaces and vector
spaces which we proceed to recall.

\begin{defi}[Direct sum]
	Given two complex vector spaces $V$ and $W$, their \emph{direct sum}, written $V
	\oplus W,$ is the vector space whose underlying set is
	\[ V \times W \triangleq \{ (v, w) \ |\ v \in V, w \in W\} \]
	with vector space operations defined componentwise, \emph{i.e.},
	for all $v,v' \in V$, $w,w' \in W$, and $\alpha,\beta \in \C$,  
	$\alpha (v,w) + \beta (v',w') = (\alpha v + \beta v', \alpha w + \beta w')$ holds.
	If $H$ and $K$ are Hilbert spaces, then their direct sum $H \oplus K$ is also
	a Hilbert space when equipped with the inner product
	$\ip{(x_1,y_1)}{(x_2,y_2)}_{H \oplus K} \triangleq \ip{x_1}{x_2}_{H} + \ip{y_1}{y_2}_{K}$.
\end{defi}


Next, we recall the universal property of the vector space tensor product.

\begin{defi}[Tensor product]
	Given two complex vector spaces $V, W$, there is a vector space $V \otimes
	W$, together with a bilinear map $- \otimes- \colon V \times W \to V \otimes
	W :: (v,w) \mapsto v \otimes w$, such that for every bilinear map $h \colon V
	\times W \to Z$, there is a unique linear map $h' \colon V \otimes W \to Z$,
	such that $h = h' \circ \otimes$.
\end{defi}

For the Hilbert space tensor product, the situation is a bit more complicated.
If $H$ and $K$ are two infinite-dimensional Hilbert spaces, then their vector
space tensor product $H \otimes K$ is not a Hilbert space, because it is not
complete with respect to its usual norm topology. The Hilbert space tensor
product of $H$ and $K$ is defined in the following way. First, we define an
inner product on the vector space tensor product $H \otimes K$ by $\ip{x_1
\otimes y_1}{x_2 \otimes y_2}_{H \otimes K} \triangleq \ip{x_1}{x_2}_H \ip{y_1}{y_2}_K$.
Then, the completion of $H \otimes K$ with respect to the induced norm is a
Hilbert space when we continuously extend the vector space operations and the
inner product to the completed space. We again write $H \otimes K$ for this
Hilbert space. In the sequel, $H \otimes K$ always means the Hilbert space
tensor product, so this should hopefully not lead to confusion.

The category of Hilbert spaces and bounded linear maps between them, written
$\cat{Hilb}$, is symmetric monoidal with $(\otimes, \C)$ and with $(\oplus, \set
0)$. Classical computers operate on bits, while quantum computers apply
operations on qubits, written $\ket 0$ and $\ket 1$. They are usually denoted
as vectors in the Hilbert space $\C \oplus \C$ with $\ket 0 \defeq (1,0)$ and
$\ket 1 \defeq (0,1)$ the elements of its canonical basis.

There is another important structure to Hilbert spaces, called the \emph{adjoint}.

\begin{lem}
	Given a linear map $f \colon X \to Y$ between Hilbert spaces, there is a
	unique $f\inv \colon Y \to X$ such that for all $x \in X$ and $y \in Y$,
	$\ip{f(x)}{y} = \ip{x}{f\inv(y)}$. The linear map $f\inv$ is called the
	adjoint of $f$.
\end{lem}

The category $\cat{Hilb}$ is equipped with a structure of symmetric monoidal
\emph{dagger} category, meaning that $(-)\inv$ is an involutive contravariant
endofunctor which is the identity on objects. Moreover, the dagger and the
monoidal tensor respect some coherence conditions. For example, given
two bounded linear maps $f \colon X_1 \to Y_1$ and $g \colon X_2 \to Y_2$,
\[
	(f \otimes g)\inv = f\inv \otimes g\inv \colon X_2 \otimes Y_2 \to X_1 \otimes
	Y_1.
\]

\begin{rem}
	We choose to use the notation $(-)^*$ instead of $(-)^\dagger$ to ensure
	consistency with the latter sections of the paper, which deal with
	$*$-algebras.
\end{rem}

\begin{rem}
	Given two maps $f \colon X \to Y$ and $g \colon Y \to Z$, we write $gf \colon X
	\to Z$ for the composition $g \circ f \colon X \to Z$. In addition, given a
	complex number $\alpha$ and a map $f \colon A \to B$, we write $\alpha f
	\colon X \to Y$ for the multiplication of the vector space $\alpha . f
	\colon X \to Y$.
\end{rem}

\begin{defi}
	Given a morphism $f \colon X \to Y$ in $\cat{Hilb}$, we say that $f$ is:
	\begin{itemize}
		\item a unitary, if it is an isomorphism and $f^{-1} = f\inv$;
		\item a contraction (or contractive map), if for all $x \in X$, $\norm{fx}
			\leq \norm x$;
		\item an isometry, if $f\inv f = \mathrm{id}$;
		\item a coisometry, if $ff\inv = \mathrm{id}$.
	\end{itemize}
\end{defi}

\subsubsection{Quantum Computing}
\label{sub:hilb-computing}

Note that Hilbert spaces and unitaries (resp. contractions) form a dagger
category. They are wide subcategories of $\cat{Hilb}$. Unitary maps are of central
importance for quantum computation. One of the most important unitary maps in quantum
computing is the basis change in $\C \oplus \C$, also known as the
\emph{Hadamard gate}: $\ketbra 0 + + \ketbra 1 -$, where $\ket + \triangleq
\frac{1}{\sqrt 2} (\ket 0 + \ket 1)$ and $\ket - \triangleq \frac{1}{\sqrt 2} (\ket 0 -
\ket 1)$. We observe here that $\ketbra 0 +$ and $\ketbra 1 -$ are
contractions, and that unitary maps can be formulated as linear combinations of
compatible contractive maps. Note also that the (normalized) states can be
seen as isometries, \emph{e.g.}, the state $\ket 0$ can be seen as the isometry
$\mathbb C \to \mathbb C^2$ defined by $\lambda \mapsto \lambda \ket 0.$

We write $\Contr$, the category of \emph{countably-dimensional} Hilbert
spaces and contractive maps, $\Iso$, the category of
\emph{countably-dimensional} Hilbert and isometries between them, and $\Coiso$
the category of \emph{countably-dimensional} Hilbert spaces and coisometries
between them. The category of \emph{countably-dimensional} Hilbert spaces and
bounded linear maps is often written $\Hilbc$.

\begin{lem}[Zero map]
	Given two Hilbert spaces $X$ and $Y$, there exists a linear map
	$0_{X,Y} \colon X \to Y$ whose image is $\set 0$, we also
	write that $\Ker (0_{X,Y}) = X$. When it is not ambiguous, we
	write $0 \colon X \to Y$. It is a contractive map.
\end{lem}

\begin{rem}
	Given a Hilbert space $X$, the morphism $0 \colon X \to \set 0$ is unique
	for every $X$ and makes $\set 0$ a terminal object in $\Contr$ and in
	$\Coiso$.
\end{rem}

Contractions are widely used in the literature for the denotational semantics
of quantum programming languages \cite{heunen2022information,
pablo2022universal, carette2023quantum, phase}. Some recent developments expose axioms
for the categories involved in this paper \cite{heunen2022axioms,
heunen2022contractions, dimeglio2024dagger}. This better mathematical understanding
of the category theory behind Hilbert spaces can only be beneficial for the
theory of programming languages.

\subsubsection{The $\ell^2$ construction.}
Given a set $X$, we can define a Hilbert space $\ell^2(X)$ with underlying
set given by
\begin{equation}
	\label{eq:l2}
	\ell^2(X) \defeq \set{\phi \colon X \to \C \alt \sum_{x \in X}
	\abs{\phi(x)}^2 < \infty},
\end{equation}
with vector space structure defined pointwise, and with inner product defined by
\[ \langle \phi, \psi \rangle \defeq \sum_{x \in X} \overline{\phi(x)}\psi(x) . \]
Given an element $x \in X$, we write $\ket x$ for the element $\ket x \in \ell^2(X)$
defined by $\ket x(x) \triangleq 1$ and $\ket x(y) \triangleq 0$ for $y \neq x.$
The family $\{\ket x\}_{x \in X}$ is called the \emph{canonical (orthonormal) basis} of $\ell^2(X)$.

\subsection{Hilbert Spaces as a Semantics}
\label{sec:qu-simple-maths}

We build upon the notations and definitions presented in
\secref{sec:back-hilbert}, where we outlined introductory notions on Hilbert
spaces. Our goal is now to detail the tools necessary to interpret the pure
quantum syntax in Hilbert spaces and specific linear maps between them. 
We mostly work with contractions, that the reader can see as
``subunitary'' linear maps. We eventually show that we can interpret terms of
our language as isometries (the proper maths for quantum states), and
isometries are, in particular, contractions. 

\subsubsection{Sums.}
Given two maps $f,g \colon X \to Y$ in $\Contr$, their linear algebraic sum $f
+ g$ is not necessarily a contraction. We introduce the notion of
\emph{compatibility}. This notion is inherited from the one in restriction and
inverse categories \cite{kaarsgaard2017join}. We use, in particular, an
observation that links zero morphisms to compatibility in inverse categories;
but this is adapted to Hilbert spaces in this section. In this context, the
\emph{join} is also different, as we use the algebraic sum inherited from the
vector space structure.


\begin{defi}[Compatibility]
	\label{def:compati}
	Given two maps $f,g \colon X \to Y$  in $\Contr$, $f$ and $g$ are said to
	be compatible if $(\Ker f)^\bot \cap (\Ker g)^\bot = \set 0$ and
	$\iim f \cap \iim g = \set 0$.
\end{defi}

This definition of compatibility ensures that there is no overlap between the
inputs, and also between the outputs. The next lemma is then direct.

\begin{lem}
	\label{lem:compati}
	Given two compatible contractive maps $f,g \colon X \to Y$, $f+g$ is
	also contractive.
\end{lem}

As said above, the conditions in Def.~\ref{def:compati} can be simplified in
more algebraic expressions.

\begin{lem}
	\label{lem:algebra-compati}
	Given two contractive maps $f,g \colon X \to Y$, if $f\inv g = 0$ and $fg\inv = 0$
	then $f$ and $g$ are compatible.
\end{lem}

\begin{rem}
	It might seem that the conditions introduced in
	Lemma~\ref{lem:algebra-compati} are not symmetric on $f$ and $g$. But one
	can observe that $0\inv = 0$ and $(f\inv g)\inv = g\inv f\inv{\inv} = g\inv
	f$, thus $f\inv g = 0$ iff $g\inv f = 0$. Similarly, $fg\inv = 0$ iff
	$gf\inv = 0$.
\end{rem}

Lemma~\ref{lem:algebra-compati} introduces a new point of view on
compatibility, by describing the \emph{orthogonality} between morphisms.

\begin{exa}
	The morphisms $\ketbra 0 0 \colon \C \to \C^2$ and $\ketbra 1 1 \colon \C \to
	\C^2$ are orthogonal in our loose sense, because the vectors $\ket 0$ and
	$\ket 1$ are orthogonal in the linear algebraic sense. This justifies that
	their linear sum $\ketbra 0 0 + \ketbra 1 1$ is a contraction (and, in this
	case, it is also a unitary).
\end{exa}

\subsubsection{Additional structure.}
Unsurprisingly, the unit type is to be represented by the one-dimensional
Hilbert space $\C$, the line of complex numbers. In the syntax, orthogonality
and thus pattern-matching, depend on direct sums. The latter are interpreted
as direct sums of Hilbert spaces. We show that this interpretation gives
rise to orthogonality in the sense of contractions.

\begin{defi}
	We write $\iota^{X,Y}_l \colon X \to X \oplus Y$ the isometry such that for
	all $x \in X$, $\iota^{X,Y}_l x = (x,0)$; it is called the left injection.
	Similarly, the right injection is written $\iota^{X,Y}_r \colon Y \to X
	\oplus Y$.
\end{defi}

\begin{lem}
	\label{lem:injection-compati}
	Given two Hilbert spaces $X$ and $Y$, $(\iota^{X,Y}_l)\inv
	\iota^{X,Y}_r = 0$ and $(\iota^{X,Y}_r)\inv \iota^{X,Y}_l = 0$.
\end{lem}

\begin{exa}
	\label{ex:qua-iota}
	The previous lemma ensures that $\iota^{X,Y}_l (\iota^{X,Y}_l)\inv$ and
	$\iota^{X,Y}_r (\iota^{X,Y}_r)\inv$ are compatible. Note that $\iota^{X,Y}_l
	(\iota^{X,Y}_l)\inv + \iota^{X,Y}_r (\iota^{X,Y}_r)\inv = \iid$.
\end{exa}

Note that given a complex number $\alpha$ and a contraction $f \colon X \to Y$,
the outer product $\alpha \cdot f$ is written $\alpha f$ when it is not
ambiguous. Given a set $S$, we write $(\alpha_i)_{i \in S}$ a family of
complex numbers indexed by $S$. Given two sets $S$ and $S'$, we write
$(\alpha_{i,j})_{(i,j) \in S \times S'}$ a matrix of complex numbers indexed by
$S$ and $S'$. The sets of indices can be omitted if there is no ambiguity, as
in Lemma~\ref{lem:normalised-sum}.

\begin{lem}
	\label{lem:normalised-sum}
	Given a family of pairwise output compatible isometries $f_i \colon X \to Y$,
	and a family of complex numbers $\alpha_i$ such that $\sum_i \abs{\alpha_i}^2
	= 1$, $\sum_i \alpha_i f_i$ is an isometry.
\end{lem}
\begin{proof}
	\begin{align*} 
		\ (\sum_i \alpha_i f_i)\inv \circ (\sum_j \alpha_j f_j) 
		&= (\sum_i \overline\alpha_i f_i\inv) \circ (\sum_j \alpha_j f_j) 
		& (\text{Dagger and sum commute}.) \\
		&= \sum_{i,j} (\overline\alpha_i \alpha_j) f_i\inv \circ f_j
		& (\text{Composition and sum commute}.) \\
		&= \sum_i (\overline\alpha_i \alpha_i) f_i\inv \circ f_i
		& (f_i\text{ are pairwise compatible}.) \\
		&= \left( \sum_i \abs{\alpha_i}^2 \right) \iid 
		& (f_i\text{ are isometries}.) \\
		&=\iid .&& \qedhere
	\end{align*}
\end{proof}

We recall that given two maps $f \colon X \to Z$ and $g \colon Y \to Z$ in $\Contr$, 
if $f\inv g = 0_{Z,X}$, we say that $f$ and $g$ are orthogonal. We show that this
orthogonality is preserved by postcomposing with an isometry.

\begin{lem}
	\label{lem:ortho-postiso}
	Given two orthogonal maps $f \colon X \to Z$ and $g \colon Y \to Z$ in
	$\Contr$. If $h \colon Z \to W$ is an isometry, then $h \circ f \colon X
	\to W$ and $h \circ g \colon Y \to W$ are also orthogonal.
\end{lem}
\begin{proof}
	\begin{align*}
		(h \circ f)\inv \circ h \circ g 		&= f\inv \circ h\inv \circ h \circ g
		& (\text{Dagger is contravariant.}) \\
		&
		= f\inv \circ g
		& (\text{by isometry}) \\
		&= 0_{Z,X}.
		& (\text{by hypothesis})
		& \qedhere
	\end{align*}
\end{proof}

In a similar vein, the postcomposition of an isometry with an isometry says an
isometry. This was already observed when we mentioned that Hilbert spaces and
isometries form a category.

The canonical countably-dimensional Hilbert spaces is $\ell^2(\N)$, defined in
\secref{sec:back-hilbert}.

\begin{defi}
	We write $\rmsucc \colon \ell^2(\N) \to \ell^2(\N)$ for the linear map
	$\ell^2(\N) \to \mathbb C \oplus \ell^2(\N) \cong \ell^2(\N)$.
\end{defi}

\begin{rem}
	Note that $\rmsucc$ can also be seen as the image of the successor function
	over the natural numbers by the function $\ell^2$. The linear map $\rmsucc$ is an
	isometry.
\end{rem}

We recall that we write $\ket n$ for the elements of the canonical basis in
$\ell^2(\N)$. Because of this, $\ket 0$ and $\ket 1$ can either be elements of
$\C^2$ or $\ell^2(\N)$. This is not an issue, since the natural embedding $\C^2
\to \ell^2(\N)$ maps $\ket 0$ to $\ket 0$ and $\ket 1$ to $\ket 1$.

\begin{exa}
	\[
		\rmsucc \ket 7 = \ket 8
		\qquad
		\rmsucc \left(\frac{\sqrt 3}{2}\ket 9 + \frac 1 2 \ket{11} \right)
		= \frac{\sqrt 3}{2}\ket{10} + \frac 1 2 \ket{12}
	\]
\end{exa}

\subsubsection{Unitaries.}
The denotational semantics of our programming language involves unitary maps
to interpret the functions. Those maps live in the category $\cat{Uni}$, which is a 
rig category: it has bifunctor $\oplus$ and $\otimes$ inherited from $\cat{Hilb}$. Hence
the next lemma.

\begin{lem}
	Given two maps $f \colon X \to Y$ and $g \colon Z \to W$ in $\cat{Uni}$, $f
	\otimes g \colon X \otimes Z \to Y \otimes W$ is a map in $\cat{Uni}$ and
	$f \oplus g \colon X \oplus Z \to Y \oplus W$ is a map in $\cat{Uni}$.
\end{lem}

Finally, we present an operation that is common to quantum computing, and thus
preserves the unitary structure.

\begin{lem}[Controlled unitary]
	For each unitary map $f \colon X \to X$, the map $\rmctrl_X(f)
	\colon (\C \oplus \C) \otimes X \to (\C \oplus \C) \otimes X$ defined by
	$\rmctrl_X (f) \triangleq \ketbra 0 0 \otimes \iid + \ketbra 1 1 \otimes f$ is unitary.
\end{lem}

\subsection{Denotational Semantics}
\label{sec:very-simple-semantics}

As usual, we write $\sem -$ for the interpretation of types and term
judgments. As mentioned in the previous section, the presentation makes
extensive use of contractions for the denotational semantics. However, we
provide a semantics to values and terms as isometries from the start, for
clarity. It also help us highlight the fact that values and terms represent
sound quantum states. In the same vein, the interpretation of unitaries is
given as unitary maps between two Hilbert spaces; but the proof that the
semantics of a unitary abstraction is unitary requires the mathematical
development at the level of contractions.

\subsubsection{Types.}
The interpretation of a type $\pQ$ is given by a countably-dimensional Hilbert
space. It is given by induction on the grammar of the types. This
interpretation is detailed in
Figure~\ref{fig:-very-simple-type-interpretation}.

\begin{figure}
	\begin{align*}
		\sem{\pQ} &\colon \Hilb \\
		\sem{\pbasic} &\triangleq \C \\
		\sem{\pQ_1 \pplus \pQ_2} &\triangleq \sem{\pQ_1} \oplus \sem{\pQ_2} \\
		\sem{\pQ_1 \ptimes \pQ_2} &\triangleq \sem{\pQ_1} \otimes \sem{\pQ_2} \\
		\sem{\pNat} &\triangleq \ell^2(\N)
	\end{align*}
	\caption{Interpretation of types.}
	\label{fig:-very-simple-type-interpretation}
\end{figure}

\subsubsection{Expressions.}
Judgments for expressions are interpreted in Figure~\ref{fig:very-simple-value-interpretation} as contractions between Hilbert
spaces. A judgment is of the form $\pGamma \entail \ps e \colon \pQ$, and its
interpretation is written $\sem{\pGamma \entail \ps e \colon \pQ}$. Contexts
$\pGamma = \ps x_1 \colon \pQ_1 \dots \ps x_n \colon \pQ_n$ are given a
denotation $\sem\pGamma \triangleq \sem{\pQ_1} \otimes \dots \otimes \sem{\pQ_n}$. When
it is not ambiguous, the interpretation of the judgment $\pGamma \entail \ps e
\colon \pQ$ is written $\sem{\ps e}$.

\begin{figure}
	\begin{align*}
		\sem{\pGamma \entail \ps v \colon \pQ} &\colon \Iso(\sem	\pGamma,\sem \pQ) \\
		\sem{\entail \ps * \colon \pbasic} &\triangleq \iid_{\sem \pbasic} \\
		\sem{\ps x \colon \pQ \entail \ps x \colon \pQ} &\triangleq \iid_{\sem \pQ} \\
		\sem{\pGamma \entail \inl{\ps e} \colon \pQ_1 \pplus \pQ_2} &\triangleq \iota_l^{\sem{\pQ_1}, \sem{\pQ_2}}
		\circ \sem{\pGamma \entail \ps e \colon \pQ_1} \\
		\sem{\pGamma \entail \inr{\ps e} \colon \pQ_1 \pplus \pQ_2} &\triangleq \iota_r^{\sem{\pQ_1}, \sem{\pQ_2}}
		\circ \sem{\pGamma \entail \ps e \colon \pQ_2} \\
		\sem{\pGamma_1, \pGamma_2 \entail \pair{\ps e}{\ps e'}\colon \pQ_1 \ptimes \pQ_2} &\triangleq
		\sem{\pGamma_1 \entail \ps  e \colon \pQ_1} \otimes \sem{\pGamma_2 \entail
		\ps e' \colon \pQ_2} \\
		\sem{\entail \pzero \colon \pNat} &\triangleq \ket 0 \\
		\sem{\pGamma \entail \psucc{\ps e} \colon \pNat} &\triangleq \rmsucc \circ \sem{\pGamma
		\entail \ps e \colon \pNat} \\
		\sem{\pGamma \entail \Sigma_{i\leq k} (\palpha_i \pdot \ps e_i) \colon
		\pQ} &\triangleq \sum_{i\leq k} \palpha_i \sem{\pGamma \entail \ps e_i \colon
		\pQ}
	\end{align*}
	\caption{Interpretation of value judgments as morphisms in $\Iso$.}
	\label{fig:very-simple-value-interpretation}
\end{figure}

\begin{lem}[Isometry]
	\label{lem:value-isometry}
	If $\pGamma \entail \ps e \colon \pQ$ is a well-formed value judgment,
	then $\sem{\pGamma \entail \ps e \colon \pQ}$ is an isometry.
\end{lem}

In pure quantum computation, the quantum operations are usually described as 
isometries, so our previous lemma is well-motivated. The proof of the previous 
lemma is included in one of a larger result, showing that the denotation of all 
terms are isometries (see Lemma~\ref{lem:term-isometry}). Moreover, expressions 
are used to define the unitary abstractions as a collection of patterns: it is 
sensible to prove that these patterns are interpreted with compatible morphisms, 
in the sense of Definition~\ref{def:compati}.

\begin{lem}
	\label{lem:orthogonal-semantics-ortho-value}
	Given two judgments $\pGamma_1 \entail \ps e_1 \colon \pQ$ and
	$\pGamma_2 \entail \ps e_2 \colon \pQ$, if $\ps e_1~\bot~\ps e_2$ holds,
	then $\sem{\ps e_1}\inv \circ \sem{\ps e_2} = 0$.
\end{lem}
\begin{proof}
	The proof is done by induction on the derivation of $\bot$. It is
	a subproof of the one for Lemma~\ref{lem:orthogonal-semantics-ortho}.
\end{proof}

This result can also be stated for the predicate ${\rm ONB}$, with the help of
Lemma~\ref{lem:od-imp-bot}, where it is shown that two values in an orthogonal
decomposition are orthogonal.

\begin{lem}
	\label{lem:orthogonal-semantics}
	Given two judgments $\pGamma_1 \entail \ps e_1 \colon \pQ$ and
	$\pGamma_2 \entail \ps e_2 \colon \pQ_1$, and a set of $S$ such that $\{\ps
	e_1,\ps e_2\} \subseteq S$. If $\ODe{\pQ}{(S)}$ holds, then $\sem{\ps e_1}\inv
	\circ \sem{\ps e_2} = 0$.
\end{lem}

The predicate ${\rm ONB}$ provides more information on the \emph{basis} $S$;
namely, that it is exhaustive. The next lemma shows that a linear function
that maps each element of the basis to itself is the identity.

\begin{lem}
	\label{lem:towards-unitary}
	Assume that $\ODe{\pQ}{(\set{\ps e_i}_{i \leq n})}$ holds. If $\pGamma_i \entail \ps e_i
	\colon \pQ$, for all $i$, then 
	\[
		\sum_{i \leq n} \sem{\ps e_i} \circ \sem{\ps e_i}\inv =
		\iid_{\sem{\pQ_1}}.
	\]
\end{lem}
\begin{proof}
	The proof is done by induction on ${\rm ONB}$.
	\begin{itemize}
			\item $\ODe{\pbasic}{(\set{\ps *})}$. 
			$\sem{\ps *} \circ \sem{\ps *}\inv = \iid_{\sem \pbasic} \circ
			\iid_{\sem \pbasic} = \iid_{\sem \pbasic}$.
		\item $\ODe{\pQ}{(\set{\ps x})}$.
			$\sem{\ps x} \circ \sem{\ps x}\inv = \iid_{\sem \pQ} \circ \iid_{\sem \pQ} = \iid_{\sem \pQ}$.
		\item $\ODe{\pQ_1 \pplus \pQ_2}{(S \boxplus T)}$. 
			\begin{align*}
				 \sum_{\ps e \in S \boxplus T} \sem{\ps e} \circ \sem{\ps e}\inv 
				&= \sum_{\ps s \in S} \sem{\inl{\ps s}} \circ \sem{\inl{\ps s}}\inv
				+
				\sum_{\ps t \in T} \sem{\inr{\ps t}} \circ \sem{\inr{\ps t}}\inv
				&\text{(by definition)} \\
				&= \iota_l \circ \left( \sum_{\ps s \in S} \sem{\ps s} \circ
				\sem{\ps s}\inv \right) \circ \iota_l\inv
				+
				\iota_r \circ \left( \sum_{\ps t \in T} \sem{\ps t} \circ
				\sem{\ps t}\inv \right) \circ \iota_r\inv
				&\text{(by linearity)} \\
				&= \iota_l \circ \iid_{\sem{\pQ_1}} \circ \iota_l\inv
				+
				\iota_r \circ \iid_{\sem{\pQ_2}} \circ \iota_r\inv
				&\text{(by IH)} \\
				&= \iota_l \iota_l\inv + \iota_r \iota_r\inv 
				&\text{(by Ex.~\ref{ex:qua-iota})} \\
				&= \iid_{\sem{\pQ_1 \oplus \pQ_2}}.
			\end{align*}
		\item $\ODe{\pQ_1 \otimes \pQ_2}{(S)}$. Suppose that $\ODe{\pQ_1}{(\pi_1(S))}$ and
			$\ODe{\pQ_2}{(S^1_{\ps e})}$ for all $\ps e \in \pi_1(S)$. 
			\begin{align*}
				\sum_{(\ps e \ptimes \ps e') \in S} &\sem{\ps e \ptimes \ps
				e'} \circ \sem{\ps e \ptimes \ps e'}\inv & \\
				&= \sum_{(\ps e \otimes \ps e') \in S} (\sem{\ps e} \otimes
				\sem{\ps e'}) \circ (\sem{\ps e} \otimes \sem{\ps e'}\inv)
 				&\text{(by definition)} \\
				&= \sum_{(\ps e \ptimes \ps e') \in S} (\sem{\ps e} \circ \sem{\ps e}\inv) \otimes 
 				(\sem{\ps e'} \circ \sem{\ps e'}\inv )
 				&\text{(by monoidal } \dagger\text{-category)} \\
				&= \sum_{\ps e \in \pi_1(S)} (\sem{\ps e} \circ \sem{\ps
				e}\inv) \otimes \left( \sum_{\ps e' \in S^1_{\ps b}} \sem{\ps
				e'} \circ \sem{\ps e'}\inv \right) & \\
				&= \sum_{\ps e \in \pi_1(S)} (\sem{\ps e} \circ \sem{\ps
				e}\inv) \otimes \iid_{\sem{\pQ_2}}
 				&\text{(by IH)} \\
				&= \left( \sum_{\ps e \in \pi_1(S)} \sem{\ps e} \circ \sem{\ps e}\inv \right)
 				\otimes \iid_{\sem{\pQ_2}}
 				&\text{(by linearity)} \\
 				&= \iid_{\sem{\pQ_1}}
 				\otimes \iid_{\sem{\pQ_2}}
 				&\text{(by IH)} \\
 				&= \iid_{\sem{\pQ_1 \ptimes \pQ_2}}. 
			\end{align*}
		\item $\ODe{\pNat}{(S^{\oplus 0})}$. 
			\begin{align*}
				\sum_{\ps e \in S^{\oplus 0}} \sem{\ps e} \circ \sem{\ps e} \inv 
				&= \sem\pzero \circ \sem\pzero +
				\sum_{\ps s \in S} \sem{\psucc{\ps s}} \circ \sem{\psucc{\ps s}}\inv
 				&\text{(by definition)} \\
				&= \sem\pzero \circ \sem\pzero\inv +
				\rmsucc \circ \left( \sum_{\ps s \in S} \sem{\ps s} \circ \sem{\ps s} \inv \right)
				\circ \rmsucc \inv
 				&\text{(by linearity)} \\
				&= \sem\pzero \circ \sem\pzero\inv +
				\rmsucc \circ \iid_{\sem\pNat} \circ \rmsucc \inv
 				&\text{(by IH)} \\
 				&= \iid_{\sem\pNat}.
			\end{align*}
		\item $\ODe{\pQ}{(S^\alpha)}$. 
			\begin{align*}
				\sum_{\ps e \in S^\alpha} \sem{\ps e} \circ \sem{\ps e} \inv 
				&= \sum_{\ps b \in S}
				\sem{\Sigma_{\ps s' \in S} (\palpha_{\ps s,\ps s'} \pdot \ps s')}
				\circ
				\sem{\Sigma_{\ps s' \in S} (\palpha_{\ps s,\ps s'} \cdot \ps s')} \inv
 				&\text{(by definition)} \\
				&= \sum_{\ps s \in S}
				\left(\sum_{\ps s' \in S} \palpha_{\ps s,\ps s'} \sem{\ps s'} \right)
				\circ
				\left(\sum_{\ps s'' \in S} \palpha_{\ps s,\ps s''} \sem{\ps s''} \right) \inv
 				&\text{(by definition)} \\
				&= \sum_{\ps s \in S}
				\sum_{\ps s',\ps s'' \in S}
				\palpha_{\ps s,\ps s'} \overline{\palpha_{\ps s,\ps s''}}
				\sem{\ps s'} \circ \sem{\ps s''} \inv
 				&\text{(by linearity)} \\
				&= \sum_{\ps s',\ps s'' \in S} \left( \sum_{\ps s \in S}
				\palpha_{\ps s,\ps s'} \overline{\palpha_{\ps s,\ps s''}}
				\right) \sem{\ps s'} \circ \sem{\ps s''} \inv & \\
				&= \sum_{\ps s',\ps s'' \in S}
				\delta_{\ps s' = \ps s''}
				\sem{\ps s'} \circ \sem{\ps s''} \inv \\
				&= \sum_{\ps s' \in S}
				\sem{\ps s'} \circ \sem{\ps s'} \inv 
				&\text{(by unitarity)} \\
				&= \iid_{\sem \pQ}.
 				&\text{(by IH)}
				& \qedhere
			\end{align*}\qedhere
	\end{itemize}
\end{proof}

One final development on the interpretation of values is the link
with substitutions, detailed in the next proposition.

\begin{prop}
	\label{prop:substitution-interpretation}
	Given a well-formed term $\pGamma \entail \ps t \colon \pQ$ and for all $(\ps x_i
	\colon \pQ_i) \in \pGamma$, a well-formed expression $~\entail \ps e_i
	\colon \pQ_i$; if $\sigma = \set{\ps x_i \mapsto \ps e_i}_i$, then:
	\[
		\sem{\entail \sigma(\ps t) \colon \pQ} = \sem{\pGamma \entail t \colon
		\pQ} \circ \left(\bigotimes_i \sem{\entail \ps e_i \colon
		\pQ_i}\right).
	\]
	We then define $\sem\sigma \triangleq \bigotimes_i \sem{\entail \ps e_i \colon \pQ_i}$.
\end{prop}
\begin{proof}
	The proof is straightforward by induction on the formation rules for $\ps t$.
\end{proof}

\begin{rem}
	The definition of the interpretation of a substitution above is somewhat
	informal. It would require a lot of care and unnecessary details to make
	the denotation of $\sigma$ fit the denotation of a particular context
	$\pGamma$. Since we are working in symmetric monoidal categories, those
	details will be overlooked when working with substitutions. We assume that we
	work up to permutations, and that when an interpretation of a substitution is
	involved, it is with the right permutation.
\end{rem}

The substitution $\sigma$ emerges from the matching of two basis values, thus
we can prove that the interpretation of the matching gives the interpretation
of the substitution, as stated in the next lemma.

\begin{lem}
	\label{lem:matching-semantics}
	Given two well-formed basis values $\pGamma \entail \ps b \colon \pQ$ and
	$~\entail \ps b' \colon \pQ$, and a substitution $\sigma$, if
	$\match{\sigma}{\ps b}{\ps b'}$ then $\sem{\ps b}\inv \circ \sem{\ps b'} =
	\sem\sigma$.
\end{lem}
\begin{proof}
	The proof is straightforward by induction on $\match{\sigma}{\ps b}{\ps b'}$.
\end{proof}

\subsubsection{Unitaries.}
A unitary of type $\isotypeonetwo$ is first interpreted
as a morphism $\sem{\pQ_1} \to \sem{\pQ_2}$ in $\Contr$. As we will see, its 
interpretation actually lies in $\cat{Uni}$. We also show that the ${\rm ONB}$
conditions ensure that the denotation of (syntactic) unitaries is not only a
contractive map, but a unitary between Hilbert spaces. Working with
contractions is necessary to use the notion of compatibility: given a unitary
$\unibasique \colon \isotypeonetwo$, we provide an interpretation to each clause
$b_i \iso \ps e_i$ as a contraction $\sem{\pQ_1} \to \sem{\pQ_2}$, and prove that all
the contractions thus obtained are compatible, and can be summed. Unitary
judgments are of the form $\entailiso u \colon \isotypeonetwo$, and their
semantics is given by a morphism in $\cat{Uni}$:
\[ \sem{\entailiso u \colon \isotypeonetwo} \colon
\cat{Uni}(\sem{\pQ_1},\sem{\pQ_2}). \]

Given $\entailiso \unibasique \colon \isotypeonetwo$, the interpretation of a clause
$b_i \iso \ps e_i$ is the following contraction: $\sem{\pGamma_i \entail \ps e_i \colon
B} \circ \sem{\pGamma_i \entail b_i \colon \pQ_1}\inv$. It shall be read as follows:
if the input of type $\pQ_1$ matches with $b_i$, it provides a substitution with 
$\pGamma_i$, that is applied to $\ps e_i$. This is better understood with a
diagram:
\[
	\sem{\pQ_1}
	\stackrel{\sem{\pGamma_i \entail b_i \colon \pQ_1}\inv}{\xrightarrow{\hspace*{2.5cm}}}
	\sem{\pGamma_i}
	\stackrel{\sem{\pGamma_i \entail \ps e_i \colon \pQ_2}}{\xrightarrow{\hspace*{2.5cm}}}
	\sem{\pQ_2}
\]

Lemma~\ref{lem:orthogonal-semantics} ensures that the semantics of clauses in a
well-formed unitary are pairwise compatible. The interpretation of a unitary
abstraction is then: \[\sem{\entailiso \unibasique \colon \isotypeonetwo} =
\sum_{i\leq n} \sem{\pGamma_i \entail \ps e_i \colon \pQ_2} \circ \sem{\pGamma_i \entail
b_i \colon \pQ_1}\inv.\]

This new contraction is the \emph{union}, or the \emph{join}, of the
interpretations of all the clauses. One can only take the \emph{join} of
compatible morphisms. It is left to prove that it is well-defined, and then
that it is a proper unitary operation.

\begin{cor}
	\label{cor:iso-sem-defined}
	Given $~\entailiso \unibasique \colon \isotypeonetwo$,
	its interpretation
	\[ \sem{\entailiso \unibasique \colon \isotypeonetwo} \]
	is a well-defined morphism in $\Contr$.
\end{cor}
\begin{proof}
	Given $~\entailiso \unibasique \colon \isotypeonetwo$, we know that
	$\ODe{\pQ_2}{(\set{\ps e_i}_{i\leq n})}$ and $\OD{\pQ_1}{(\set{\ps
	b_i}_{i\leq n})}$ hold, and for all $i \leq n$, $\pGamma_i \entail \ps b_i
	\colon \pQ_1$ and $\pGamma_i \entail \ps e_i \colon \pQ_2$.

	Since $\OD{\pQ_1}{(\set{\ps b_i}_{i\leq n})}$ holds,
	Lemma~\ref{lem:orthogonal-semantics} ensures that for all $i \neq j\leq k$,
	$\sem{\ps b_i}\inv \circ \sem{\ps b_j} =
	0_{\sem{\pGamma_j},\sem{\pGamma_i}}$. The same lemma with
	$\ODe{\pQ_2}{(\set{\ps e_i}_{i\leq n})}$ ensures that for all $i\neq j \leq
	n$, $\sem{\ps e_i}\inv \circ \sem{\ps e_j} =
	0_{\sem{\pGamma_j},\sem{\pGamma_i}}$. This proves that, for all $i \neq j
	\leq n$, $(\sem{\ps e_i} \circ \sem{\ps b_i}\inv)\inv \circ \sem{\ps e_j}
	\circ \sem{\ps b_j}\inv = 0_{\sem{\pGamma_j},\sem{\pGamma_i}}$ and
	$\sem{\ps v_i} \circ \sem{\ps b_i}\inv \circ (\sem{\ps e_j} \circ \sem{\ps
	b_j}\inv)\inv = 0_{\sem{\pGamma_j},\sem{\pGamma_i}}$. This proves that for
	all $i \neq j \leq n$, $\sem{\ps e_i} \circ \sem{\ps b_i}\inv$ and
	$\sem{\ps e_j} \circ \sem{\ps b_j}\inv$ are compatible, thanks to
	Lemma~\ref{lem:algebra-compati}. Then, Lemma~\ref{lem:compati} ensures that
	$\sum_{i\leq n} \sem{\ps e_i} \circ \sem{\ps b_i}\inv$ is a contraction.
\end{proof}

\begin{thm}
	Given $~\entailiso \unibasique \colon \isotypeonetwo$, its interpretation 
	\[ \sem{\entailiso \unibasique \colon \isotypeonetwo} \] is unitary.
\end{thm}
\begin{proof}
	Given $~\entailiso \unibasique \colon \isotypeonetwo$, we know that
	$\ODe{\pQ_2}{(\set{\ps e_i}_{i\leq n})}$ and $\OD{\pQ_1}{(\set{\ps
	b_i}_{i\leq n})}$ hold, and for all $i \leq n$, $\pGamma_i \entail \ps b_i
	\colon \pQ_1$ and $\pGamma_i \entail \ps e_i \colon \pQ_2$.

	First step, prove that $\sem u \inv \circ \sem u = \iid_{\sem{\pQ_1}}$,
	with $u = \unibasique$.
	\begin{align*}
		\sem u \inv \circ \sem u 
		&= \left( \sum_{i\leq n} \sem{\ps e_i} \circ \sem{\ps b_i}\inv \right)\inv \circ
		\sum_{j\leq n} \sem{\ps e_j} \circ \sem{\ps b_j}\inv
		& \text{(by definition)} \\
		&= \sum_{i\leq n} (\sem{\ps e_i} \circ \sem{\ps b_i}\inv)\inv \circ
		\sum_{j\leq n} \sem{\ps e_j} \circ \sem{\ps b_j}\inv
		& \text{(Dagger distributes over sum.)} \\
		&= \sum_{i\leq n} \sem{\ps b_i} \circ \sem{\ps e_i}\inv \circ
		\sum_{j\leq n} \sem{\ps e_j} \circ \sem{\ps b_j}\inv
		& \text{(Dagger is contravariant.)} \\
		&= \sum_{i,j \leq n} \sem{\ps b_i} \circ \sem{\ps e_i}\inv \circ \sem{\ps e_j} \circ
		\sem{\ps b_j}\inv 
		&\text{(by linearity)} \\
		&= \sum_{i \leq n} \sem{\ps b_i} \circ \sem{\ps e_i}\inv \circ \sem{\ps e_i} \circ
		\sem{\ps b_i}\inv 
		&\text{(by Lemma~\ref{lem:orthogonal-semantics})} \\
		&= \sum_{i \leq n} \sem{\ps b_i} \circ
		\sem{\ps b_i}\inv 
		&\text{(by Lemma~\ref{lem:value-isometry})} \\
		&= \iid_{\sem{\pQ_1}}.
		&\text{(by Lemma~\ref{lem:towards-unitary})}
	\end{align*}
	The other direction $\sem u \circ \sem u \inv = \iid_{\sem{\pQ_2}}$ is similar.
\end{proof}

Note that we have only proven so far that unitary abstractions have a sound
denotational semantics in $\cat{Uni}$. The interpretation of operations on unitaries
is given in Figure~\ref{fig:unitary-semantics}. It is explained in
\secref{sec:qu-simple-maths} why this interpretation is in $\cat{Uni}$, although
this does not come as a surprise.

\begin{figure}
	\begin{align*}
		\sem{\entailiso u \colon \pQ_1 \iso \pQ_2} &\colon \cat{Uni}(\sem{\pQ_1}, \sem{\pQ_2}) \\
		\sem{\unibreduit} &\triangleq \sum_{i \in I} \sem{\ps e_i} \circ \sem{\ps b_i}\inv \\
		\sem{u_1 \oplus u_2} &\triangleq \sem{u_1} \oplus \sem{u_2} \\
		\sem{u_1 \otimes u_2} &\triangleq \sem{u_1} \otimes \sem{u_2} \\
		\sem{u_2 \circ u_1} &\triangleq \sem{u_2} \circ \sem{u_1} \\
		\sem{u\inv} &\triangleq \sem{u}\inv \\
		\sem{\uctrl u} &\triangleq {\rm ctrl}_{\sem{\pQ_1}} (\sem u )
	\end{align*}
	\caption{Interpretation of unitaries in $\cat{Uni}$.}
	\label{fig:unitary-semantics}
\end{figure}

Moreover, any canonical symmetric monoidal isomorphism can be represented by
the unitaries of our language. The canonical symmetric monoidal isomorphisms
are simply those that can be defined using the structure of a symmetric
monoidal category only, \emph{i.e.}, isomorphisms that are given by: left and
right unitors, symmetric swaps, associators, closed under composition and
tensor products. We use this result later for quantum configurations.

\begin{prop}
    \label{prop:monoidal-unitary}
    Let $\pQ_1$ and $\pQ_2$ be two pure types, and let
    $f \colon \sem{\pQ_1} \to \sem{\pQ_2}$ be a canonical symmetric monoidal isomorphism. Then, there exists
    a unitary $\udash u \colon U(\pQ_1,\pQ_2)$ such that $\sem u = f$.
\end{prop}

\subsubsection{Terms.}
The interpretation of all term judgments is provided in
Figure~\ref{fig:simple-term-interpretation}. 

\begin{figure}
	\begin{align*}
		\sem{\pGamma \entail \ps t \colon \pQ} &\colon \Iso(\sem\pGamma,\sem{\pQ}) \\
		\sem{\entail \ps * \colon \pbasic} &\triangleq \iid_{\sem \pbasic} \\
		\sem{\ps x \colon \pQ \entail \ps x \colon \pQ} &\triangleq \iid_{\sem{\pQ}} \\
		\sem{\pGamma \entail \inl {\ps t} \colon \pQ_1 \pplus \pQ_2} &\triangleq
		\iota_l^{\sem{\pQ_1}, \sem{\pQ_2}} \circ \sem{\pGamma \entail {\ps
		t}\colon \pQ_1} \\
		\sem{\pGamma \entail \inr {\ps t} \colon \pQ_1\pplus \pQ_2} &\triangleq
		\iota_r^{\sem{\pQ_1}, \sem{\pQ_2}} \circ \sem{\pGamma \entail {\ps t}
		\colon \pQ_1} \\
		\sem{\pGamma_1, \pGamma_2 \entail \pair{\ps t}{\ps t'}\colon \pQ_1 \ptimes
		\pQ_2} &\triangleq \sem{\pGamma_1 \entail {\ps t} \colon \pQ_1} \otimes
		\sem{\pGamma_2 \entail \ps t' \colon \pQ_2} \\
		\sem{\entail \pzero \colon \pNat} &\triangleq \ket 0 \\
		\sem{\pGamma \entail \psucc {\ps t} \colon \pNat} &\triangleq \rmsucc \circ \sem{\pGamma
		\entail {\ps t} \colon \pNat} \\
				\sem{\pGamma \entail u~\ps t \colon \pQ_2} &\triangleq \sem{\entailiso u \colon
		\isotypeonetwo} \circ \sem{\pGamma \entail \ps t \colon \pQ_1}\\
		\sem{\pGamma \entail \Sigma_{i\leq k} (\palpha_i \pdot {\ps t}_i) \colon \pQ} &\triangleq
		\sum_{i\leq k} \palpha_i \sem{\pGamma \entail {\ps t_i} \colon \pQ} 
	\end{align*}
	\caption{Interpretation of term judgments as morphisms in $\Iso$.}
	\label{fig:simple-term-interpretation}
\end{figure}

We can already show that this interpretation of terms is sound with the sketch
of operational semantics given in the previous section.

\begin{prop}[Operational Soundness]
	\label{prop:qua-op-soundness}
	Given a well-formed unitary abstraction $~\entailiso \unibasique \colon
	\isotypeonetwo$ and a well-formed basis value $~\entail \ps b' \colon
	\pQ_1$, if $\match{\sigma}{\ps b_i}{\ps b'}$ holds, then 
	\[
		\sem{\entail \unibasique~\ps b' \colon \pQ_2} 
		= \sem{\entail \sigma(\ps e_i) \colon \pQ_2}.
	\]
\end{prop}
\begin{proof}
	First, we deduce from the assumption $\match{\sigma}{\ps b_i}{\ps b'}$ that
	\begin{itemize}
		\item $\neg(\ps b_i~\bot~\ps b')$, and thus $\sem{\ps b_i}\inv \circ \sem{\ps b'} =
			\sem\sigma$, thanks to Lemma~\ref{lem:matching-semantics}.
		\item for all $j \neq i$, $\ps b_j~\bot~\ps b'$, and thus $\sem{\ps b_j}\inv \circ
			\sem{\ps b'} = 0$, thanks to
			Lemma~\ref{lem:orthogonal-semantics-ortho-value}.
	\end{itemize}
	We can then compute the semantics, with $u \defeq \unibasique$:
	\begin{align*}
\sem{u~\ps b'} 
		&= \sem u \circ \sem{\ps b'} & \text{(by definition)} \\
		&= \left( \sum_j \sem{\ps e_j} \circ \sem{\ps b_j}\inv \right) \circ \sem{\ps b'}
		& \text{(by definition)} \\
		&= \sum_j \sem{\ps e_j} \circ \sem{\ps b_j}\inv \circ \sem{\ps b'}
		& \text{(by linearity)} \\
		&= \sem{\ps e_i} \circ \sem{\ps b_i}\inv \circ \sem{\ps b'}
		& \text{(by Lemma~\ref{lem:orthogonal-semantics-ortho-value})} \\
		&= \sem{\ps e_i} \circ \sem\sigma
		& \text{(by Lemma~\ref{lem:matching-semantics})} \\
		&= \sem{\sigma(\ps e_i)}.
		& \text{(by Proposition~\ref{prop:substitution-interpretation})}
		& \qedhere
	\end{align*}
\end{proof}

We prove that, like expressions, terms are indeed interpreted as isometries,
reinforcing the link with quantum physics. This requires several lemmas. The
first lemma shows that the denotational orthogonality is preserved by linear
combinations.

\begin{lem}
	\label{lem:sum-ortho}
	Given $\pGamma_1 \entail t \colon \pQ$ and $\pGamma_2 \entail \Sigma_i (\palpha_i
	\pdot \ps t_i) \colon \pQ$ such that for all $i$, $\sem{\pGamma_1 \entail \ps t \colon
	\pQ}\inv \circ \sem{\pGamma_2 \entail \ps t_i \colon \pQ} =
	0_{\sem{\pGamma_2},\sem{\pGamma_1}}$; then $\sem{\pGamma_1 \entail \ps t \colon
	A}\inv \circ \sem{\pGamma_2 \entail \Sigma_i (\palpha_i \pdot \ps t_i)
	\colon \pQ} = 0_{\sem{\pGamma_2},\sem{\pGamma_1}}$.
\end{lem}
\begin{proof}
	The proof involves few steps, without surprises.
	\begin{align*}
		\sem{\pGamma_1 \entail \ps t \colon \pQ}\inv
		\circ \sem{\pGamma_2 \entail \Sigma_i (\palpha_i \cdot \ps t_i) \colon \pQ}
		&= \sem{\pGamma_1 \entail \ps t \colon \pQ}\inv \circ \sum_i \palpha_i
		\sem{\pGamma_2 \entail \ps t_i \colon \pQ}
		& \text{(by definition)} \\
		&= \sum_i \palpha_i \sem{\pGamma_1 \entail \ps t \colon \pQ}\inv \circ
		\sem{\pGamma_2 \entail \ps t_i \colon \pQ}
		& \text{(by linearity)} \\
		&= \sum_i \palpha_i\, 0_{\sem{\pGamma_2},\sem{\pGamma_1}}
		& \text{(by hypothesis)} \\
		&= 0_{\sem{\pGamma_2},\sem{\pGamma_1}}. &
		& \qedhere
	\end{align*}
\end{proof}

The next lemma states that syntactic orthogonality (defined in
Definition~\ref{def:orthogonality}) implies denotational orthogonality.

\begin{lem}
	\label{lem:orthogonal-semantics-ortho}
	Given two judgments $\pGamma_1 \entail \ps t_1 \colon \pQ$ and $\pGamma_2
	\entail \ps t_2 \colon \pQ$, if $\ps t_1~\bot~\ps t_2$ holds, then	$\sem{\ps t_1}\inv \circ \sem{\ps t_2} =
	0_{\sem{\pGamma_2},\sem{\pGamma_1}}$.
\end{lem}

The proof of this lemma is interdependent with the proof of the next lemma,
where it is proven that the interpretations of term judgment are sound quantum
states, namely, isometries.

\begin{lem}[Isometry]
	\label{lem:term-isometry}
	If $\pGamma \entail \ps t \colon \pQ$ is a well-formed judgment, then
	$\sem{\pGamma \entail \ps t \colon \pQ}$ is an isometry.
\end{lem}
\begin{proof}[Proof of Lemma~\ref{lem:orthogonal-semantics-ortho} and 
	Lemma~\ref{lem:term-isometry}]
	We prove both theorems in a row, because they are interdependent. We remind
	the whole statement:
	Given a well-formed judgment $\pGamma_1 \entail \ps t_1 \colon \pQ$,
	\begin{itemize}
		\item $\sem{\pGamma_1 \entail \ps t_1 \colon \pQ}\inv \circ
			\sem{\pGamma_1 \entail \ps t_1 \colon \pQ} =
			\iid_{\sem{\pGamma_1}}$;
		\item for all $\pGamma_2 \entail \ps t_2 \colon \pQ$ such that $\ps
			t_1~\bot~\ps t_2$, $\sem{\pGamma_1 \entail \ps t_1 \colon \pQ}\inv
			\circ \sem{\pGamma_2 \entail \ps t_2 \colon \pQ} =
			0_{\sem{\pGamma_2},\sem{\pGamma_1}}$.
	\end{itemize}
	We prove by induction on the formation rules of values that $\sem{\ps t}
	\inv \circ \sem{\ps t} = \iid_{\sem\pGamma}$.
	\begin{itemize}
		\item $~\entail \ps * \colon \pbasic$. $\sem{\ps *} \inv \circ \sem{\ps
			*} = \iid_{\sem \pbasic} \circ \iid_{\sem \pbasic} = \iid_{\sem
			\pbasic}$. This term is not orthogonal to any other well-formed
			term.
		\item $\ps x \colon \pQ \entail \ps x \colon \pQ$. 
			$\sem{\ps x}\inv \circ \sem{\ps x} = \iid_{\sem{\pQ}} \circ
			\iid_{\sem{\pQ}} = \iid_{\sem{\pQ}}$. This term is not orthogonal to any other.
		\item $\pGamma \entail \ini{\ps t} \colon \pQ_1 \pplus \pQ_2$.
			The denotation of this term is an isometry because the injections $\iota$
			also are and a composition of isometries keeps being an isometry. Now,
			given another term $\pGamma_1 \entail \ps t_2 \colon \pQ_2$ such that
			$\ini{\ps t}~\bot~\ps t_2$, we have several cases:
			\begin{itemize}
				\item either $\ps t_2$ is of the form $\ini{\ps t'_2}$ with the
					same injection as $\ini{\ps t}$, and the result is given by
					Lemma~\ref{lem:ortho-postiso};
				\item or $\ini{\ps t} = \inl{\ps t}$ and $t_2 = \inr{\ps
					t'_2}$, in which case the conclusion is direct thanks to
					Lemma~\ref{lem:injection-compati} and the induction
					hypothesis;
				\item or $\ps t_2$ is a linear combination, and
					Lemma~\ref{lem:sum-ortho} concludes.
			\end{itemize}
		\item $\pGamma_1, \pGamma_2 \entail \ps t_1 \ptimes \ps t_2 \colon
			\pQ_1 \ptimes \pQ_2$.
			We first prove that its denotation is an isometry.
			\begin{align*}
				\sem{\ps t_1 \ptimes \ps t_2}\inv \circ \sem{\ps t_1 \ptimes
				\ps t_2} 
				&= (\sem{\ps t_1} \otimes \sem{\ps t_2})\inv \circ (\sem{\ps
				t_1} \otimes \sem{\ps t_2})
				& \text{(by definition)} \\
				&= (\sem{\ps t_1}\inv \otimes \sem{\ps t_2}\inv) \circ
				(\sem{\ps t_1} \otimes \sem{\ps t_2})
				& \text{(* is monoidal.)} \\
				&= (\sem{\ps t_1}\inv \circ \sem{\ps t_1}) \otimes (\sem{\ps
				t_2}\inv \circ \sem{\ps t_2}) 
				& (\otimes \text{ is monoidal.)} \\
				&= \iid_{\sem{\pGamma_1}} \otimes \iid_{\sem{\pGamma_2}}
				& \text{(by IH)} \\
				&= \iid_{\sem{\pGamma_1} \otimes \sem{\pGamma_2}} \\
				&=
				\iid_{\sem{\pGamma_1 , \pGamma_2}} .
			\end{align*}
			Then we show that, if it is orthogonal to any other well-formed
			term, say $\pGamma_3 \entail \ps t_3 \colon \pQ_1 \ptimes \pQ_2$,
			then their interpretations are also orthogonal. We reason on a case
			by case basis. The term $\ps t_3$ can be of the form $\ps t'_3
			\otimes \ps t''_3$, in which case, either $\ps t_1~\bot~\ps t'_3$
			or $\ps t_2~\bot~\ps t''_3$.  In both cases, the result is direct,
			because the zero morphism tensored with any other morphism is still
			zero. The other cases is $\ps t_3$ being a linear combination; this
			case is covered by Lemma~\ref{lem:sum-ortho}.
		\item $~\entail \pzero \colon \pNat$. $\sem\pzero\inv \circ \sem\pzero = \bra{0}\ket{0} = 1 = \iid_{\sem \pbasic}$. Moreover, given any natural number $n$,
			$\bra{0}\ket{n+1} = 0$. Lemma~\ref{lem:sum-ortho} concludes for linear
			combinations.
		\item $\pGamma \entail \psucc{\ps t} \colon \pNat$. The interpretation is an
			isometry by composition of isometries. The orthogonality part is either
			ensured with the last point, or by Lemma~\ref{lem:ortho-postiso} and
			the induction hypothesis.
					\item $\pGamma \entail u~\ps t \colon \pQ_2$. The interpretation is an isometry
			by composition of isometries. The orthogonality is covered by either
			Lemma~\ref{lem:ortho-postiso}, because a unitary is in particular an
			isometry, or Lemma~\ref{lem:sum-ortho}, similarly to the previous points. 
		\item $\pGamma \entail \Sigma_i (\palpha_i \pdot \ps t_i) \colon \pQ$.
			The interpretation is an isometry by induction hypothesis and 
			Lemma~\ref{lem:normalised-sum}. Moreover, the only case of orthogonality
			not yet covered: given that for all $i\neq j \in I$, $\ps t_i~\bot~\ps t_j$, $J,K
			\subseteq I$, and $\sum_{i \in J \cap K} \bar{\palpha_i}\pbeta_i = 0$:
			\begin{align*}
				&\ \sem{\Sigma_j (\palpha_j \pdot \ps t_j)}\inv 
				\circ \sem{\Sigma_k (\pbeta_k \pdot \ps t_k)} & \\
				&= \left( \sum_j \palpha_j \sem{\ps t_j} \right)\inv \circ
				\sum_k \pbeta_k \sem{\ps t_k}
				& \text{(by definition)} \\
				&= \left( \sum_j \overline\palpha_j \sem{\ps t_j}\inv \right) \circ
				\sum_k \pbeta_k \sem{\ps t_k}
				& \text{(Dagger distributes over the sum.)} \\
				&= \sum_{j,k} \overline\palpha_j \pbeta_k \sem{\ps t_j}\inv \circ \sem{\ps t_k}
				& \text{(by linearity)} \\
				&= \sum_{i \in J \cap K} \overline\palpha_i \pbeta_i \sem{\ps
				t_i}\inv \circ \sem{\ps t_i}
				& \text{(by IH)} \\
				&= \sum_{i \in J \cap K} \overline\palpha_i \pbeta_i \, \iid_{\sem\pGamma}
				& \text{(by IH)} \\
				&= \left( \sum_{i \in J \cap K} \overline\palpha_i \pbeta_i
				\right) \iid_{\sem\pGamma} \\
				&=	0_{\sem{\pGamma_2},\sem{\pGamma_1}}. &\qedhere
			\end{align*}
	\end{itemize}
\end{proof}

Providing an interpretation that fits the expectations from quantum physics is
meaningful, but not completely satisfying. We thus prove a stronger link between
the syntax and the semantics through an equational theory.

\subsection{Equational Theory}
\label{sub:quantum-equational}

The main equational theory of our language is given in
Figure~\ref{fig:eq-rules-qu-control}, with the rules of reflexivity, symmetry,
and transitivity given in Figure~\ref{fig:eq-rules-base}. Linear algebraic
identities are given in Figure~\ref{fig:eq-rules-vector-space} and congruence
identities are given in Figure~\ref{fig:eq-rules-qu-control-cong}.

\begin{rem}
	Among the equational rules presented in this section, only the equations in
	Figure~\ref{fig:eq-rules-qu-control} provide an operational account of the
	language. They can be seen as a reduction system from left to right. On the
	other hand, the equations in Figure~\ref{fig:eq-rules-vector-space} show that
	the algebraic power of Hilbert spaces and isometries can be captured within
	the type system.
\end{rem}

\begin{figure}
	\[ \begin{array}{c}
		\infer[(refl)]{
			\pGamma \entail \ps t = \ps t \colon A
		}{
			\pGamma \entail \ps t \colon A
		}
		\qquad
		\infer[(symm)]{
			\pGamma \entail \ps t_2 = \ps t_1 \colon A
		}{
			\pGamma \entail \ps t_1 = \ps t_2 \colon A
		}
		\mynl\\
		\infer[(trans)]{
			\pGamma \entail \ps t_1 = \ps t_3 \colon A
		}{
			\pGamma \entail \ps t_1 = \ps t_2 \colon A
			& \pGamma \entail \ps t_2 = \ps t_3 \colon A
		}
	\end{array}
	\]
	\caption{Basic equational rules.}
	\label{fig:eq-rules-base}
\end{figure}

\begin{figure}
	\[ \begin{array}{c}
		\infer[(perm)]{\pGamma \entail \Sigma_i (\palpha_i \pdot \ps t_i) = \Sigma_i
		(\palpha_{\ell(i)} \pdot \ps t_{\ell(i)}) \colon \pQ}{\ell \colon \set{1,\dots,n}
		\to \set{1,\dots,n}\text{ is a bijection} & \pGamma \entail \Sigma_i
		(\palpha_i \pdot \ps t_i) \colon \pQ}
		\mynl\\
		\infer[(0.scal)]{\pGamma \entail \Sigma_{i=1}^n (\palpha_i \pdot \ps t_i) =
		\Sigma_{i=1}^{n-1} (\palpha_i \pdot \ps t_i) \colon \pQ}{\pGamma \entail
		\Sigma_{i=1}^n (\palpha_i \pdot \ps t_i) \colon \pQ & \palpha_n = 0}
		\qquad
		\infer[(1.scal)]{\pGamma \entail \Sigma_{i=1}^1 (1 \pdot \ps t) = \ps t \colon
		\pQ}{\pGamma \entail \ps t \colon \pQ}
		\mynl\\
		\infer[(fubini)]{\pGamma \entail \Sigma_i \left( \palpha_i \pdot \Sigma_j (\beta_{ij}
		\pdot \ps t_{j}) \right) = \Sigma_j ( \left( \sum_i \palpha_i \beta_{ij} \right) \pdot \ps t_{j})
		\colon \pQ}{
			\pGamma \entail \Sigma_i \left( \palpha_i \pdot \Sigma_j (\beta_{ij} \pdot \ps t_{j})
			\right) \colon \pQ
		}
		\mynl\\
		\infer[(double)]{\pGamma \entail \Sigma_i \left( \palpha_i \pdot \Sigma_j (\beta_{ij}
		\pdot \ps t_{ij}) \right) = \Sigma_{ij} (\palpha_i \beta_{ij}\pdot \ps t_{ij})
		\colon \pQ}{
			\pGamma \entail \Sigma_{ij} (\palpha_i \beta_{ij}\pdot \ps t_{ij})
			\colon \pQ
		}
		\mynl\\
		\infer[(u.linear)]{\pGamma \entail u~\Sigma_i (\palpha_i \pdot
		\ps t_i) = \Sigma_i (\palpha_i \pdot u~\ps t_i) \colon \pQ_2}
		{
			\entailiso u \colon \isotypeonetwo
			&
			\pGamma \entail \Sigma_i (\palpha_i \pdot \ps t_i) \colon \pQ_1
		}
		\mynl\\
		\infer[(\iota.linear_1)]{\pGamma \entail \inl{} \Sigma_i (\palpha_i \pdot
		\ps t_i) = \Sigma_i (\palpha_i \pdot \inl{\ps t_i}) \colon \pQ_1 \oplus \pQ_2}{\pGamma \entail
		\Sigma_i (\palpha_i \pdot \ps t_i) \colon \pQ_1}
		\mynl\\
		\infer[(\iota.linear_2)]{\pGamma \entail \inr{} \Sigma_i (\palpha_i \pdot
		\ps t_i) = \Sigma_i (\palpha_i \pdot \inr{\ps t_i}) \colon \pQ_1 \oplus \pQ_2}{\pGamma \entail
		\Sigma_i (\palpha_i \pdot \ps t_i) \colon \pQ_2}
		\mynl\\
		\infer[(\otimes.linear_1)]{\pGamma_1,\pGamma_2 \entail \ps t \otimes (\Sigma_i
		(\palpha_i \pdot \ps t_i)) = \Sigma_i (\palpha_i \pdot \ps t \otimes \ps t_i) \colon \pQ_1
		\otimes \pQ_2}{\pGamma_1 \entail \ps t \colon \pQ_1 & \pGamma_2 \entail \Sigma_i
		(\palpha_i \pdot \ps t_i) \colon \pQ_2}
		\mynl\\
		\infer[(\otimes.linear_2)]{\pGamma_1,\pGamma_2 \entail (\Sigma_i (\palpha_i
		\pdot \ps t_i)) \otimes \ps t = \Sigma_i (\palpha_i \pdot \ps t_i \otimes \ps t) \colon \pQ_1
		\otimes \pQ_2}{\pGamma_1 \entail \ps t \colon \pQ_2 & \pGamma \entail \Sigma_i (\palpha_i
		\pdot \ps t_i) \colon \pQ_1}
		\mynl\\
		\infer[(S.linear)]{\pGamma \entail \psucc \Sigma_i (\palpha_i \pdot \ps t_i) =
		\Sigma_i (\palpha_i \pdot \psucc{\ps t_i}) \colon \pNat}{\pGamma \entail \Sigma_i
		(\palpha_i \pdot \ps t_i) \colon \pNat}
	\end{array}
	\]
	\caption{Vector space and linear applications equational rules.}
	\label{fig:eq-rules-vector-space}
\end{figure}

\begin{figure}
	\[ \begin{array}{c}
		\infer[(\iota.eq_1)]{
			\pGamma \entail \inl{\ps t_1} = \inl{\ps t_2} \colon \pQ_1 \oplus \pQ_2
		}{
			\pGamma \entail \ps t_1 = \ps t_2 \colon \pQ_1
		}
		\qquad
		\infer[(\iota.eq_2)]{
			\pGamma \entail \inr{\ps t_1} = \inr{\ps t_2} \colon \pQ_1 \oplus \pQ_2
		}{
			\pGamma \entail \ps t_1 = \ps t_2 \colon \pQ_2
		}
		\mynl\\
		\infer[(\otimes.eq_1)]{
			\pGamma, \pGamma' \entail \pair{\ps t_1}{\ps t} = \pair{\ps t_2}{\ps t} \colon \pQ_1 \otimes \pQ_2
		}{
			\pGamma \entail \ps t_1 = \ps t_2 \colon \pQ_1 & \pGamma' \entail \ps t \colon \pQ_2
		}
		\qquad
		\infer[(\otimes.eq_2)]{
			\pGamma, \pGamma' \entail \pair{\ps t}{\ps t_1} = \pair{\ps t}{\ps t_2} \colon \pQ_1 \otimes \pQ_2
		}{
			\pGamma \entail \ps t_1 = \ps t_2 \colon \pQ_2 & \pGamma' \entail \ps t \colon \pQ_1
		}
		\mynl\\
		\infer[(S.eq)]{
			\pGamma \entail \psucc{\ps t_1} = \psucc{\ps t_2} \colon \pNat
		}{
			\pGamma \entail \ps t_1 = \ps t_2 \colon \pNat
		}
		\qquad
		\infer[(u.eq)]{
			\pGamma \entail u~\ps t_1 = u~\ps t_2 \colon \pQ_2
		}{
			\pGamma \entail \ps t_1 = \ps t_2 \colon \pQ_1 & \entailiso u \colon \isotype{\pQ_1}{\pQ_2}
		}
		\mynl\\
		\infer[(\Sigma.eq)]{
			\pGamma \entail \Sigma_i (\palpha_i \pdot \ps t_i) 
			= \Sigma_i (\palpha_i \pdot \ps t'_i) \colon \pQ
		}{
			\pGamma \entail \Sigma_i (\palpha_i \pdot \ps t_i) \colon \pQ
			& \forall i, \pGamma \entail \ps t_i = \ps t'_i \colon \pQ
		}
	\end{array}
	\]
	\caption{Congruence equational rules of quantum control.}
	\label{fig:eq-rules-qu-control-cong}
\end{figure}

\begin{figure}
	\[ \begin{array}{c}
		\infer[(u.\beta)]{\entail \unibasique~\ps b' = \sigma(\ps e_i) \colon \pQ_2
		}{
			\entail \ps b' \colon \pQ_1 
			& \entailiso \unibasique \colon \isotype{\pQ_1}{\pQ_2} 
			& \match{\sigma}{\ps b_i}{\ps b'} 
		}
				\mynl\\
		\infer[(u.\otimes)]{
			\entail (u_1 \otimes u_2)~(\ps b_1 \ptimes \ps b_2) = (u_1~\ps b_1)
			\ptimes (u_2~\ps b_2) \colon \pQ_3 \otimes \pQ_4
		}{
			\entailiso u_1 \colon \isotype{\pQ_1}{\pQ_3} 
			& \entailiso u_2 \colon \isotype{\pQ_2}{\pQ_4} 
			& \entail \ps b_1 \colon \pQ_1 
			& \entail \ps b_2 \colon \pQ_2
		}
		\mynl\\
		\infer[(u.\oplus_1)]{
			\entail (u_1 \oplus u_2)~(\inl{\ps b}) = \inl (u_1~\ps b) \colon \pQ_3 \oplus \pQ_4
		}{
			\entailiso u_1 \colon \isotype{\pQ_1}{\pQ_3} 
			& \entailiso u_2 \colon \isotype{\pQ_2}{\pQ_4} 
			& \entail b \colon \pQ_1
		}
		\mynl\\
		\infer[(u.\oplus_2)]{
			\entail (u_1 \oplus u_2)~(\inr{\ps b}) = \inl (u_2~\ps b) \colon \pQ_3 \pplus \pQ_4
		}{
			\entailiso u_1 \colon \isotype{\pQ_1}{\pQ_3} 
			& \entailiso u_2 \colon \isotype{\pQ_2}{\pQ_4} 
			& \entail b \colon \pQ_2
		}
		\mynl\\
				\infer[(u.comp)]{
			\pGamma \entail (u_2 \circ u_1)~\ps b = u_2~(u_1~\ps b) \colon \pQ_3
		}{
			\entailiso u_1 \colon \isotype{\pQ_1}{\pQ_2} 
			& \entailiso u_2 \colon \isotype{\pQ_2}{\pQ_3} 
			& \pGamma \entail \ps b \colon \pQ_1
		}
		\qquad
		\infer[(u.inv)]{
			\entail u\inv~\ps v = \ps b \colon \pQ_1
		}{
			\entail u~\ps b = \ps v \colon \pQ_2
		}
		\mynl\\
		\infer[(u.ctrl_1)]{
			\entail (\uctrl u)~((\inl{\ps *}) \ptimes \ps b) = (\inl{\ps *})
			\ptimes \ps b \colon \pqbit \otimes \pQ
		}{
			\entailiso u \colon \isotype{\pQ}{\pQ}
			& \entail b \colon \pQ
		}
		\mynl\\
		\infer[(u.ctrl_2)]{
			\entail (\uctrl u)~((\inr{\ps *}) \ptimes \ps b) = (\inr{\ps *})
			\otimes (u~\ps b) \colon \pqbit \otimes \pQ
		}{
			\entailiso u \colon \isotype{\pQ}{\pQ}
			& \entail b \colon \pQ
		}
	\end{array}
	\]
	\caption{Computational equational rules of quantum control.}
	\label{fig:eq-rules-qu-control}
\end{figure}

\subsubsection{Equations and formation rules.}
We start by proving that the equational theory presented is sound with the
formation rules of the language. In other words, we show that if two terms are
equal in our theory, they are both well-formed. To do so, we need to show that
equality between terms preserve orthogonality.

\begin{lem}
	\label{lem:equal-ortho}
	Given three terms $\ps t_1$, $\ps t_2$, and $\ps t_3$, if $\ps t_1~\bot~\ps t_2$
	and $\pGamma \entail \ps t_1 = \ps t_3 \colon \pQ$ both hold, then $\ps
	t_2~\bot~\ps t_3$ holds.
\end{lem}
\begin{proof}
	By induction on the rules of the equational theory. 
\end{proof}

The proof of the next proposition heavily relies on the previous lemma. Indeed,
to prove that a linear combination is well-formed, one needs to prove that all
the terms involved in the linear combination are pairwise orthogonal. 

\begin{prop}
	\label{prop:equal-to-typed}
	If $\pGamma \entail \ps t_1 = \ps t_2 \colon \pQ$ is well-formed, then
	$\pGamma \entail \ps t_1 \colon \pQ$ and $\pGamma \entail \ps t_2 \colon \pQ$
	also are.
\end{prop}
\begin{proof}
	By induction on the rules of the equational theory. 
\end{proof}

\subsubsection{Bases.}
As expected, and thanks to the equational theory, an orthogonal decomposition
has properties similar to those of linear algebraic bases: in a Hilbert space,
any vector can be written as a combination of elements of a basis. We start by
proving that any expression is equal to a combination of basis values.

\begin{lem}
	\label{lem:values-are-combinations}
	Given a well-formed closed expression $\entail \ps e \colon \pQ$, there
	exist  a set of indices $I$, a family of basis values $(\ps b_i)_{i \in I}$, and
	 a family of complex numbers $(\palpha_i)_{i \in I}$, such that $~\entail
	\ps e = \Sigma_i (\palpha_i \pdot \ps b_i) \colon \pQ$ holds.
\end{lem}
\begin{proof}
	The proof is done by induction on $~\entail \ps e \colon \pQ$.
	\begin{itemize}
		\item $~\entail \ps * \colon \pbasic$. Nothing to do.
		\item $~\entail \inl{\ps e} \colon \pQ_1 \pplus \pQ_2$. The induction hypothesis
			gives $~\entail \ps e = \Sigma_i (\palpha_i \pdot \ps b_i) \colon \pQ_1$,
			and observe that $~\entail \inl{} \Sigma_i (\palpha_i \pdot \ps b_i)
			= \Sigma_i (\palpha_i \pdot \inl{} \ps b_i) \colon \pQ_1 \pplus \pQ_2$.
		\item $~\entail \inr{} \ps e \colon \pQ_1 \pplus \pQ_2$ has the same conclusion.
		\item $~\entail \ps e_1 \ptimes \ps e_2 \colon \pQ_1 \ptimes \pQ_2$.
			The induction hypothesis gives $I$, $(\ps b^1_i)$ and
			$(\palpha_i)$, as well as $J$, $(\ps b^2_j)$ and $(\pbeta_j)$, such that
			$~\entail \ps e_1 = \Sigma_i (\palpha_i \pdot \ps b^1_i) \colon \pQ_1$
			and $~\entail \ps e_2 = \Sigma_j (\pbeta_j \pdot \ps b^2_j) \colon
			\pQ_2$. Thus, we have that 
 			\begin{align*}
 				\entail \ps e_1 \ptimes \ps e_2 				&= \left( \Sigma_i (\palpha_i \pdot \ps b^1_i) \right) \ptimes
				\left( \Sigma_j (\pbeta_j \pdot \ps b^2_j) \right) \colon \pQ_1
				\ptimes \pQ_2 
				&\text{(by IH)} \\
 				&= \Sigma_i \left( \palpha_i \pdot \ps b^1_i \ptimes \left( \Sigma_j
 				(\pbeta_j \pdot b^2_j) \right) \right) \colon \pQ_1 \ptimes \pQ_2 
 				&(\text{by Rule }\otimes.linear_2) \\
 				&= \Sigma_i \left( \palpha_i \pdot \Sigma_j ( \pbeta_j \pdot 
 				\ps b^1_i \ptimes \ps b^2_j ) \right) \colon \pQ_1 \ptimes \pQ_2 
 				& (\text{by Rule }\otimes.linear_1) \\
				&= \Sigma_{ij} (\palpha_i \pbeta_j \pdot \ps b^1_i \ptimes
				b^2_j) \colon \pQ_1 \ptimes \pQ_2.
 				& (\text{by Rule } double)
 			\end{align*}
		\item $~\entail \psucc{} \ps e \colon \pNat$ is treated similarly to $\inr{}\!$.
		\item $~\entail \Sigma_i (\palpha_i \pdot \ps e_i) \colon \pQ$. The
			induction hypothesis gives $(\pbeta_{ij})$ and $(\ps b_j)$ (the
			$\ps b$ does not depend on $i$ without loss of generality, because
			$0 \pdot \ps b$ can be added to any sum term $\ps t$, as long as
			$\ps b$ is orthogonal to $\ps t$). Finally, we have $~\entail
			\Sigma_i (\palpha_i \pdot \Sigma_j (\pbeta_{ij} \pdot \ps b_j)) =
			\Sigma_j ((\sum_i \palpha_i \pbeta_{ij}) \pdot \ps b_j) \colon
			\pQ$. \qedhere
	\end{itemize}
\end{proof}

\begin{rem}
	The resulting term in the previous lemma, written $\Sigma_i (\palpha_i
	\pdot \ps b_i)$ is a value if the basis values are correctly ordered and if
	all the scalar are non zero. Thanks to the (perm) and ($0$.scal) rules in
	Fig.~\ref{fig:eq-rules-vector-space}, we can assume so. Thus, the lemma
	above shows that expressions have a unique normal form.
\end{rem}

We can then generalize to any orthogonal decomposition, with the help of
substitutions.

\begin{lem}
	\label{lem:od-output-exhaustive}
	Assume that $\ODe{\pQ}{(S)}$ holds and that all elements of $S$ are well-formed. Given a
	well-formed closed expression $~\entail \ps e \colon \pQ$, there exist 
	a set of indices $I$, a family $(\ps s_i)_{i \in I}$ of elements of $S$,
	a family of complex numbers $(\palpha_i)_{i \in I}$, and  a family of valuations $(\sigma_i)_{i \in
	I}$ such that $~\entail \ps e = \Sigma_i (\palpha_i
	\pdot \sigma_i(\ps s_i)) \colon \pQ$ holds.
\end{lem}
\begin{proof}
	This is proven by induction on ${\rm ONB}$. The previous lemma gives a term
	equal to $\ps e$ in the equational theory written as a finite sum of basis
	values $\Sigma_i (\palpha_i \pdot \ps b_i)$.
	\begin{itemize}
		\item $\ODe{\pbasic}{(\set{\ps *})}$. Nothing to do.
		\item $\ODe{\pQ}{(\set{\ps x})}$. The substitution $\sigma = \set{\ps x
			\mapsto \ps e}$ is suitable with $\ps e = \sigma(\ps x)$.
		\item $\ODe{\pQ_1 \pplus \pQ_2}{(S \boxplus T)}$. Each $\ps b_i$ is
			either $\inl{} \ps b'_i$, with gives a suitable substitution
			$\sigma_i$ and $\ps s_i \in S$, thus $\inl{} \ps s_i \in S \boxplus
			T$, or $\ps b_i$ is $\inr{} \ps b'_i$, giving suitable substitution
			$\sigma_i$ and $\ps s_i \in T$, thus $\inr{} \ps s_i \in S \boxplus
			T$; all this by induction hypothesis.
		\item $\ODe{\pQ_1 \ptimes \pQ_2}{(S)}$. Each $\ps b_i$ is of the form
			$\ps b'_i \ptimes \ps b''_i$, the induction hypothesis gives
			suitable $\sigma'_i$, $\ps s'_i$, $\sigma''_i$, $\ps s''_i$, that
			can be assembled into $\sigma_i = \sigma'_i
			\cup \sigma''_i$ and $\ps s_i = \ps s'_i \ptimes \ps s''_i$.
		\item $\ODe{\pNat}{(S^{\oplus 0})}$. Each $\ps b_i$ is either $\pzero$,
			for which there is nothing to do, or $\psucc{} \ps b'_i$, in which
			case the induction hypothesis concludes.
		\item $\ODe{\pQ_2}{(S^\beta)}$. First, we show that each $\ps s_i \in S$ can
			be written as a linear combination of elements of $S^\beta$. Indeed,
			in the equational theory: 
			\begin{align*}
				\Sigma_{\ps s \in S} (\overline\pbeta_{\ps s,\ps s_i} \pdot
				\Sigma_{\ps s' \in S} (\pbeta_{\ps s,\ps s'} \pdot \ps s')) 
				&= \Sigma_{\ps s' \in S} \left(\sum_{\ps s \in S} \overline\pbeta_{\ps s,\ps s_i}
				\pbeta_{\ps s,\ps s'} \right) \pdot \ps s' \\
				&= \Sigma_{\ps s' \in S} (\delta_{\ps s' = \ps s_i} \pdot \ps s') \\
				&= \ps s_i.
			\end{align*}
			The conclusion is then direct. \qedhere
	\end{itemize}
\end{proof}

\subsubsection{Progress.}
The following lemma is our equivalent to progress, and involves both
directions: the application of a unitary to a value reduces to a value, and
given a unitary and a value, there exists a value that is the inverse image of
the latter. 

\begin{lem}
	\label{lem:equational-progress}
	If $~\entailiso u \colon \isotypeonetwo$ holds, then we have that:
	\begin{itemize}
		\item for all $~\entail \ps e \colon \pQ_1$, there exists $~\entail \ps
			v \colon \pQ_2$ such that $~\entail u~\ps e = \ps v \colon \pQ_2$ holds;
		\item for all $~\entail \ps e \colon \pQ_2$, there exists $~\entail \ps
			v \colon \pQ_1$ such that $~\entail u~\ps v = \ps e \colon \pQ_2$ holds.
	\end{itemize}
\end{lem}
\begin{proof}
	First, observe that terms are taken up to vector space identities, and up
	to linearity of unitaries. Thus, if $\ps e = \Sigma_i (\palpha_i \pdot \ps
	b_i)$, we have $u~\ps e = \Sigma_i (\palpha_i \cdot u~\ps b_i)$, and this
	shows that we can work with $u~\ps b$ with $\ps b$ a basis value without
	loss of generality. The equivalent statement to the one we prove is: given
	$~\entailiso u \colon \isotypeonetwo$ and $~\entail \ps b \colon \pQ_1$ a
	basis value, there exists $~\entail \ps v \colon \pQ_2$ such that $~\entail
	u~\ps b = \ps v \colon \pQ_1$. A similar observation can be done for $u$.
	This is proven by induction on the judgment $~\entailiso u \colon \isotypeonetwo$.
	\begin{itemize}
		\item If $~\entailiso \unibasique \colon \isotypeonetwo$.
			Lemma~\ref{lem:od-substitution} provides a $\ps b_i$ that matches with $\ps b$
			and a substitution $\sigma$ that fits this match. Thus, we have $~\entail
			u~\ps b = \sigma(\ps v_i) \colon \pQ_2$ with $~\entail \sigma(\ps v_i) \colon \pQ_2$. On
			the other hand, Lemma~\ref{lem:od-output-exhaustive} provides three
			families $(\ps v'_j)$ that can have multiple copies of the $(\ps v_i)$,
			$(\palpha_j)$ and $\sigma_j$ such that $u = \Sigma_j (\palpha_j \pdot
			\sigma_j(\ps v'_j))$, and note that $~\entail \unibasique~\Sigma_j (\palpha_j
			\pdot \sigma_j(\ps b'_j)) = \Sigma_j (\palpha_j \pdot \sigma_j(\ps v_j)) \colon \pQ_2$
			where $\ps b'_j$ is a $\ps b_i$ when $\ps v'_j$ is $\ps v_i$; and this is because given a
			substitution $\sigma$ that fits a $\ps b_i$, $~\entail
			\unibasique~\sigma(\ps b_i) = \sigma(\ps v_i) \colon \pQ_2$.
		\item If $~\entailiso u_1 \oplus u_2 \colon \isotype{\pQ_1 \pplus
			\pQ'_1}{\pQ_2 \pplus \pQ'_2}$. There are two similar cases, namely
			$\inl{} \ps b$ and $\inr{} \ps b$. In the first case, the induction
			hypothesis gives $\ps v_1$ such that $~\entail u_1~\ps b = \ps v_1
			\colon \pQ_2$, thus $~\entail (u_1 \oplus u_2)~(\inl{} \ps b) =
			\inl{} \ps v_1 \colon \pQ_2 \pplus \pQ'_2$. The other case is
			similar. A related reasoning proves the second point.
		\item If $~\entailiso u_1 \otimes u_2 \colon \isotype{\pQ_1 \ptimes
			\pQ'_1}{\pQ_2 \ptimes \pQ'_2}$. We write $\ps b_1 \ptimes \ps b_2$
			for $\ps b$, and the induction hypothesis provides $\ps v_1$ and
			$\ps v_2$ such that $~\entail u_1~\ps b_1 = \ps v_1 \colon \pQ_2$
			and $~\entail u_2~\ps b_2 = \ps v_2 \colon \pQ'_2$, ensuring that
			$~\entail (u_1 \otimes u_2)~(\ps b_1 \ptimes \ps b_2) = \ps v_1
			\ptimes \ps v_2 \colon \pQ_2 \ptimes \pQ'_2$. A related reasoning
			proves the second point.
		\item If $~\entailiso u_2 \circ u_1 \colon \isotype{\pQ_1}{\pQ_3}$. The
			induction hypothesis gives us $\ps v_1$ such that $~\entail u_1~\ps
			b = \ps v_1 \colon \pQ_2$ and then $\ps v_2$ such that $~\entail
			u_2~\ps v_1 = \ps v_2 \colon \pQ_3$, which ensures the result. A
			related reasoning proves the second point.
		\item If $~\entailiso u\inv \colon \isotypetwoone$. Given $~\entail \ps v \colon
			\pQ_2$, the induction hypothesis gives $~\entail \ps v' \colon \pQ_1$ such that
			$~\entail u~\ps v = \ps v' \colon \pQ_2$, and thus $~\entail u\inv~\ps v' = \ps v
			\colon \pQ_1$. Moreover, given $~\entail \ps w \colon \pQ_1$, the induction
			hypothesis gives $~\entail \ps w' \colon \pQ_2$ such that $~\entail u~\ps w' = \ps w
			\colon \pQ_2$, and thus $~\entail u\inv~\ps w = \ps w' \colon \pQ_1$.
		\item If $~\entailiso \uctrl u \colon \isotype{\pqbit \ptimes \pQ}{\pqbit
			\ptimes \pQ}$. In the case $(\inl{} \ps *) \ptimes \ps b$, there is
			nothing to do. The other case is $(\inr{} \ps *) \ptimes \ps b$,
			and the induction hypothesis gives $\ps v$ such that $~\entail
			u~\ps b = \ps v \colon \pQ$, and then $~\entail (\uctrl u)~((\inr{}
			\ps *) \ptimes \ps b) = (\inr{} \ps *) \ptimes \ps v \colon \pqbit
			\ptimes \pQ$. A related reasoning proves the second point. \qedhere
	\end{itemize}
\end{proof}

\subsubsection{Normalization.}
We have proven that unitary applications progress and reduce to values, if one
wishes to have an operational point of view. This allows use to prove, with the
same operational view, that the system strongly normalizes; this means that any
term $\ps t$ eventually reduces to a value.

\begin{lem}
	\label{lem:equational-sn}
	If $~\entail \ps t \colon \pQ$ holds, then there exists $~\entail \ps v \colon \pQ$
	such that $~\entail \ps t = \ps v \colon \pQ$ holds.
\end{lem}
\begin{proof}
	This is proven by induction on the formation rules of $~\entail \ps t \colon \pQ$.
	\begin{itemize}
		\item The cases $\ps *$, $\ps x$, $\pzero$ and sum are straightforward.
		\item In the case $\ini{} \ps t$, the induction hypothesis provides
			$\ps v$ such that $~\entail \ps t = \ps v \colon \pQ_i$, thus
			$~\entail \inx{} \ps t = \inx{} \ps v \colon \pQ_1 \pplus \pQ_2$.
		\item In the case $\ps t_1 \ptimes \ps t_2$, the induction hypothesis
			gives corresponding $\ps v_1$ and $\ps v_2$, that ensure $~\entail
			\ps t_1 \ptimes \ps t_2 = \ps v_1 \ptimes \ps v_2 \colon \pQ_1
			\ptimes \pQ_2$.
		\item In the case $\psucc{} \ps t$, the induction hypothesis provides
			$\ps v$ such that $~\entail \ps t = \ps v \colon \pNat$, thus
			$~\entail \psucc{} \ps t = \psucc{} \ps v \colon \pNat$.
		\item In the case $u~\ps t$, the induction hypothesis gives $\ps v$
			such that $~\entail \ps t = \ps v \colon \pQ_1$. The previous
			lemma, Lemma~\ref{lem:equational-progress}, provides $\ps v'$ such
			that $~\entail u~\ps v = \ps v' \colon \pQ_2$, thus $~\entail \ps t
			= \ps v' \colon \pQ_2$. \qedhere
	\end{itemize}
\end{proof}

\subsubsection{Completeness.}
Finally, we prove a strong link between the denotational semantics and the
equational theory, namely that an equality statement in one is also an equality
statement in the other. We start by showing soundness, meaning that two terms
equal in the equational theory have the same denotational interpretation.

\begin{prop}
	\label{prop:qua-soundness}
	If $\pGamma \entail \ps t_1 = \ps t_2 \colon \pQ$ holds, then $\sem{\pGamma
	\entail \ps t_1 \colon \pQ} = \sem{\pGamma \entail \ps t_2 \colon \pQ}$.
\end{prop}
\begin{proof}
	By induction on the rules of the equational theory. The only non trivial
	case was done within Proposition~\ref{prop:qua-op-soundness}.
\end{proof}

We prove then completeness, starting with terms without unitary functions.

\begin{lem}[Completeness of values]
	\label{lem:value-completeness}
	Given two well-formed values $~\entail \ps v_1 \colon \pQ$ and $~\entail \ps v_2 \colon \pQ$, if
	$\sem{\entail \ps v_1 \colon \pQ} = \sem{\entail \ps v_2 \colon \pQ}$, then
	$~\entail \ps v_1 = \ps v_2 \colon \pQ$ holds.
\end{lem}
\begin{proof}
	This is proven by induction on $~\entail \ps v_1 \colon \pQ$.
	\begin{itemize}
		\item The cases for $\ps *$ and $\pzero$ are straightforward.
 		\item If $\ps v_1 = \ini{} \ps b_1$ of type $\pQ_1 \pplus \pQ_2$, with
			$\sem{\ps v_1} = \sem{\ps v_2}$, necessarily $\ps v_2 = \ini{} \ps
			b_2$, and $\iota^{\pQ_1,\pQ_2}_l \circ \sem{\ps b_1} =
			\iota^{\pQ_1,\pQ_2}_l \circ \sem{\ps b_2}$, thus
			$(\iota^{\pQ_1,\pQ_2}_l)\inv \circ \iota^{\pQ_1,\pQ_2}_l \circ
			\sem{\ps b_1} = (\iota^{\pQ_1,\pQ_2}_l)\inv \circ
			\iota^{\pQ_1,\pQ_2}_l \circ \sem{\ps b_2}$ which ends with
			$\sem{\ps b_1} = \sem{\ps b_2}$, and the induction hypothesis gives
			that $~\entail \ps b_1 = \ps b_2 \colon \pQ_i$, and thus $~\entail
			\ini b_1 = \ini b_2 \colon \pQ_1 \pplus \pQ_2$.
		\item If $\ps v_1 = \ps b_1 \ptimes \ps b'_1$ with type $\pQ_2 \otimes
			\pQ_2$, then also $\ps v_2 = \ps b_2 \otimes \ps b'_2$, and
			$\sem{\ps b_1} \otimes \sem{b'_1} = \sem{\ps b_2} \otimes
			\sem{b'_2}$, thus $\sem{\ps b_1} = \sem{\ps b_2}$ and $\sem{b'_1} =
			\sem{b'_2}$, the induction hypothesis ensures that $~\entail \ps
			b_1 = \ps b_2 \colon \pQ_1$ and $~\entail \ps b'_1 = \ps b'_2
			\colon \pQ_2$, and thus $~\entail \ps b_1 \ptimes \ps b'_1 = \ps
			b_2 \ptimes \ps b'_2 \colon \pQ_1 \ptimes \pQ_2$.
		\item Else, with type $\pQ$, $\ps v_1 = \Sigma_i (\palpha_i \pdot \ps b^1_i)$ with the
			$(\ps b^1_i)$ that are pairwise orthogonal, and $\ps v_2 = \Sigma_j (\pbeta_j
			\pdot \ps b^2_j)$ with the $(\ps b^2_j)$ that are also pairwise orthogonal. We
			know that \[ \sem{\ps v_1} = \sum_i \palpha_i \sem{\ps b^1_i} = \sum_j \pbeta_j
			\sem{\ps b^2_j} = \sem{\ps v_2}. \]
			Thus, for all $~\entail \ps b \colon \pQ$, $\sem{\ps b}\inv \circ \sum_i \palpha_i
			\sem{b^1_i} = \sem{\ps b} \inv \circ \sum_j \pbeta_j \sem{\ps b^2_j}$; by
			orthogonality, we have $\sem{\ps b}\inv \circ \sum_i \palpha_i \sem{\ps b^1_i} =
			\palpha_k \sem{\ps b^1_k}$ for some $k$, and $\sem{\ps b}\inv \circ \sum_j \pbeta_j
			\sem{\ps b^2_j} = \pbeta_{k'} \sem{\ps b^2_{k'}}$ for some $k'$. Thus, $\palpha_k
			\sem{\ps b^1_k} = \pbeta_{k'} \sem{\ps b^2_{k'}}$, and because they are basis
			values, we have $\sem{\ps b^1_k} = \sem{\ps b^2_{k'}}$, and the induction
			hypothesis ensures that $~\entail \ps b^1_k = \ps b^2_{k'} \colon \pQ$. Note that
			this is done for all $~\entail \ps b \colon \pQ$, and thus $~\entail \ps v_1 = \ps v_2
			\colon \pQ$. \qedhere
	\end{itemize}
\end{proof}

We then prove the final result of this section, namely, completeness on closed
terms, meaning that the equality statement of the equational theory and of the
denotational semantics on closed terms are equivalent.

\begin{thm}[Completeness]
	\label{th:complete}
	\(
		~\entail \ps t_1 = \ps t_2 \colon \pQ 
		\text{ iff } 
		\sem{\entail \ps t_1 \colon \pQ} = \sem{\entail \ps t_2 \colon \pQ}.
	\)
\end{thm}
\begin{proof}
	We prove both directions.
	\begin{itemize}
		\item For the implication $~\entail \ps t_1 = \ps t_2 \colon \pQ$ to $\sem{\entail
			\ps t_1 \colon \pQ} = \sem{\entail \ps t_2 \colon \pQ}$, see
			Proposition~\ref{prop:qua-soundness}.
		\item The other direction uses Lemma~\ref{lem:equational-sn}, the
			equivalent of strong normalization, that gives $\ps v_1$ and $\ps
			v_2$ such that $~\entail \ps t_1 = \ps v_2 \colon \pQ$ and
			$~\entail \ps t_2 = \ps v_2 \colon \pQ$. Our hypothesis is that
			$\sem{\ps t_1} = \sem{\ps t_2}$, and thus $\sem{\ps v_1} = \sem{\ps
			t_1} = \sem{\ps t_2} = \sem{\ps v_2}$.
			Lemma~\ref{lem:value-completeness} gives then that $~\entail \ps
			v_1 = \ps v_2 \colon \pQ$ and transitivity ensures that $~\entail
			\ps t_1 = \ps t_2 \colon \pQ$. \qedhere
	\end{itemize}
\end{proof}

%% file: combining_qc.tex
This section lays out the syntax and semantics of the main calculus, which
combines the quantum control fragment with a classically controlled layer, a
variant of a call-by-value linear lambda-calculus.  The subsystem which we
presented in the previous section is designed to represent pure quantum
information and computation, whereas in this section, we are concerned with
mixed quantum information and computation.  Our denotational approach towards
this uses von Neumann algebras, which provide an appropriate mathematical
setting to achieve this treatment. Operationally, this is achieved by adapting the
configurations (also called closures) $(\ket \psi, \ell, M)$ from an effectful quantum lambda-calculus, \emph{e.g.}, \cite{kenta-bram,selinger2009quantum},
by replacing the quantum state $\ket \psi$ with a pure term $\pt$, replacing
the linking function $\ell$ with a unitary $u$ that can be described via our
syntax and which represents a suitable permutation, and finally by defining suitable
formation conditions for the configuration. 

We refer to this
system as the ``main calculus'', because it allows us to integrate the pure
quantum subsystem within it, but it also allows us to represent classical
information, quantum measurements and mixed-state quantum computation.

\subsection{Syntax and Typing Rules}
\label{sec:syntax}
\subsubsection{Types.}
The \emph{types} of the main calculus, ranged over by capital Latin letters
(\emph{e.g.}, $\mathrm A, \mathrm B$) are given in Figure~\ref{fig:syntax-c}.
We write $I$ for the unit type, $\mathrm A + \mathrm B$ for sum types
(which allow us to handle classical control), $\mathrm A \otimes \mathrm B$ for pair types, $\mathrm A
\multimap \mathrm B$ for higher-order linear function types, and $!\mathrm A$ for exponential types
(in the sense of linear logic).
All types except $\Nat$ and $\mathcal B(\ptype)$ are standard in a system for a classical call-by-value
linear lambda-calculus. $\Nat$ is a ground type representing classical natural
numbers and $\mathcal B(\ptype)$ is the type that represents
\emph{mixed state quantum} computation on the Hilbert space determined by the type
$\pQ$. As we see later, this type plays an important role by
introducing quantum and probabilistic effects into the system. 
The modality $\mathcal B$ is inspired from the denotational model, where $\mathcal B(H)$ 
represents the von Neumann algebra of bounded operators on the Hilbert space $H$. 

\input{syntax-c}
\input{typing-rules}

\subsubsection{Terms.}
\emph{Terms} (and values) of the calculus are described in Figure~\ref{fig:syntax-c}.
They are either standard linear lambda-calculus terms, terms representing natural numbers, 
or terms that we introduce to model the interaction of the classical control fragment with the quantum control fragment. 
 We refer the reader to \cite{selinger2009quantum} for a detailed exposition on linear lambda-calculus.
$zero$ and $succ$ are constructors for natural numbers and $\tmatch{L}{M}{N}$ is the corresponding pattern-match.
The term $\tpure{\pt}$ models the preparation of a pure quantum state represented by the quantum term $\pt$. 
Measurement is modeled by the term $meas(M)$. 
The term $\tun{\isoterm}{M}$ represents the application of the unitary $\isoterm$ to a term $M$.
Additionally, we have two syntactic constructs to model the isomorphism between the types $\mathcal B(\ptype_1 \ps \otimes \ptype_2)$ and $\mathcal B(\ptype_1) \otimes \mathcal B(\ptype_2)$:
the term $\tletb{z}{M}{N}$ allows one to construct a term of type $\mathcal B(\ptype_1 \ps \otimes \ptype_2)$ from a term of type $\mathcal B(\ptype_1) \otimes \mathcal B(\ptype_2)$ and
the term $\tletpairb{x}{y}{M}{N}$ deals with the opposite direction. 
This isomorphism arises because the functor $\mathcal B(\cdot)$ used in the denotational model is \emph{strict monoidal}. 
Both constructs are necessary to allow the language to manipulate composed quantum systems and also terms possibly representing \emph{entangled states}.

\subsubsection{Type system.}
\emph{Typing contexts} (ranged over by symbols $\Delta,\Sigma_1, \Sigma_2$) are finite sequences $\Delta \defeq x_1: \mathrm A_1, \ldots, x_n : \mathrm A_n$ mapping variables to types. 
We use the notation $!\Delta$ for contexts of the form $x_1:\ !\mathrm A_1, \ldots, x_n:\ !\mathrm A_n.$ 
\emph{Typing judgments} have the form $\Delta \vdash M : \mathrm A$ and follow a linear typing discipline to deal with quantum data. 
The typing rules are given in Figure~\ref{fig:typing-rules-excerpt} where, as usual, $\Sigma_1,\Sigma_2$ denotes the disjoint union of $\Sigma_1$ and $\Sigma_2$. 
Classical bits correspond to the type $\textit{bit} \eqdef I + I$ and mixed state qubits have type $\qbit \eqdef \mathcal B(\pqbit)$.
Note that we interchangeably use
contexts of the form $\Delta$ and $!\Delta, \Sigma_1$. The former is used to lay out rules which do not require a distinction between a \textit{linear} typing context and a non-linear one.
We use the latter when a rule modifies the context, and it is necessary to differentiate between the \emph{linear} part of the typing context and the \emph{non-linear} part. 
The rule for the term $\tpure{\ps t}$ introduces a state $\ps t$ from the quantum control fragment as a term of 
type $\mathcal B(\ps Q)$ into the main calculus. The typing rule for measurement maps a term $M$ of quantum type $\mathcal B(\ptype)$ to a term $\tmeas{M}$ of classical type $\ov \ptype$, where
the type $\ov \ptype$ is defined inductively on the structure of the pure quantum type $\ptype$: 
\begin{center}
	$\ov \pbasic \defeq I,\quad  \ov {\ptypeone \ptimes \ptypetwo} \defeq \ov
	\ptypeone \otimes \ov \ptypetwo,\quad \ov{\ptypeone \pplus \ptypetwo}
	\defeq \ov \ptypeone + \ov \ptypetwo, \quad \ov \pnat \defeq \Nat .$ 
\end{center}
This gives the classical analogue of the type $\pQ$, \emph{e.g.}, $\ov{\pqbit} = \textit{bit}.$
For a quantum basis term $\pb$ of type $\ptype$, $\ov{\pb}$ (defined below) denotes its translation to
a classical value of type $\ov \ptype$, which we use to represent measurement outcomes.
\begin{center}
	$\ov{ \ps \ast} \defeq \ast,\quad \ov{\ps b_1 \ps \otimes \ps b_2} \defeq
	\ov{\ps b_1} \otimes \ov{\ps b_2},\quad \ov{\inl {\ps b}} \defeq
	\tinl{\ov{\ps b}},\quad \ov \pzero  \defeq \tzero,\quad \ov{\psucc {\ps b}}
	\defeq \tsucc{\ov{\ps b}}.$
\end{center}

\subsection{Operational Semantics}\label{sub:operational-c}
\subsubsection{Quantum configurations.}
We lay out the operational semantics of the calculus using an adaptation of the notion of a \emph{quantum configuration}
from quantum lambda-calculi (\cite{kenta-bram,selinger2009quantum}).

\begin{defi}[Quantum Configuration]
\label{def:quantum-config}
	Let $\confs$ be the set of  \emph{quantum configurations} $\mathcal C \triangleq (\pt , \isoterm_{\sigma} , M),$ where:
	\begin{itemize}
		\item $\pt \in \pterms$ represents a \emph{pure quantum state};
		\item $\isoterm_{\sigma} \in \punitaries$ represents a symmetric monoidal isomorphism;
		\item $M \in \cterms$ is a term from the main calculus.
	\end{itemize}

\end{defi}
In the above definition, we can think of the quantum term $\pt$ as representing a quantum state, because the
quantum control fragment ensures that $\ps t$ can be rewritten (in its equational theory) to a quantum value~(Lemma~\ref{lem:equational-sn}). 
Next, $\isoterm_{\sigma}$ permutes the components of $\pt$ so that they appear in the same order as variables that represent
them in $M$. We slightly abuse terminology and refer to $u_{\sigma}$ as a permutation, but it is meant to be understood as a canonical symmetric
monoidal isomorphism (see Proposition~\ref{prop:monoidal-unitary}). If we fix the domain and codomain of $\isoterm_\sigma$, then the essential data is given by a choice of permutation $\sigma$ and this
makes it easier to understand how each unitary acts on the coordinates, and the contextual rules formulation becomes precise. Note that, thanks to Proposition \ref{prop:monoidal-unitary}, the
pure quantum syntax is expressive enough to represent all the monoidal
permutations that we need in the sequel.
\begin{exa}
  Consider the quantum configuration $(\ps t , \isoterm_{id}, x \otimes y)$, where we have
  $\ps t \eqdef \ps \ast \ptimes  ( \sfrac{1}{\sqrt2} \psdot \pket 0 \ptimes \pket 0 + \sfrac{1}{\sqrt 2} \psdot \pket 1 \ptimes \pket 1)$ and
  $\isoterm_{id} \eqdef \{ \mid \ps{\ast \otimes w \otimes z} \iso \ps{w \otimes z} \}$ 
  is a unitary with type $\udash \isoterm_{id} \colon U(\ps I \ptimes \pqbit \ptimes \pqbit, \pqbit \ptimes \pqbit)$. It conveys the information 
  that the variable $x$ in $x \otimes y$ keeps track of the first qubit in $\ps t$ and that the variable $y$ keeps track of the second one. Whereas in 
  the configuration $(\ps t , \isoterm_{\emph{swap}}, x \otimes y)$ with $\isoterm_{\emph{swap}} \eqdef \{ \mid \ps{\ast \otimes w \otimes z} \iso \ps{z \otimes w} \}$, 
  the variable $x$ keeps track of the second qubit and $y$ of the first one. Note that, due to quantum entanglement, the term $\ps t$ cannot 
  be decomposed into a non-trivial tensor product, and the variables $x$ and $y$ should be thought of as identifying qubit components of $\ps t$, 
  instead of storing quantum data themselves.
\end{exa}
A quantum configuration $\mathcal V = (\pv , \isoterm_{\sigma} , V)$, when both the quantum term $\pv$ and term $V$ are values, is called a \emph{value configuration}. 
For example, $(\ket 0 , \isoterm_{id} , x)$ is a value configuration.
A quantum configuration $\mathcal{SV}= (\pv , \isoterm_{\sigma} , M)$, where $\pv$ is a quantum value, is a \emph{semi-value configuration}.  We write $\vconfs$ (resp. $\svconfs$) for the set of (resp. semi-)value configurations.
The reason we introduce semi-value configurations, is because they allow us to define our operational semantics for the main calculus modulo equality of the pure quantum terms $\pt$ (analogous to the quantum lambda-calculus). We already proved that values $\pv$ are normal forms for the pure terms~(Lemma~\ref{lem:equational-sn}), so this makes them a natural choice.

\begin{defi}
\label{def:well-formed-configuration}
A quantum configuration $(\pt , \isoterm_{\sigma} , M)$ is \emph{well-formed} with type $\mathrm{A}$, written as $ ( \pt , \isoterm_{\sigma} , M ) \colon \mathrm{A}$,
if the following judgments can be derived:
\begin{itemize}
	\item $\cdot \pdash \ps t \colon \ptype_1 \ptimes {\ps \cdots} \ptimes \ptype_n$;
       \item $\cdot \udash \isoterm_{\sigma} \colon U(\ptype_1 \ptimes {\ps \cdots} \ptimes \ptype_n , \ptype_1' \ptimes {\ps \cdots} \ptimes \ptype_m')$;
       \item 
             $x_1 : \mathcal B(\ptype_1') , \ldots , 
               x_m : \mathcal B(\ptype_m') \vdash M : \mathrm{A}$.
\end{itemize}
\end{defi}
 
We write $\wfconfs{X}{A}$ for the set of well-formed configurations  in $X$ with type $\mathrm{A}$, where $X \in \{ \confs, \vconfs, \svconfs \}$. 

\begin{exa}
The configuration $((\ps{\sfrac{1}{\sqrt{2}} \psdot \ket{00} + \sfrac{1}{\sqrt{2}} \psdot  \ket{11}})\ptimes \ps{\ket{0}} , \isoterm_{\id} , x\otimes y)$ with:
\begin{itemize}
	\item $\cdot \pdash (\ps{\sfrac{1}{\sqrt{2}} \psdot\ket{00} + \sfrac{1}{\sqrt{2}} \psdot \ket{11}}) \ptimes \pket{0} \ : \ptqbit{3}$;
	\item $\cdot \udash \isoterm_{\id} \colon U(\ptqbit{3}, \ptqbit{3})$;
	\item $x: \mathcal B(\ptqbit{2}) , y : \qbit \vdash x \otimes y: \mathcal B(\ptqbit{2}) \otimes \qbit$.
\end{itemize}
In this configuration, $u_\id$ is just the identity, $x$ points to the Bell state in $\pt$ whereas $y$ points to $\ps{\ket 0}$.
We have $n=3,m=2$, with $\ptype_1 = \ptype_2 = \ptype_3 = \pqbit$, $\ptype_1'= \ptqbit{2}$ and $ \ptype_2'= \pqbit$.
The configuration is well-formed.
\end{exa}

The above example highlights that different values for $m$ and $n$
(from Definition \ref{def:well-formed-configuration}) should be admissible. 
The intuition behind this difference is that the unitary $u_{\sigma}$ partitions the $n$ quantum types into $m$ blocks,
where each block is represented by a variable in $M$.
Furthermore, in situations where we wish to combine two such
blocks which are not next to each other in the tensor expression of the
$n$ types, we allow the possibility of permuting these blocks and merging
them so that now two blocks that were initially represented by two variables are represented by one. This is formally expressed in
the reduction rules we define next. The role of $u_{\sigma}$ in a well-formed configuration is illustrated in Figure~\ref{fig:well-formed}. The  
\emph{small-step reduction} $\cdot \rightarrow_{\cdot} \cdot \subseteq \confs \times [0,1] \times \confs$ is the relation
defined by the rules of Figure~\ref{fig:red_uni}
: $\mathcal C \rightarrow_p  \mathcal C'$ holds when the
quantum configuration $\mathcal C$ reduces to $\mathcal C'$ with probability $p  \in [0,1]$.
\input{well-formed.tex}

\subsubsection{Well-Formedness.}
Additionally, we need to consider configurations where the entire state $\ps t$ is not covered by the variables appearing
in $M$. This is a matter of convenience and technicality for proofs of some of the results that follow.
Hence, we define a \emph{well-formed} quantum configuration with \emph{auxiliary types} to account for such configurations.
Consequently, the configurations introduced above should be seen as being \emph{total} in a certain sense, whereas the configurations 
with auxiliary types that we introduce next are more general in that the quantum data which is represented by the first 
component of the configuration does not have to be entirely covered by the free variables of the term from the main calculus.

\begin{defi}
\label{def:well-formed-auxiliary}
A quantum configuration $ (\pt , \isoterm_{\sigma} , M)$
is said to be \emph{well-formed with type $\mathrm{A}$ and auxiliary quantum data} of types
$[\ptype_1',\ldots,\ptype_k']$, written  $ ( \pt , \isoterm_{\sigma} , M ) \colon \mathrm{A}[\ptype_1', \ldots, \ptype_{\textit{k}}']$, if the following hold:
\begin{itemize}
       \item $\cdot \pdash \ps t \colon \ptype_1 \ptimes \ps{\cdots} \ptimes \ptype_n$ can be derived;
       \item $\cdot \udash \isoterm_{\sigma} \colon U(\ptype_1 \ptimes \ps{\cdots} \ptimes \ptype_n , \ptype_1'' \ptimes \ps{\cdots} \ptimes \ptype_{m+k}'')$
             can be derived;
       \item 
             $x_1 : \mathcal B(\ptype_1'') , \ldots , 
               x_m : \mathcal B(\ptype_m'') \vdash M : \mathrm{A}$
             can be derived;
       \item $[\ptype_1'', \dots, \ptype_{m+k}''] = [\ptype_1'', \dots, \ptype_m'',\ptype_1', \dots, \ptype_k']$ as lists of pure quantum types.
\end{itemize}
\end{defi}
$\wfconfs{\confs}{\mathrm{A}[\ptype_1',\ldots,\ptype_{\textit{k}}']}$ denotes the set of well-formed configurations wrt type $\mathrm{A}$ and auxiliary quantum types $[\ptype_1',\ldots,\ptype_k']$.

\subsubsection{Reduction Rules}
The reduction rules along with the unitaries required to state them are given in Figure~\ref{fig:red_uni}.
$\{\pb_{ij} \}_{i,j}$ stand for basis terms and $s$ in the reduction rule for $\tmeas{x}$ corresponds to the
position of $\ps b_i$ in the tensor product appearing in the quantum term. All sums are finite. 
The unitary $u_{\sigma_{gather}}$ 
appearing in the reduction rule for the term $\tletb{z}{x \otimes y}{N}$ merges two consecutive variables into one block, so that now a single
variable can represent it. Similarly, in the reduction rule for the term $\tletpairb{x}{y}{z}{N}$, 
the unitary $u_{\sigma_{divide}}$ partitions a block into two, so that two variables can represent them.
The two rules are symmetric as expected, since these terms only allow switching between
types $\mathcal B(\ptype_1 \ptimes \ptype_2)$ and $\mathcal B(\ptype_1) \otimes \mathcal B(\ptype_2)$. The last rule of 
Figure~\ref{fig:red_uni} is a \emph{contextual rule} which holds for any \emph{evaluation context} $E$. Evaluation contexts are defined by the following grammar:
\[
\scalebox{1}{
\begin{tabular}{l l l}
       $E$ & := & $[.] \alt E \otimes N \alt V \otimes E \alt \tletpair{x}{y}{E}{N} \alt \tinl{E} \alt \tinr{E}$\\
        & \multicolumn{2}{l}{$\alt \tcase{E}{M}{N} \alt \tforce{E} \alt \tsucc{E} \alt E\, N$}\\
        & \multicolumn{2}{l}{$\alt V\, E \alt meas(E) \alt \tun{\isoterm}{E} \alt \tletb{z}{E}{N} \alt \tletpairb{x}{y}{E}{N}$}
\end{tabular}
}
\]

One of the premises of the contextual rule states that if $\sigma$ is a permutation which does not
act on the same set of indices as $\sigma_1,\sigma_2$, then we can form a ``union'' of these  permutations 
by composing an \emph{extension} of the two. For a permutation $\sigma : S \rightarrow S$, we can define its extension $\sigma^{ext} : S \sqcup S' \rightarrow S \sqcup S'$  by $\sigma^{ext}(s) \defeq  \sigma(s)$, if $ s \in S$, $\sigma^{ext}(s') \defeq  s'$, if $s' \in S'$.
We now introduce the reduction relation we use to define the operational semantics. 
\input{unitaries-rearrange.tex}

\begin{defi}\label{def:reduction}
The \emph{reduction relation} $\cdot \leadsto_{\cdot} \cdot \subseteq \svconfs \times [0,1] \times \svconfs $ is defined in the following way: we write $(\ps v , u_{\sigma} , M) \leadsto_p (\ps v' , u_{\sigma'} , M')$ whenever
\[
\scalebox{1}{
\begin{bprooftree}
\AxiomC{$ \cdot \pdash \ps v = \pt : \pQ$}
\AxiomC{$(\pt , u_{\sigma}, M) \rightarrow_p (\pt' , u_{\sigma'} , M')$}
\AxiomC{$ \cdot \pdash \pt' = \ps v' : \pQ'$}
\TrinaryInfC{$(\ps v , u_{\sigma}, M) \leadsto_p (\ps v' , u_{\sigma'} , M')$}
\end{bprooftree}
}
\]
\end{defi}

The intuition behind the above definition is that it allows us to reason modulo equality of the pure terms $\pt$, as a consequence of Lemma~\ref{lem:equational-sn}.

\subsection{Main Properties}
We show that well-formedness is preserved by the reduction $\cdot \leadsto_{\cdot} \cdot$.

\begin{lem}[Substitution]
\label{lem:subst-o}
For a value $V$, if $!\Delta,\Sigma_1 \vdash V : \mathrm{A}$ and for a term $M$, $!\Delta, \Sigma_2, x:\mathrm{A} \vdash M : \mathrm{B}$ then $!\Delta, \Sigma_1 , \Sigma_2 \vdash M[V/x] : \mathrm{B}$.
\end{lem}

\begin{lem}[Substitution for Contexts]
\label{lem:contexts_subst}
If $!\Delta , \Sigma_1 \vdash M : A$ and $!\Delta , \Sigma_1 , \Sigma_2 \vdash E[M] : B$, then for $!\Delta , \Sigma_1' \vdash M' : A$,
we have $!\Delta , \Sigma_1' , \Sigma_2 \vdash E[M'] : B$.
\end{lem}

\begin{lem}[Preservation for translation]
\label{lem:pres_t}
If $\cdot \pdash \ps v : \pQ$ then $\cdot \vdash \ov{\ps v} : \ov{\ps Q}$.
\end{lem}
\begin{proof}
Proof is by induction on the rules for the translation.
\end{proof}

\begin{thm}
\label{thm:subred_strong}
For any $\mathcal C \in \wfconfs{\confs}{\mathrm{A}[\ptype_1', \ldots, \ptype_{\textit{k}}']}$,  if $\mathcal C \rightarrow_p \mathcal C'$, for some probability $p \in [0,1]$, then 
$\mathcal C' \in \wfconfs{\confs}{\mathrm{A}[\ptype_1', \ldots, \ptype_{\textit{k}}']}$.
\end{thm}
\begin{proof}
We provide a sketch of the proof here.
Proof is by induction on the reduction relation $\rightarrow$.
If $\mathcal C$ is a value configuration, then there is no $\mathcal C'$ s.t. $\mathcal C \rightarrow_p \mathcal C'$.
Hence, we consider the cases when $\mathcal C$ is not a value configuration. Note that without loss of generality, we can assume that
$\mathcal C$ is of the form $(\ps t, u_{\sigma}, M)$. 
We give the proof for a few important cases,
others follow similarly.
	\begin{itemize}
		\item $\mathcal C = (\ps t , u_{\sigma}, (\lambda x.M)\, V)$ for term M and value V:\\
		      According to the rules in ~\ref{fig:red_uni}, this configuration reduces to $\mathcal C' = (\ps t , u_{\sigma}, M[V/x])$.
                      Since the rule doesn't modify the quantum state, the auxiliary quantum data associated with $\mathcal C'$ is the same as that for 
		      $\mathcal C$.
		      Since $\mathcal C$ is well-formed, we have that $!\Delta , \Sigma_1 , \Sigma_2 \vdash (\lambda x.M)V : B$ is a valid typing judgment.
		      The formation rules for such a judgment imply that $!\Delta , x: A , \Sigma_1 \vdash M : B$ and $!\Delta , \Sigma_2 \vdash V : A$ are valid 
		      typing judgments. Hence, substitution lemma gives us that $!\Delta , \Sigma_1 , \Sigma_2 \vdash M[V/x] : B$ is a valid typing judgment.
		      For a more general configuration $\mathcal C = (\ps t , u_{\sigma}, M\, N)$ we apply the induction hypothesis along with Lemma~\ref{lem:contexts_subst}. 
	      \item $\mathcal C = (\ps t , u_{\sigma} , \tpure{\ps t'})$:\\
		      This configuration reduces to $(\ps{t \otimes t'} , u_{\sigma_{swap}} \circ (u_{\sigma} \otimes u_{id}),x)$.
		      Formation rules of $\tpure{\ps t'}$, require $\ps t'$ to be a closed term of type $\pQ$ from the quantum control fragment. Moreover,
		      $\ps t$ is a closed term of type $\ps {Q_1 \otimes \cdots \otimes Q_{m+k}}$. Hence, the formation rules from the quantum control fragment
		      ensure that $\cdot \pdash \ps {t \otimes t' : Q_1 \ptimes   \ps{\cdots} \ptimes Q_{m+k} \otimes Q}$ is a valid typing judgment. $x$ is the only free variable
		      in the new term , which corresponds to $\ps t'$. Hence applying $u_{\sigma_{swap}}$ on $\ps{t \otimes t'}$ ensures that the reduct is
		      well-formed. The auxiliary quantum data in $\mathcal C'$ remains unchanged, since $x$ covers $\ps t'$. Hence the list equation for types
		      holds. Therefore, $(\ps{t \otimes t'} , u_{\sigma_{swap}} \circ (u_{\sigma} \otimes u_{id}) , x)$ 
		      is a well-formed configuration of the same type as $\mathcal C$.
	      \item $\mathcal C = (\ps t , u_{\sigma}, \tun{\isoterm}{z})$ for a unitary $u$:\\
		      This configuration reduces to $((u_{\sigma}^* \circ (\isoterm \otimes id) \circ u_{\sigma}) \ps t , u_{\sigma} ,  z)$. If $.\udash \isoterm : U(\pQ_1, \pQ_2)$,
		      from the typing rules for unitary application it follows that $z : \mathcal B(\pQ_1) \vdash \tun{\isoterm}{z} : \mathcal B(\pQ_2)$. Since,
		      $(u_{\sigma}^* \circ (\isoterm \otimes id) \circ u_{\sigma})$ applies $\isoterm$ to the component of $\ps t$ that $z$ corresponds to, after reduction,
		      we have $z : \mathcal B(\pQ_2) \vdash z : \mathcal B(\pQ_2)$. Moreover,
		      everything else in the quantum term remains unchanged, and hence the term remains well-typed after reduction. Therefore, the quantum
		      configuration obtained after reduction is well-formed.
	      \item $\mathcal C = (\Sigma_i p_i \cdot \Sigma_j {\ps \alpha_{ij}} {\ps \cdot} \ps b'_{ij} \ptimes {\ps \cdots} \ptimes \ps b_i \ptimes {\ps \cdots} \ptimes \ps b''_{ij} , 
		      u_{\sigma_s} , \tmeas{x})$:\\
		      This configuration reduces to $\mathcal C' = ( \Sigma_j \ps \alpha_{mj}  \ps{\cdot} \ps b'_{mj} \ptimes \ps{\cdots} \ptimes \ps b''_{mj} , u_{id} , \ov{\ps b_m})$ 
		      with probability $p_m$. Note that since $\mathcal C$ is well-formed, we have that 
		      $\cdot \pdash \Sigma_i p_i \ps{\cdot} \Sigma_j {\ps \alpha_{ij}} \ps{\cdot} \ps b'_{ij} \ptimes {\ps \cdots} \ptimes \ps b_i \ptimes {\ps \cdots} \ptimes \ps b''_{ij} :
		      \ps {Q_1 \otimes\ps{\cdots} \ptimes Q_m \ptimes \ps{\cdots}\ptimes Q_k}$ is
		      a valid typing judgment for some quantum types ${\pQ_{\ps 1} }, \ldots , \pQ_{\ps k}$. Hence, formation rules of the quantum control fragment ensure that 
		      $\cdot \pdash \Sigma_j \ps \alpha_{mj} \cdot \ps b'_{mj} \ptimes \ldots \ptimes \ps b''_{mj} : \ps {Q_1 \otimes \ldots \otimes Q_k}$ is a valid typing judgment.
		      Furthermore, since we have $ x: \mathcal B(\ps Q_s) \vdash \tmeas{x} : \ov{\ps Q_s}$, and ~\ref{lem:pres_t} gives us $. \vdash \ov{\ps b_m} : \ov{\ps Q_s}$,
		      type is preserved for the term. Finally, since the auxiliary quantum data types remain unchanged, we have that $\mathcal C'$ is well-formed with the same
		      type as that of $\mathcal C$.\qedhere
	\end{itemize}
\end{proof}

\begin{thm}[Subject Reduction]
\label{lem:sub_red}
For a configuration $\mathcal C_1 \in \wfconfs{\svconfs}{A} $, 
if $\mathcal C_1 \leadsto_p \mathcal C_2$ for some probability $p \in [0,1]$, then $\mathcal C_2 \in \wfconfs{\svconfs}{A}$.
\end{thm}
\begin{proof}
Follows from Theorem~\ref{thm:subred_strong}.
\end{proof}

Next, we show that progress holds. More specifically, a
well-formed semi-value configuration $\mathcal C$ is either a value
configuration or it reduces to a finite number of semi-value configurations
with total probability $1$, \emph{i.e.}, if we have not reached a normal form then
the probability of being stuck is $0$.

\begin{thm}
\label{thm:progress-strong}
	For a quantum configuration $\mathcal C \in \wfconfs{\svconfs}{\mathrm{A}[\ptype_1',\ldots,\ptype_k']}$, there are two possibilities:
\begin{itemize}
\item either $\mathcal C \in \vconfs$,
\item or there exists $\mathcal C_1 \in  \wfconfs{\svconfs}{\mathrm{A}[\ptype_1',\ldots,\ptype_k']}$ such that $\mathcal C \leadsto_{p_1} \mathcal C_1$.
 \end{itemize}
Let $\{ \mathcal C_i \}_{i \in I}$ be the set  of all such distinct (\emph{i.e.}, not $\alpha$-equivalent) configurations.
Then, the set $I$ is finite, and $\Sigma_{i \in I}p_i = 1$.
\end{thm}
\begin{proof}
Without loss of generality we assume $\mathcal C = (\ps v, u_{\sigma}, M)$ for some quantum value $\ps v$ and term $M$.
Proof proceeds by induction on the typing rule for $M$. We prove a few cases here, the other cases follow similarly.
	\begin{itemize}
	       \item $M = M'\, N'$:\\
		      By the induction hypothesis and rules in Figure~\ref{fig:red_uni}, we can assume that both $M',N'$ are values.
		      In this case $M = (\lambda x.M'')\, V$. From the rules in Figure~\ref{fig:red_uni}, we have that this configuration
		      reduces to $\mathcal C_2 = (\ps v , u_{\sigma}, M''[V/x])$.
	       \item $M = \tmeas{M'}$:\\
		      We rewrite the quantum value $\ps v$ in the first component of the configuration using the rule for $\leadsto_p $,
		      and reduce using the rule for measurement in Figure~\ref{fig:red_uni}. The resulting configuration is a value
		      configuration.
	       \item $M = \tun{\isoterm}{M}$:\\
		      From the rules in Figure~\ref{fig:red_uni} and the induction hypothesis, we can assume that $M$ is a value
		      of type $\mathcal B(\pQ)$ for some quantum type $\pQ$. By inspection, it follows that it has to be a variable.
		      Now, we can apply the rule for unitary application in Figure~\ref{fig:red_uni}, to obtain a configuration of the 
		      form $(u_{\sigma}^* \circ (\isoterm \otimes id) \circ u_{\sigma}\ps t , u_{\sigma} ,  z)$ for appropriate quantum
		      state $\ps t$ and unitaries $u_{\sigma}, \isoterm$. From the equational theory~(Lemma~\ref{lem:equational-sn}) it follows
              that the quantum state can be rewritten to a value, and thus applying the reduction rule for $\leadsto_p$, gives us that
		      the resulting configuration is a value configuration.
	       \item $M = \tpure{\ps t'}$:\\
		      We apply the rule in Figure~\ref{fig:red_uni}, the reduction rule for $\leadsto_p$ and Lemma~\ref{lem:equational-sn},
		      to obtain a value configuration. \qedhere
	\end{itemize}
\end{proof}

\begin{thm}[Progress]
	If $\mathcal C \in \wfconfs{\svconfs}{A}$, then either $\mathcal C \in
	\vconfs$, or there exists $\mathcal C_i \in  \wfconfs{\svconfs}{A}$ such
	that $\mathcal C \leadsto_{p_i} \mathcal C_i$.  Moreover, if $\{ \mathcal
	C_i \}_{i \in I}$ is the set of all such distinct (not
	$\alpha$-equivalent) configurations, then $I$ is finite, and
	$\sum_i p_i = 1$. 
\end{thm}
\begin{proof}
	Follows from Theorem~\ref{thm:progress-strong}.
\end{proof}

\subsection{Strong Normalization}

Next, we show that any reduction issued from a well-formed configuration under $\leadsto$ 
always terminates in finitely many steps.

\input{sn_new.tex}

\begin{thm}[Strong Normalization]
	\label{thm:st-norm}
	For a configuration $\mathcal C \in \wfconfs{\svconfs}{A}$, there is no
	infinite sequence of reductions $\mathcal C \leadsto_{p_0} \mathcal C_1
	\leadsto_{p_1} \mathcal C_2 \leadsto_{p_2} \cdots$.
\end{thm}
\begin{proof}
Follows from Theorem~\ref{thm:st-norm-strong}.
\end{proof}

\section{Illustrative Examples}

Here we give several examples to illustrate the utility of our calculus.
We start with some seminal quantum algorithms such as Teleportation expressed
in our syntax and how it evolves in our semantics. To further highlight the
need to combine quantum and classical control, we illustrate a commonplace
quantum algorithm requiring a language which supports both.

\begin{exa}
Let $\ps{ \ket{\phi}} \defeq  \ps{\sfrac{1}{\sqrt{3}} \psdot \ket{000} + \sfrac{1}{\sqrt{3}} \psdot \ket{010} + \sfrac{1}{\sqrt{3}} \psdot \ket{011}}$ and consider the following term:
\begin{gather}
\scalebox{1}{$M : \mathcal B (\ptqbit{2}) \otimes \textit{bit}$} \notag \\
\scalebox{1}{$M \defeq \tletpairb{x}{y}{\tpure{\ps{ \ket{\phi}}}}{x \otimes \tmeas{y}}$} \notag
\end{gather}

 Since, by Figure~\ref{fig:red_uni}, the following reduction holds
\[ (\ps \ast , u_{id} , \tpure{\ps{ \ket{\phi}}}) \rightarrow_1
                (\ps{ \ket{\phi}} , \isoterm_{\sigma_{swap}} , z),
                \]
we obtain the following small-step reduction using an application of the contextual rule of Figure~\ref{fig:red_uni}
\[
(\ps \ast , \isoterm_{id} , M) \rightarrow_1 ( \ps{ \ket{\phi}} , \isoterm_{\sigma_{swap}} , \tletpairb{x}{y}{z}{x \otimes \tmeas{y}}).
                     \]
Consequently, the reductions below hold
\begin{align*}
&\scalebox{1}{$(\ps{\ket{\phi}} , u_{\sigma_{swap}} , \tletpairb{x}{y}{z}{x \otimes \tmeas{y}})$} \notag  \\
&\scalebox{1}{$\leadsto_1(\ps{\sfrac{\sqrt{2}}{\sqrt{3}}\psdot (\sfrac{1}{\sqrt{2}} \psdot \ket{000} + \sfrac{1}{\sqrt{2}} \psdot \ket{010}) +
                \sfrac{1}{\sqrt{3}} \psdot \ket{011}} ,
                u_{\sigma_{swap}} , x \otimes \tmeas{y})$} \notag \\
&\scalebox{1}{$\begin{cases}
                                  \leadsto_{\sfrac{2}{3}}(\ps{\sfrac{1}{\sqrt{2}}\psdot \ket{00} + \sfrac{1}{\sqrt{2}}\psdot \ket{01}} , u_{\sigma_{swap}} , x \otimes \tinl{\ast}) \\
                                                  \leadsto_{\sfrac{1}{3}}(\ps{\ket{01}} , u_{\sigma_{swap}} , x \otimes \tinr{\ast})
\end{cases}$} \notag
\end{align*}
        where $x$ points to the first two qubits, and $y$ points to the last qubit. Note that the unitary $u_{\sigma_{swap}}$ acts on the state $\ps {\ast \otimes \ket{\phi}}$, and $\ps{\ket{\phi}}$
        is syntactic sugar representing it, as clarified in the beginning of this section.
\end{exa}

\begin{exa}[\textit{Bell State}]
We show how a \emph{Bell state} $\textbf{Bell} \defeq  \ps{\sfrac{1}{\sqrt{2}} \psdot \ket{00} + \sfrac{1}{\sqrt{2}} \psdot \ket{11}}$ can be prepared by a term of our calculus. We do this by
preparing two qubits and applying the necessary unitary operations on them:
\begin{gather}
\scalebox{1}{$\textrm{Bell}_S : \mathcal B (\ptqbit{2})$} \notag  \\
\scalebox{1}{$\textrm{Bell}_S \defeq  \tletb{z}{\tun{H}{\tpure{\ketz}} \otimes \tpure{\ketz}}{\tun{CNOT}{z}}$} \notag
\end{gather}
By the rules given in Figure~\ref{fig:red_uni}, we have the following:
\[
\scalebox{1}{$(\ps \ast , u_{id} , \tun{H}{\tpure{\ketz}} \otimes \tpure{\ketz}) \rightarrow_1
                (\textbf{Bell} , u_{\sigma_{swap}} , x \otimes y)$}
        \]
Hence, $\scalebox{1}{$(\ps \ast , u_{id}, \textrm{Bell}_S) \rightarrow_1 (\textbf{Bell} , u_{\sigma_{swap}} , \tletb{z}{x \otimes y}{\tun{CNOT}{z}})$}$

Thus, applying rules from Figure~\ref{fig:red_uni} and Definition~\ref{def:reduction} we have,
\begin{gather}
\scalebox{1}{$(\textbf{Bell} , u_{\sigma_{swap}} , \tletb{z}{x \otimes y}{\tun{CNOT}{z}}) \rightarrow_1
(\textbf{Bell} , u_{\sigma_{swap}} ,\tun{CNOT}{z})$} \notag \\
\scalebox{1}{$(\textbf{Bell} , u_{\sigma_{swap}} ,\tun{CNOT}{z}) \leadsto_1
(\textbf{Bell} , u_{\sigma_{swap}} , z)$} \notag
\end{gather}
\end{exa}

\begin{exa}[\textit{Bell Measurement}]
We illustrate the measurement of two qubits with respect to the \textit{Bell basis}. For that purpose, we first
apply a unitary operation on the qubits and finally measure them with respect to the computational basis:
\begin{gather}
\scalebox{1}{$\textrm{Bell}_M : \qbit \multimap (\qbit \multimap \bit \otimes \bit)$} \notag\\
\scalebox{1}{$\textrm{Bell}_M \defeq \lambda x.\lambda y. \tletb{z}{x \otimes y}{\text{let } q_1 \otimes q_2 = \tun{CNOT}{z} \text{ in }} \tmeas{\tun{H}{q_1}} \otimes \tmeas{q_2}$} \notag
\end{gather}
\end{exa}

\begin{exa}[\textit{Teleportation}]
Teleportation is an algorithm designed to transport a qubit in
an unknown state $|\psi \rangle$ using two classical bits (\cite{nielsen2001quantum}) of communication.
We now showcase the teleportation algorithm in our calculus. In order to do so, we first construct a
term which allows us to apply different unitaries to a qubit based on the values that two bits take:
\[ \textrm{App}_U : \qbit \multimap (\bit \otimes \bit \multimap \qbit)   \]
\begin{align*}
\textrm{App}_U \defeq \lambda q. \lambda p. \tletpair{x}{y}{p}{\tcaseof{x,y}}\{  0,0 &\to\ U_{00}(q),\\
 0,1 &\to\ U_{01}(q),\\ 
 1,0 &\to\ U_{10}(q),\\ 
 1,1 &\to\ U_{11}(q)\}
\end{align*}
where  $ 0 \defeq inl(\ast)$, $ 1 \defeq inr(\ast)$, and $
\textit{case }x_1,x_2  \textit{ of } \{\ldots\}
 $
denotes $2$ nested cases, and  $U_{ij}$ denotes the unitary to be applied corresponding to the case when the first bit has value $i$ and the second bit has value $j$. Finally, terms $\textrm{App}_U$, $\textrm{Bell}_M$, and $\textrm{Bell}_S$ can be used to construct quantum teleportation:
\begin{gather}
\scalebox{1}{$\textrm{tele}  : \qbit \multimap \qbit$} \notag \\
\scalebox{1}{$\textrm{tele} \defeq \lambda q. \tletpairb{x}{y}{\textrm{Bell}_S}{}  \tletpair{b_1}{b_2}{\textrm{Bell}_M \; x \; q}{\textrm{App}_U \; q  \; (b_1 \otimes b_2) }$} \notag
\end{gather}
\end{exa}

\begin{exa}[\textit{Quantum walk search}]\label{exa:walk}
Given a graph with some marked vertices, a quantum walk search is the quantum
analogue of a random walk~\cite{childs2003exponential}, which looks for
these marked vertices.  Each node in the graph represents the state of two
quantum registers. A conditional on the application of a coin operator on the
first register decides the direction in which a walker should take the next step.
A step in this direction is represented by the application of a unitary on the second
register. After a given number of steps, measurement is performed on
the second register to reveal whether the walker has arrived on a marked
node. Here, we present a $k$-step quantum walk on a cycle with 10 nodes for
simplicity in our language. Representing the second register with qubits,
in this case, would require 4 qubits, whereas we represent the state of the
second register with a single $\pnat$. In the unitary $\pcnot_{\pnat}$ below, 
we use some (hopefully obvious) syntactic sugar with the expressions involving 
$\textbf{n}$ in order to avoid writing all the cases; the letter $\textbf{n}$ 
there should not be seen as a variable. \\
\begin{minipage}{0.3\textwidth}
	\scalebox{0.7}{ $\stikz{q-walk.tikz}$ }
\end{minipage}
\begin{minipage}{0.6\textwidth}
		$
			u_1 \defeq
			\left\{ \begin{array}{lcll}
				\mid \ket{\psu{2n}} & \iso & \ket{\psu{2n+1}}, 
				& \textbf{n} \leq 4 \\
				\mid \ket{\psu{2n+1}} & \iso & \ket{\psu{2n}}, 
				& \textbf{n} \leq 4 \\
				\mid \pket{\psu{y+10}} & \iso & \pket{\psu{y+10}}
			\end{array} \right\}
 		$
 		\\
		$
 			u_2 \defeq
			\left\{ \begin{array}{lcll}
				\mid \pket{\psu{0}} & \iso & \pket{\psu{9}} & \\
 				\mid \pket{\psu{2n+2}} & \iso & \pket{\psu{2n+1}}, 
 				& \textbf{n} \leq 3 \\
 				\mid \pket{\psu{2n+1}}  & \iso & \pket{\psu{2n+2}}, 
 				& \textbf{n} \leq 3 \\
 				\mid \pket{\psu{9}} & \iso & \pket{\psu{0}} & \\
				\mid \pket{\psu{y+10}} & \iso & \pket{\psu{y+10}}
				&
			\end{array} \right\}
		$

	\ \\

	\scalebox{1}{
		$
		\begin{array}{l}
			\pcnot_{\pnat} \colon \isotype{\pqbit \ptimes \pnat}
			{\pqbit \ptimes \pnat} \\
			\pcnot_{\pnat} \defeq \qif{x}{u_1}{u_2}
		\end{array}
		$
	}
\end{minipage}

We define the unitary $S$ corresponding to one step of the walk as $S \defeq \pcnot_{\pnat} \circ (\phad \otimes \mathrm{Id})$.
Finally, the $k$-step walk can be represented in our language as follows, where $S^k$ represents
$k$ compositions of $S$ with itself:
\begin{gather}
 	\textrm{walk} : \mathcal B(\pqbit \ptimes \pnat) \multimap
 	\qbit \otimes \mathcal \Nat \notag \\
	\textrm{walk} \defeq \lambda z.
	\tletpairb{x_1}{x_2}{\tun{S^k}{z}}{x_1 \otimes \tmeas{x_2}} \notag
\end{gather}

We now exhibit how the lambda expression for the quantum walk would reduce in
our semantics when evaluated at the quantum state $\ketz \ptimes \ket{\psu{0}}$
for $k=1$. For the rest of this example, let $M \defeq
\tletpairb{x_1}{x_2}{\tun{W}{z}}{x_1 \otimes \tmeas{x_2}}$ and $\ps v \defeq
\ketz \ptimes \ket{\psu{0}}$. The following reductions can be derived:
\begin{align*}(\ps{\ast}, \isoterm_{id}, \tpure{\ps v}) &\leadsto_1 (\ps{\ast} \ptimes \ps v, \isoterm_{\sigma_{swap} \circ (\isoterm_{id} \otimes \isoterm_{id})}, x)\\
(\ps{\ast}, \isoterm_{id}, (\lambda z.M)\, \tpure{\ps v})& \leadsto_1 (\ps{\ast} \ptimes \ps v, \isoterm_{\sigma_{swap}} \circ (\isoterm_{id} \otimes \isoterm_{id}), M[x/z])
\end{align*}
It also holds that,
\begin{align*}
&(\ps{\ast} \ptimes \ps v ,\isoterm_{\sigma_{swap}} \circ (\isoterm_{id} \otimes \isoterm_{id}), \tletpairb{x_1}{x_2}{\mathcal B(S)(x)}{x_1 \otimes \tmeas{x_2}})\\
&\leadsto_1 (S~(\ps{\ast} \ptimes \ps v) , \isoterm_{\sigma_{swap}} \circ (\isoterm_{id} \otimes \isoterm_{id}), \tletpairb{x_1}{x_2}{x}{x_1 \otimes \tmeas{x_2}})
\end{align*}
	
\[
	\begin{bprooftree}
		\AxiomC{$(\ps{\ast}, \isoterm_{id}, \tpure{\ps v}) \leadsto_1 (\ps{\ast} \ptimes \ps v, u_{\sigma_{swap} \circ (u_{id} \otimes u_{id})}, x)$}
		\UnaryInfC{$(\ps{\ast}, u_{id}, (\lambda z.M)\tpure{\ps v}) \leadsto_1 (\ps{\ast} \ptimes \ps v, u_{\sigma_{swap}} \circ (u_{id} \otimes u_{id}), M[x/z])$}
		\UnaryInfC{\stackanchor{$(\ps{\ast} \ptimes \ps v ,u_{\sigma_{swap}} \circ (u_{id} \otimes u_{id}), 
		\tletpairb{x_1}{x_2}{\mathcal B(S)(x)}{x_1 \otimes \tmeas{x_2}})
		\leadsto_1$} 
		{$(u(\ps{\ast} \ptimes \ps v) , \\
		u_{\sigma_{swap}} \circ (u_{id} \otimes u_{id}),
		\tletpairb{x_1}{x_2}{x}{x_1 \otimes \tmeas{x_2}})$}}
	\end{bprooftree}
\]

The unitary $\isoterm$ in the above derivation is given by 
$\isoterm \defeq (\isoterm_{\sigma_{swap}} \circ (\isoterm_{id} \otimes \isoterm_{id}))^{\ast} \circ (S \otimes id) \circ (\isoterm_{\sigma_{swap}} \circ (\isoterm_{id} \otimes \isoterm_{id}))$.
The quantum term $\isoterm(\ps{\ast} \ptimes \ps v)$ rewrites to $\ps v' \defeq \sfrac{1}{\sqrt 2} \psdot (\ps{\ast} \ptimes \ketz \ptimes \ket{\psu{1}} + \ps{\ast} \ptimes \keto \ptimes \ket{\psu{9}})$.       We define $\textit{nine} \defeq \textit{succ}^9(\tzero)$.
Thus, the following reductions hold
\begin{align*}
	&(\ps v', \isoterm_{\sigma_{swap}} \circ (\isoterm_{id} \otimes \isoterm_{id}),  	\tletpairb{x_1}{x_2}{x}{x_1 \otimes \tmeas{x_2}})\\ &\leadsto_1  (\ps v',\isoterm_{\sigma_{divide}} \circ (\isoterm_{\sigma_{swap}} \circ (\isoterm_{id} \otimes \isoterm_{id})),
	x_1 \otimes \tmeas{x_2})\\
	&\begin{cases}
		\leadsto_{\sfrac{1}{2}} (\ps{\ast} \ptimes \ketz,\isoterm_{\sigma_{swap}} , x_1 \otimes \tsucc{\tzero}) \\
		\leadsto_{\sfrac{1}{2}} (\ps{\ast} \ptimes \keto,\isoterm_{\sigma_{swap}} , x_1 \otimes \textit{nine})
	\end{cases}
\end{align*}

\end{exa}

\section{Denotational Semantics: Soundness and Adequacy}\label{sub:denotational}

In this section we describe the denotational semantics of the main calculus. In order to do so,
we first introduce some mathematical preliminaries and the relevant categorical structure. The denotational
model is based on von Neumann algebras, mathematical structures well-suited to describe not only mixed
quantum computation, but also classical information in the Heisenberg picture
of quantum mechanics.

\subsection{Von Neumann Algebras.}
\label{sub:von-neumann-algebras}
Intuitively, a von Neumann algebra may be understood as a certain kind of
(sub)algebra of the space $\mathcal B(H)
$ of bounded operators on the Hilbert space $H$.  More specifically, a non-empty subset $M \subseteq \mathcal B(H)$ is
called a $*$\emph{-subalgebra} of $\mathcal B(H)$ if $M$ is closed under the operations
of addition, scalar multiplication, composition, and involution of $\mathcal B(H).$
Given a subset $S \subseteq \mathcal B(H)$, its \emph{commutant}, written $S'$, is the
set $S' \defeq \{ x \in \mathcal B(H)\ |\ \forall y \in S,\ x y = y x \},$
\emph{i.e.}, the operators $x$ that commute (with respect to composition) with all
operators in $S.$ A simple algebraic argument shows that $S \subseteq S''
\defeq (S')'$ for any $S \subseteq \mathcal B(H).$ We now give the main definition of this subsection.

\begin{defi}[\cite{takesaki}]
  \label{def:von-neumann-algebra}
  A \emph{von Neumann algebra} on a Hilbert space $H$ is a $*$-subalgebra $M$
  of $\mathcal B(H)$ such that $M'' = M.$
\end{defi}

\begin{exa}
  One of the most important examples of a von Neumann algebra is $\mathcal
  B(H)$ itself.  Observe that in the special case when $H$ has dimension $n \in
  \mathbb N$, then $\mathcal B(H)$ may be identified with the algebra $\mathrm
  M_n(\mathbb C)$ of complex $n \times n$ matrices.
\end{exa}

A von Neumann algebra is \emph{commutative} if for every $x, y \in
M$, we have that $xy = yx.$ Such von Neumann algebras can be used to represent
classical information.  The Banach space $\ell^\infty(X)
\defeq \left\{ f \colon X \to \mathbb C \ |\ \sup_{x \in X} |f(x)| < \infty
\right\}$, for any set $X$, equipped with the supremum norm, can be identified with a
commutative von Neumann algebra on the Hilbert space $\ell^2(X)$ via the
isometric embedding given by $m_X \colon \ell^\infty(X) \to \mathcal B(\ell^2(X))$ and $m_X(f)= m_f,$ where $m_f \colon \ell^2(X) \to \ell^2(X)$ is the bounded linear
operator defined by $m_f(\sum_{x \in X} \alpha_x \ket x)$ $\eqdef \sum_{x \in
X} f(x)\alpha_x \ket x.$

An important special
case for our development includes the von Neumann algebra
$\ell^\infty(\mathbb B)$, where $\mathbb B = \{true,false\}$, which can be
identified with $\mathbb C \oplus \mathbb C$ and which corresponds to the
interpretation for the type of (classical) bits in our language.
Another important example is given by the space $\ell^\infty(\mathbb N)$ which
we use to interpret the type $\Nat$ of classical natural numbers in our
language.

It follows easily from Definition \ref{def:von-neumann-algebra} that every von
Neumann algebra $M \subseteq \mathcal B(H)$ contains the identity operator $1_H$ of $H$,
which we often write as $1_M$ as well. Obviously, $1_M \in M$
is the neutral
element of $M$ with respect to multiplication, called the \emph{unit}
of $M$. Similarly, $M$ always contains the zero operator $0_H$ which we
also write as $0_M.$ Another important kind of operators are the
\emph{self-adjoint} (or Hermitian) operators, which are
operators $x \in M$, such that $x = x^*.$ An operator $x \in M$
is  \emph{positive} if $\langle h, xh \rangle \geq 0$ for every $h \in H.$ The poset
(partially ordered set) of Hermitian elements of $M$ is given by the set $M_h
\defeq \{ x \in M\ |\ x = x^* \}$ together with the \emph{L\"owner order}: $x
\leq y$ iff $y - x$ is a positive operator. It follows that $x \in M$ is
positive iff $0_M \leq x.$ Another important poset is the \emph{unit interval}
of $M$: the subposet of $M_h$ given by $[0,1]_M \defeq \{ x \in M_h \
|\ 0_M \leq x \leq 1_M \} $.  In fact, $[0,1]_M$ is a dcpo
(\emph{directed-complete partial order}) with least element $0_M$.

\subsection{Heisenberg vs \Schrod{} Pictures of Quantum Mechanics}

Quantum computation can be described in two different pictures of quantum theory called
the \emph{Heisenberg picture} and \emph{\Schrod{} picture}, respectively. The two
pictures are dual to each other and we briefly recall the Heisenberg-\Schrod{} duality
between these two pictures. We hope that this can help readers who are familiar with only
the \Schrod{} picture to gain intuition about our semantics, which is described in the
Heisenberg picture.

The Banach space $\mathcal T(H)$ of \emph{trace class} operators on the Hilbert space $H$
can be intuitively understood as the space of operators on $H$ that admit an appropriate notion of
trace (\emph{e.g.}, if $H$ is finite-dimensional, then all linear operators on it are
trace class).
The Banach space $\mathcal B(H)$ is dual to $\mathcal T(H)$, i.e. there exists an isometric isomorphism
$\mathcal T(H)^* \cong \mathcal B(H)$, so for simplicity in the following explanation we identify $\mathcal B(H)$ with $\mathcal T(H)^*$.
A quantum operation (also called quantum channel) from $H_1$
to $H_2$ in the \emph{\Schrod{} picture} is a \emph{completely-positive
trace-preserving} (CPTP) map $\varphi \colon \mathcal T(H_1) \to \mathcal
T(H_2)$.  Whereas the corresponding quantum operation (channel) in the
\emph{Heisenberg} picture is a \emph{normal completely-positive unital} (NCPU)
map $\varphi^* \colon \mathcal B(H_2) \to \mathcal B(H_1)$ that goes in the
other direction and which is determined by the Banach space adjoint of
$\varphi$ (i.e., $\varphi^*$ is the dual map of $\varphi$). The Heisenberg-\Schrod{} duality
can be given a mathematical formulation through the fact that there exists a
bijective correspondence between CPTP maps $\varphi \colon \mathcal T(H_1) \to \mathcal
T(H_2)$ and NCPU maps $\varphi^* \colon \mathcal B(H_2) \to \mathcal B(H_1)$ which is given by taking the dual map
(in the left-to-right direction).

As we do not need the details of CPTP maps between infinite-dimensional spaces, we only
introduce NCPU maps in the sequel. This can be done not only for spaces of the form $B(H)$; in fact,
NCPU maps can be defined between von Neumann algebras as well and we do so in the sequel.
One of the strengths of von Neumann algebras is the fact that
they allow us to easily and elegantly model not only quantum information, but
also classical information by using \emph{commutative} von Neumann algebras.

Finally, it can be useful to consider the following view of the Heisenberg and \Schrod{}
pictures. In the \Schrod{} picture, one says that the state changes, whereas the observable remains unchanged.
Indeed, note that CPTP maps send density operators (i.e., states) to density operators. Whereas in the Heisenberg
picture, one says that the state remains invariant, but the observable changes. Observables in the Heisenberg picture
are represented by self-adjoint operators and these are preserved by NCPU maps (with the unitality ensuring appropriate normalization constraints are respected).

\subsection{Categorical Structure.}

Next, we recall two important types of morphisms between von Neumann algebras.
We use them to organize the relevant data into categories that allow us to
define our semantic interpretation.

\subsubsection{NMIU maps.}
  Let $M$ and $N$ be two von Neumann algebras (not necessarily over the same
  Hilbert space). A linear function $\varphi \colon M \to N$ is called
  \emph{multiplicative}/\emph{involutive}/\emph{unital} if it preserves
  multiplication/involution/the unit of $M$, respectively. A linear map
  $\varphi \colon M \to N$ is called \emph{subunital} if $\varphi(1_M) \leq
  1_N$.  A linear map is positive if it maps every positive operator to a positive operator. 
  When a map $\varphi \colon M \to N$ is both positive and subunital,
  then it restricts to a map $[0,1]_M \to [0,1]_N$ and we say that $\varphi$
  is \emph{normal} whenever this restriction is a Scott-continuous map between
  the two dcpo's. A map $\varphi \colon M \to N$  is called an \emph{NMIU}-map, if it is normal, multiplicative, involutive, and unital.

\subsubsection{NCPSU maps.}
If $N$ is a von Neumann algebra (over $H$), then we can equip the set $\matalg
n (N)$ of $n \times n$ matrices with entries in $N$ with the structure of a von
Neumann algebra (over $\oplus_{i=1}^n H$, and for simplicity, here we identify
$\matalg n (\mathcal B(H))$ and $\mathcal B(\oplus_{i=1}^n H)$). If $\varphi
\colon N \to N'$ is a linear map between von Neumann algebras, then we can
define a linear map $\varphi^{(n)} \colon \matalg n (N) \to \matalg n (N')$,
for every $n \in \mathbb N$ by using $\varphi$ component-wise. The map $\varphi \colon N \to N'$ is \emph{completely positive} if
$\varphi^{(n)} \colon \matalg n (N) \to \matalg n (N')$ is positive for every
$n \in \mathbb N$. A map $\varphi \colon M \to N$ is called an \emph{NCP(S)U}-map, if it is normal, completely positive, and (sub)unital.

The two main notions of morphisms that we use in our semantics are given by the
NMIU and NCPSU maps. Every NMIU-map is also NCP(S)U, but not vice-versa (\emph{e.g.},
state preparation), hence NMIU-maps give a stronger notion of morphism. We
write $\NMIU \subseteq \NCPSU$ for the \emph{opposite} (sub)categories with
objects von Neumann algebras and with NMIU/ NCPSU maps as morphisms between
them. There are several reasons for working with the opposite categories:
\begin{itemize}
\item They allow us to reason appropriately about certain physical processes
in the Heisenberg picture (as it is customary with von Neumann
algebras), rather than in the Schrödinger one;
\item The aforementioned opposite categories
enjoy several categorical properties that we use in our semantics that are not
known to hold for their non-opposite counterparts (\emph{e.g.}, monoidal closure,
coproducts).
\end{itemize}

\subsubsection{Symmetric Monoidal Structure.}
The categories $\NMIU$ and $\NCPSU$ have a symmetric monoidal structure when
equipped with the \emph{spatial tensor product} \cite[Prop
7.2]{kornell2012quantum}: the tensor unit can be identified with $\mathbb C$, and the tensor product of von Neumann
algebras $M \subseteq B(H)$ and $N \subseteq B(K),$ written $M \otimes N$,
is the smallest von Neumann subalgebra of $\mathcal B(H \otimes K)$ that
contains the algebraic tensor product $M \odot N$ of the underlying vector spaces.
The subcategory inclusion $\mathcal I \colon \NMIU \hookrightarrow \NCPSU$
is a \emph{strict} symmetric monoidal functor.
Moreover, the category $\NMIU$ is a \emph{closed} symmetric monoidal category
and we write $N^{\ast M}$ for its internal hom \cite{kornell2012quantum}.

\subsubsection{Coproducts.}
Both categories $\NMIU$ and $\NCPSU$ have coproducts that are described in the
same way: the initial object is the von Neumann algebra on the zero-dimensional Hilbert space $0$ and the coproduct is given by direct sum with
respect to the supremum norm \cite[Definition 3.4]{takesaki}. More
specifically, if $M \subseteq \mathcal B(H)$ and $N \subseteq \mathcal B(K)$ are von Neumann
algebras, then $M \oplus N$ is a Banach space (equipped with the supremum norm)
which can be isometrically embedded into $\mathcal B(H \oplus K)$.
The image of $M \oplus N$ under this embedding is a von Neumann subalgebra of $\mathcal B(H \oplus K)$, but we abuse notation and we do not carefully distinguish between the two as the differences are not essential.

Before delving into the details of the interpretation, we present the functor which facilitates the main contribution of the paper, i.e combining
quantum and classical control.

\begin{figure}[!h]
		\stikz{adjunctions.tikz}
	\caption{Categorical relations.}
	\label{fig:adjunctions-c}
\end{figure}
\subsection{The $\mathcal B$ Functor.}\label{sub:b-functor}
In order to model the interaction between the pure quantum subsystem and the
main calculus, we define $\mathcal B \colon \Isometry \to \NCPSU$ and show it is a
functor (Lemma ~\ref{lem:functor}) mapping a Hilbert space $H$ to $\mathcal
B(H).$ For an isometry $f \colon H_1 \to H_2$, the map $(f^* \circ (-) \circ
f)$ is an NCPSU morphism $\mathcal B(H_2) \to \mathcal B(H_1),$ so by defining $\mathcal B(f)
\defeq (f^* \circ (-) \circ f)^{op}$, we get that $\mathcal B(f) \colon
\mathcal B(H_1) \to \mathcal B(H_2)$ in $\NCPSU$ (recall that $\NCPSU$ is the
\emph{opposite} category of the category of von Neumann algebras with NCPSU
maps).
We remark that this functor restricts to a functor $\mathcal B \colon \Unitary \to \NMIU$
and we show that the functor $\mathcal B : \Isometry \rightarrow \NCPSU$ is strict monoidal.

\begin{lem}
\label{lem:functor}
The map $\mathcal B : \Isometry \rightarrow \NCPSU$ as defined above is a functor. 
\end{lem}
\begin{proof}
We first need to show that if $f$ is an isometry then the function $\phi \mapsto f^* \circ \phi \circ f$ is a
ncpu map from $\mathcal B(H_2)$ to $\mathcal B(H_1)$. Note that for $\phi = Id_{B(H_2)}$ we have, 
\[
	f^* \circ (Id) \circ f = f^* \circ f = Id_{B(H_1)}
\]
Hence the map is unital. To show that it is completely positive, recall that every von Neumann algebra is a
$C^*$-algebra, and hence it suffices to check an equivalent condition for complete positivity of maps defined
on $C^*$-algebra. Hence we need to show that for any $\{\phi_i \}_{i \in \mathbb{N}} \in \mathcal B(H_2)$
and $\{ \psi_i \}_{i \in \mathbb{N}} \in \mathcal B(H_1)$,
$\Sigma_{i,j}\psi_i^*(f^* \circ (\phi_i^* \phi_j) \circ f)\psi_j$ is positive. However we have,
	\[ \Sigma_{i,j}\psi_i^*(f^* \circ (\phi_i^* \phi_j) \circ f)\psi_j = (\Sigma_i \phi_i \circ f \circ \psi_i)^*(\Sigma_i \phi_i \circ f \circ \psi_i)
\]
Hence,the above summation is positive.

  To show that $\psi \eqdef f^* \circ (-) \circ f \colon \mathcal B(H_2) \to \mathcal B(H_1)$ is normal, it suffices to show that $\psi$ has a predual $\tau \colon \mathcal T(H_1) \to \mathcal T(H_2).$ But the
  predual is obviously $\tau = f \circ (-) \circ f^*.$ Indeed, one can easily see that $\tau$ is positive and therefore bounded with respect to the trace norm (it is also trace preserving).
  To check that $\tau$ is indeed the predual, we simply have to verify that for any $t \in \mathcal T(H_1)$ and $b \in \mathcal B(H_2)$ we have that
  \[ \trace{\psi(b)t} = \trace{b\tau(t)} . \]
  Indeed we have
  \[ \trace{\psi(b)t} = \trace{f^*bft} = \trace{bftf^*} = \trace{b\tau(t)} \]
  as required, so that $\psi$ is indeed normal.

Every other property of a functor is trivially true.
\end{proof}

\begin{defi}[\cite{heunen2015categories}]
\label{def:lax_monoidal}
Let $(C, \otimes_C, 1_C)$ and $(D, \otimes_D, 1_D)$ be monoidal categories. A \textbf{lax monoidal functor}
$\mathcal F : (C, \otimes_C, 1_C) \rightarrow (D, \otimes_D, 1_D)$ is
\begin{itemize}
        \item a functor $\mathcal F : C \rightarrow D$,
        \item a morphism $e : 1_D \rightarrow \mathcal F (1_C)$,
        \item a natural transformation $\mu_{x,y} : \mathcal F(x) \otimes_D \mathcal F(y) \rightarrow \mathcal F(x \otimes_C y)$,
\end{itemize}
such that:
\begin{itemize}
         \item (\textbf{associativity}) for all $x,y,z \in \textbf{Ob}(C)$, the diagram
                \[\scalebox{0.75}{
                        \stikz{associativity.tikz}
                        }
                     \]
              commutes,
       \item (\textbf{unitality})for all $x \in \textbf{Ob}(C)$, the diagrams
              \[\scalebox{0.75}{
                      \stikz{unit_D.tikz}
                      }
                \quad
                \scalebox{0.75}{
                      \stikz{unit_C.tikz}
                      }
                \]
             commute.
\end{itemize}
\end{defi}

\begin{defi}
\label{def:strong-m}
A \textbf{strict monoidal functor} is a lax monoidal functor(Definition~\ref{def:lax_monoidal}) $(\mathcal F, e, \mu_{x,y})$ such that $e$ and $\mu_{x,y}$
are identity maps. 
\end{defi}

\begin{restatable}{proposition}{strictmonoidal}
\label{lem:strict-monoidal}
The functor $\mathcal B : \Isometry \rightarrow \NCPSU$ is strict monoidal.
\end{restatable}
\begin{proof}
From $eqn(10)$ on page 185 in \cite{takesaki}, it follows that $\mathcal B(H_1) \otimes \mathcal B(H_2) \cong_{\mu_{H_1,H_2}} \mathcal B(H_1 \otimes H_2)$.
Here $\mu_{H_1,H_2}$ is just the identity.
Then for Hilbert spaces $H_1,H_2,K_1,K_2$ and isometries $f,g$ we need to show the following:
\[
	\mathcal B(f) \otimes \mathcal B(g) = \mathcal B(f \otimes g)
\]
First we show this for the generators of the von Neumann algebra $\mathcal B(H_1) \otimes B(H_2)$ which are of the form $\phi_1 \otimes \phi_2$. Note that,
\begin{align*}
	\mathcal B(f) \otimes \mathcal B(g) (\phi_1 \otimes \phi_2) &= (f^* \circ \phi_1 \circ f) \otimes (g^* \circ \phi_2 \circ g) \notag \\
	&= (f^* \otimes g^*) \circ (\phi_1 \otimes \phi_2) \circ (f \otimes g) \\
	&= (f \otimes g)^* \circ (\phi_1 \otimes \phi_2) \circ (f \otimes g) \notag \\
	&= \mathcal B(f \otimes g) (\phi_1 \otimes \phi_2) \notag
\end{align*}
This argument extends to linear combinations of elements of the form $\phi_1 \otimes \phi_2 \in \mathcal B(H_1) \otimes \mathcal B(H_2)$. Recalling the
definition of the spatial tensor product, note that $(\mathcal B(H_1) \odot \mathcal B(H_2))''$ is the closure of $\mathcal B(H_1) \otimes \mathcal B(H_2)$
in the ultraweak topology. Hence we have that the other elements of this von Neumann algebra are limits of such linear combinations in the ultraweak topology.
Let $X$ be a directed set and $\phi \triangleq lim_{x \in X}\phi_x$ be
the limit of a directed net $\{\phi_x\}_{x \in X}$,
where the above equation holds for $\phi_x$ for all $x$. Then we have the following:
\begin{align}
	\mathcal B(f) \otimes \mathcal B(g) (\phi) &= \mathcal B(f) \otimes \mathcal B(g)(lim_{x \in X} \phi_x) \\
	&= lim_{x \in X} \mathcal B(f) \otimes \mathcal B(g)(\phi_x) \\
	&= lim_{x \in X} \mathcal B(f \otimes g)  \\
	&= \mathcal B(f \otimes g) (lim_{x \in X} \phi_x)  \\
	&= \mathcal B(f \otimes g) (\phi)
\end{align}
In the above calculation, $(2)$ holds, because normal maps are continuous with respect to the ultraweak topology, and the tensor product
of normal maps is a normal map. Similarly, $(4)$ holds because composition of normal maps is a normal map.
\end{proof}

\subsection{Interpretations}
We now introduce interpretations for types, typing judgments, and configurations.
\subsubsection{Interpretation of Types.}\label{sub:interpretation-types}
We interpret types as von Neumann algebras. 
The types corresponding to a linear call-by-value
lambda-calculus are interpreted in the standard way:
$\lrb{\mathrm{I}} \defeq \mathbb{C}$, $\lrb{\mathrm{A+B}} \defeq \lrb{\mathrm{A}} \oplus \lrb{\mathrm{B}}$,
$\lrb{\mathrm{A \otimes B}} \defeq \lrb{\mathrm{A}} \otimes \lrb{\mathrm{B}}$, $\lrb{!\mathrm{A}} \defeq !\lrb{\mathrm{A}}$
and $\lrb{\mathrm{A \multimap B}} \defeq \lrb{\mathrm{A}} \multimap \lrb{\mathrm{B}}$.
The type for natural numbers is interpreted as the commutative von Neumann algebra $\ell^\infty(\mathbb N)$, i.e
$\lrb{\Nat} \defeq \ell^\infty(\mathbb{N}).$ Finally, the interpretation of the type $\mathcal B(\pQ)$ uses the
$\mathcal B$ functor that allows us to incorporate pure quantum primitives into our semantics, i.e 
$\lrb{\mathcal B (\ptype)} \defeq \mathcal B \lrb{\ptype}$.

\subsubsection{Interpretation of Typing judgments.}\label{sub:interpretation-typing-rules}
A typing judgment of the form 
$\Delta \vdash M \colon \mathrm{A}$ is interpreted as a morphism $\lrb{\Delta \vdash M \colon \mathrm{A}} \colon \lrb{\Delta} \to \lrb{\mathrm{A}}$ in $\NCPSU$ and we often abbreviate this by writing $\lrb M.$
The full interpretation of typing judgments are given in Figure~\ref{fig:typing-judgements}. We now describe
some of the notation used therein.

We use the symbol $\diamond$ for the canonical discarding map that may be defined at types $!\mathrm A$ and non-linear contexts $!\Delta.$
For a non-linear context $!\Delta$ we may define a copy (or diagonal) map $d_{!\Delta} \colon \lrb{!\Delta} \to \lrb{!\Delta} \otimes \lrb{!\Delta}$ in the standard way
and also the map
$ \textit{split} \eqdef \left( \lrb{!\Delta, \Sigma_1, \Sigma_2} \xrightarrow{d \otimes \id \otimes \id} \lrb{!\Delta, !\Delta, \Sigma_1, \Sigma_2}
  \xrightarrow{\cong} \lrb{!\Delta, \Sigma_1} \otimes \lrb{!\Delta , \Sigma_2} \right) .$

Given a non-linear context $!\Delta \defeq \{ x_1 \colon !{\mathrm{A}}_1 , \dots , x_n \colon !{\mathrm{A}}_n \}$, the \emph{promotion} map is defined as usual by
$\rho_{!\Delta} \eqdef \lrb{!\Delta} \xrightarrow{\cong} \mathcal L(\mathcal R \lrb{\mathrm{A}_1} \times$
$\cdots \times \mathcal R \lrb{\mathrm{A}_n}) \xrightarrow{\mathcal L \eta}
  !\mathcal L(\mathcal R \lrb{\mathrm{A}_1} \times \cdots \times \mathcal R \lrb{\mathrm{A}_n}) \xrightarrow{\cong} !\lrb{!\Delta}.$
We write $\gamma_{A,B} \colon A \otimes B \to B \otimes A$ for the monoidal symmetry.
The successor map is given by $s \colon \lrb{\Nat} \rightarrow \lrb{\Nat} ::  (x_1,x_2,\ldots) \mapsto (0,x_1,x_2,\ldots).$

\input{typing-judgements-c.tex}

In order to interpret measurement, we need the following theorem: 
\begin{restatable}{thm}{basis}
\label{lem:basis_set}
For every pure quantum type $\pQ$, there exist a set $X$, an NMIU isomorphism 
$\alpha_{\ov{\pQ}} : \lrb{\ov{\pQ}} \cong \ell^\infty(X)$, and an isometric isomorphism $\beta_{\pQ} : \lrb{\pQ} \cong \ell^2(X)$. Moreover, $ \forall \ps{b} \in \pbterms$, $ \exists x \in X$ s.t.
$\beta_{\pQ}(\lrb{\ps{b}}) = \alpha_{\ov{\pQ}}(\lrb{\ov{\ps{b}}}) = \ket{x}$.
\end{restatable}
\begin{proof}(Sketch)
      Proof is by induction on the form of $\pQ$.
      \begin{itemize}
	  \item $\pQ = \ps I$:\\
		  \[ \lrb{\ps I} = \mathbb{C} \cong \ell^2(\{ \ast \}) 
		     \quad
		  \lrb{\ov{\ps I}} = \mathbb{C} \cong \ell^\infty(\{ \ast \})
		  \]
		  Moreover , in this case we have $\ps b = \ps \ast$, and $\ov{\ps b} = \ast$. We also know that $\lrb{\ps \ast} = \lrb{\ov{\ps b}} = 1 \in \mathbb{C}$.
		  However, the scalar $1$ can be identified with $\ket \ast$. Hence the second statement of the theorem also holds.
		   	   \item $\pQ = \pQ_1 \ps \oplus \pQ_2$:\\
		   By induction hypothesis we have that $\lrb{\ov{\pQ_1}} \cong_{\alpha_{\ov{\pQ_1}}} \ell^\infty(X_1)$
                  and $\lrb{\pQ_1} \cong_{\beta_{\pQ_1}} \ell^2(X_1)$ for some set $X_1$.
                  Similarly, for $\pQ_2$ we have, $\lrb{\ov{\pQ_2}} \cong_{\alpha_{\ov{\pQ_2}}} \ell^\infty(X_2)$ 
                  and $\lrb{\pQ_2} \cong_{\beta_{\pQ_2}} \ell^2(X_2)$ for some set $X_2$.
                  Consider the following:
                      \[
			      \lrb{\pQ_1 \ps \oplus \pQ_2} = \lrb{\pQ_1} \oplus \lrb{\pQ_2} \cong_{\beta_{\pQ_1} \oplus \beta_{\pQ_2}} \ell^2(X_1) \oplus \ell^2(X_2) 
			      \cong \ell^2(X_1 \sqcup  X_2)
                      \]
                      \[
			      \lrb{\ov{\pQ_1 \ps \oplus \pQ_2}} = \lrb{\ov{\pQ_1}} \oplus \lrb{\ov{\pQ_2}} \cong_{\alpha_{\ov{\pQ_1}} \oplus \alpha_{\ov{\pQ_2}}} \ell^\infty(X_1) \oplus \ell^\infty(X_2) 
			      \cong \ell^\infty(X_1 \sqcup X_2)
                      \]
		      The second statement holds because $\beta_{\pQ_1 \ps \oplus \pQ_2}(\lrb{\inl{\ps b_i}}) = \tinl{\beta_{\pQ_1}(\lrb{\ps b_i})} = \ket{x_i}^{ext}$. 
		      Here $\ket{x_i}^{ext}$ is just the extension of the functional $\ket{x_i}$ defined on $X_1$ to $X_1 \sqcup X_2$ by defining it to be identically zero on
		      whole of $X_2$. Same is true for the right injection.
		      Similar statement holds for $\{\ov{b_i} \}_{i \in I}$.
	  \item $\pQ = \pQ_1 \ps \otimes \pQ_2$:\\
		  By induction hypothesis we have that $\alpha_{\ov{\ptype_1}} :\lrb{\ov{\pQ_1}} \rightarrow \ell^\infty(X_1)$
		  and $\beta_{\pQ_1} : \lrb{\pQ_1} \rightarrow \ell^2(X_1)$ for some set $X_1$. Furthermore, for all $i$, we have the relation 
		      $\beta_{\pQ_1}(\lrb{\ps{b_i}}) = \alpha_{\ov{\pQ_1}}(\lrb{\ov{\ps{b_i}}}) = \ket{x_i}$ for $x_i \in X_1$.
		  Similarly, for $\pQ_2$ we have, $\alpha_{\ov{\ptype_2}} :\lrb{\ov{\pQ_2}} \rightarrow \ell^\infty(X_2)$
                  and $\beta_{\pQ_2} : \lrb{\pQ_2} \rightarrow \ell^2(X_2)$ for some set $X_2$. We also have that
		      $\beta_{\pQ_2}(\lrb{\ps{b_i'}}) = \alpha_{\ov{\pQ_2}}(\lrb{\ov{\ps{b_i'}}}) = \ket{x'_i}$ for $x'_i \in X_2$.
		  Consider the following:
		      \[
			      \lrb{\pQ_1 \ps \otimes \pQ_2} = \lrb{\pQ_1} \otimes \lrb{\pQ_2} \cong_{\beta_{\pQ_1} \otimes \beta_{\pQ_2}} \ell^2(X_1) \otimes \ell^2(X_2) 
			      \cong \ell^2(X_1 \times X_2)
		      \]
		      \[
			      \lrb{\ov{\pQ_1 \ps \otimes \pQ_2}} = \lrb{\ov{\pQ_1}} \otimes \lrb{\ov{\pQ_2}} \cong_{\alpha_{\ov{\pQ_1}} \otimes \alpha_{\ov{\pQ_2}}} \ell^\infty(X_1) \otimes \ell^\infty(X_2) 
			      \cong \ell^\infty(X_1 \times X_2)
		      \]
		      The second statement holds because $\beta_{\pQ_1 \ptimes \pQ_2}(\lrb{\ps b_i \ptimes \ps b_j'}) = \beta_{\pQ_1}(\lrb{\ps b_i}) \otimes \beta_{\pQ_2}(\lrb{\ps b_j'}) = 
		      \ket{x_i} \otimes \ket{x_j'}$.
		      Similar statement holds for $\{\ov{b_i} \}_{i \in I}$.
		      	  \item $\pQ = \pnat$:\\
		  \[ \lrb{\pnat} = \ell^2(\mathbb N)
		     \quad
		  \lrb{\ov{\pnat}} = \lrb{\Nat} = \ell^\infty(\mathbb N)
		  \]
		  In this case the isomorphisms are identity maps, and we know that $\{ \ps b_i \}_{i \in \mathbb{N}} = 
		  \{\ket n \}_{n \in \mathbb{N}}$ because of Definition~\ref{def:orthogonality}.
		   Similarly, $\{ \ov{\ps{b_i}} \}_{i \in \mathbb{N}} = \{\ket n \}_{n \in \mathbb{N}}$. Hence, the second statement of the theorem also holds.
      \end{itemize}
	Note that this argument works because we deal with separable Hilbert spaces only, and hence the set $\{\ps b_i \}_{i \in I}$ is always countable.
\end{proof}

\noindent Using the above theorem, measurement is interpreted as the map $m^{\ptype} : \mathcal B(\lrb{\pQ}) \rightarrow \lrb{\ov{\pQ}}$ defined as
$m^{\ptype}\defeq \alpha_{\ov{\pQ}}^{op} \circ m_X^{op} \circ \mathcal B(\beta_{\pQ})$.

\subsubsection{Interpretation of Configurations.}\label{sub:interpretation-q-config}
A well-formed configuration $(\pt, u_\sigma, M) $ with $\cdot \pdash \pt \colon \pQ$, $\udash u_\sigma \colon (\pQ, \pQ_1' \ptimes \cdots \ptimes \pQ_m')$, and
$x_1 \colon \mathcal B(\pQ_1'), \ldots, x_m \colon \mathcal B(\pQ_m') \vdash M \colon \mathrm{A}$ is interpreted as the NCPSU morphism
\noindent
\[
    \lrb{(\pt, u_{\sigma}, M):\mathrm{A}} \eqdef \left( \mathbb{C} \xrightarrow{\mathcal B(\lrb{u_{\sigma} \ps t})} \begin{array}{c} \mathcal B(\ptype_1' \ps \otimes \cdots \ps \otimes \ptype_m') \\
    \shortparallel\\  \mathcal B(\ptype_1') \otimes \cdots \otimes \mathcal B(\ptype_m') \end{array} \xrightarrow{\lrb{M}} 
   \mathrm{A}\right)\! . 
\]

We naturally extend the interpretation defined for a \emph{total} configuration above to a configuration
with auxiliary quantum types as follows:
\[
\scalebox{0.8}{\stikz{quantum_config.tikz}}
\]

\subsection{Soundness and Adequacy}
We can now show that our semantic interpretation is sound with respect to single-step reduction. In order to do so, we first state and prove some auxilliary lemmas.

\begin{lem}[\textit{Substitution}]
\label{lem:subst-d}
Given typing judgments $!\Delta, \Sigma_1 \vdash V : A$ ,$!\Delta, \Sigma_2, x: A \vdash M : B$
and $!\Delta, \Sigma_1 , \Sigma_2 \vdash M[V/x] : B$ (Lemma~\ref{lem:subst-o}), the following
equation holds:
\[\scalebox{0.9}{$      
        \lrb{!\Delta, \Sigma_1, \Sigma_2 \vdash M[V/x]:B} = \lrb{!\Delta, \Sigma_2, x:A \vdash M : B}
        \circ
        (\lrb{!\Delta, \Sigma_1 \vdash V : A} \otimes id) \circ \textit{split}$}
        \]
\end{lem}
\begin{proof}
        Proof is by induction on the form of the term $M$.
\end{proof}

\begin{lem}
The following equations hold for valid typing judgments:
\begin{align*}
	\lrb{\Delta \vdash (\lambda x.M)V : A}& = \lrb{\Delta \vdash M[V/x] : A}\\
	\lrb{\Delta \vdash \tletpair{x}{y}{V \otimes W}{M} : A} &= \lrb{\Delta \vdash M[V/x,W/y] : A}\\
	\lrb{\Delta \vdash \tcase{inl(V)}{M}{N} : A} &= \lrb{\Delta \vdash M[V/x] : A} \\
	\lrb{\Delta \vdash \tcase{inr(W)}{M}{N} : A} &= \lrb{\Delta \vdash M[W/y] : A}\\
	\lrb{\Delta \vdash \tforce{\tlift{M}} : A} &= \lrb{\Delta \vdash M : A} \\
	\lrb{\Delta \vdash \tmatch{\tzero}{M}{N} : A} &= \lrb{\Delta \vdash M : A}\\
	\lrb{\Delta \vdash \tmatch{\tsucc{V}}{M}{N} : A} &= \lrb{\Delta \vdash N[V/x] : A}
\end{align*}
\end{lem}
\begin{proof}
	Follows from Lemma~\ref{lem:subst-d}.
\end{proof}

\begin{lem}
\label{lem:interpretation-measurement}
       For a quantum type $\pQ$ and finitely many basis terms $\{\ps{b_i}\}_{i \in I}$ of $\pQ$, the following equation holds:
	   \[
		   \mathcal B^r (\lrb{\Sigma_i p_i \ps{b_i}}) \circ \mathcal
		   B^r(\beta_{\pQ}) \circ m_X \circ \alpha_{\ov{\pQ}} = \Sigma_i
		   |p_i|^2 \lrb{\ov{\ps{b_i}}}
	   \]
       provided that $\cdot \pdash \Sigma_i p_i \ps{b_i}.$
Here $\mathcal B^r: \Isometry \rightarrow \textbf{W}$ is defined as $\mathcal B^r(H) \defeq \mathcal B(H)$ on objects, and for a morphism
$f \in \Isometry$, $\mathcal B^r(f)(\phi) \defeq f^* \circ \phi \circ f$ and $\NCPSU = \textbf{W}^{op}$.
\end{lem}
\begin{proof}
 It follows from the definition of the functor $\mathcal B^r$, that it suffices to show the following:
        \[
                \mathcal B^{r} (\lrb{\ps{b_i}}) \circ \mathcal B^r(\beta_{\pQ}) \circ m_X \circ \alpha_{\ov{\pQ}} = \lrb{\ov{\ps{b_i}}}
        \quad
        \forall i
       \]
        This follows from the definition of $\mathcal B^r , m_X$, Theorem~\ref{lem:basis_set} and the
        observation that $\{\ps{b_i}\}_{i \in I} = \{ \ket{x}\}_{x \in X}$.
\end{proof}

\begin{prop}
\label{lem:strong-small-step}
For a well-formed quantum configuration $\mathcal C : A[\ptype'_1, \ldots,\ptype'_k]$, we have 
$\lrb{\mathcal C : A[\ptype'_1, \ldots, \ptype'_k]} = \Sigma_{\mathcal C'} p \lrb{\mathcal C' : A[\ptype'_1, \ldots, \ptype'_k]}$.
\end{prop}
\begin{proof}
Note that, the proof essentially reduces to
proving the statement for the relation $\rightarrow_p$, since $\leadsto_p$ holds after we rewrite the first component. More precisely,
if $\mathcal C$ doesn't reduce further, the statement holds trivially. If it does, then since the equational theory for the quantum
control fragment is complete(Theorem~\ref{th:complete}), it suffices to show the statement for $\rightarrow_p$.
We can assume without loss of generality that $\mathcal C = (\ps t, u_{\sigma} , M)$ for some
quantum value $\ps t$, some unitary $u_{\sigma}$ and some term $M$.
Now the proof is by induction on the form of the term $M$.
Note that if $M$ is a value, then the above equation holds trivially.
Let $\mathcal C = (\ps t, u_{\sigma}, M\, N)$, where $M$ is not a value. Then we know from the contextual rule fo Figure~\ref{fig:red_uni} that the only possible reductions
take the form $(\ps t , u_{\sigma}, M\, N) \rightarrow_{p_i} (\ps t_i , u_{\sigma_i} , M_i\, N)$ whenever $(\ps t , u_{\sigma}, M) \rightarrow_{p_i} (\ps t_i , u_{\sigma_i}, M_i)$. By the induction
hypothesis , we know that $\lrb{(\ps t , u_{\sigma}, M)} = \Sigma_ip_i \lrb{(\ps t_i , u_{\sigma_i}, M_i)}$. Hence we have the following:
\begin{align*}
       & \lrb{(\ps t , u_{\sigma}, M\,N) : \mathrm{A}[\ptype'_1, \ldots,\ptype'_k]}\\
       &= \lrb{x_1 ,\ldots,x_m:\mathcal B(\ptype_m) \vdash M\,N : \mathrm{A}} \otimes id_{\mathcal B(\ptype'_1)\ldots}
        \circ \mathcal B (\lrb{u_{\sigma}\ps t})\\
        &= ((\epsilon_{\mathrm{A,B}} \circ (\lrb{x_1,\ldots ,x_j \vdash M : \mathrm{B \multimap A}} \otimes 
        \lrb{x_{j+1},\ldots \vdash N:\mathrm{B}})) \otimes 
        id_{\mathcal B(\ptype'_1)\ldots}) \circ \mathcal B (\lrb{\ps v})\\
        &= ((\epsilon_{\mathrm{A,B}}\otimes id_{\mathcal B(\ptype'_1)\ldots}) \circ (id_{\mathrm{B \multimap A}} \otimes \lrb{x_{j+1},\ldots \vdash N:\mathrm{B}} \otimes id_{\mathcal B(\ptype'_1)\ldots})
	\\
	&\quad\circ(\lrb{x_1,\ldots ,x_j \vdash M : \mathrm{B \multimap A}} \otimes id_{\mathcal B(\ptype_{j+1}) \ldots} \otimes 
        id_{\mathcal B(\ptype'_1)\ldots}) \circ \mathcal B (\lrb{u_{\sigma}\ps t}) \\
        &= ((\epsilon_{\mathrm{A,B}}\otimes id_{\mathcal B(\ptype'_1)\ldots}) \\
        &\quad \circ (id_{\mathrm{B \multimap A}} \otimes \lrb{x_{j+1},\ldots \vdash N:\mathrm{B}} \otimes id_{\mathcal B(\ptype'_1)\ldots})
        \circ
        (\Sigma_ip_i \lrb{(\ps t_i , u_{\sigma_i}, M_i)}) \otimes (id_{\mathcal B(\ptype'_1)\ldots}) \\ 
        &= \Sigma_i p_i \lrb{(\ps t_i , u_{\sigma_i} , M_iN)}
\end{align*}
        The case when $M\,N = (\lambda x.M')\, V$ is a consequence of ~\ref{lem:subst-d}.
        The other cases in the induction step for $M\, N$ follow similarly.
        We now prove the induction step for the term $\tpure{\ps t}$. From Figure~\ref{fig:red_uni}, it follows that
we need this equation to hold:
        \[\scalebox{1}{
                $\lrb{(\ps t , u_{\sigma} , \tpure{\ps t'}) : \mathcal B(\ptype_{\ps t'})[\ptype_{\ps t}]} = \lrb{(\ps t \otimes \ps t' , u_{\sigma_{swap}} \circ (u_{\sigma} \otimes u_{id}) , x) 
        : \mathcal B(\ptype_{\ps t'})[\ptype_{\ps t}]}$
                }\]
Note the following:
\begin{align*}
\lrb{(\ps t , u_{\sigma} , \tpure{\ps t'}) : \mathcal B(\ptype_{\ps t'})[\ptype_{\ps t}]}&= (\lrb{\cdot\vdash \tpure{\ps t'} : \mathcal B(\ptype_{\ps t'})}\otimes id_{\mathcal B(\ptype_{\ps t})})\circ \mathcal B(\lrb{u_{\sigma}\ps t})\\
        &= (\mathcal B(\lrb{\ps t'}) \otimes id_{\mathcal B(\ptype_{\ps t})})\circ \mathcal B(\lrb{u_{\sigma}\ps t}) \\
        &= (id_{\mathcal B(\ptype_{\ps t'})} \otimes id_{\mathcal B(\ptype_{\ps t})})\circ \mathcal B(\lrb{u_{\sigma_{swap}} \circ (u_{\sigma} \otimes u_{id}) (\ps t \otimes \ps t')})\\
&= (\lrb{x}\otimes id_{\mathcal B(\ptype_{\ps t})}) \circ \mathcal B(\lrb{u_{\sigma_{swap}} (u_{\sigma} \ps t) \otimes \ps t'}) \\
&= \lrb{(\ps t \otimes \ps t' , u_{\sigma_{swap}} \circ (u_{\sigma} \otimes u_{id}), x) 
        : \mathcal B(\ptype_{\ps t'})[\ptype_{\ps t}]}
\end{align*}
       The proof for the induction step for the term $\tmeas{M}$ follows from Lemma~\ref{lem:interpretation-measurement}. Rest of the cases follow
       using arguments similar to the above cases.
\end{proof}

\begin{thm}[Soundness]
	\label{thm:small-step-sound}
	For $\mathcal C \in \wfconfs{\svconfs}{A}$, if $\mathcal C \notin \vconfs$, then
	\noindent
	\[ 
	\lrb{\mathcal C : \mathrm{A}} = \sum_{\mathcal C \leadsto_p \mathcal C'} p \lrb{\mathcal C' : \mathrm{A}} , 
	\]
	where the sum ranges over all possible reducts $\mathcal C \leadsto_p \mathcal C'$.
\end{thm}
\begin{proof}
	Follows from Proposition~\ref{lem:strong-small-step}.
\end{proof}

Next, we can show that our interpretation is sound in a big-step sense as well, which follows easily using the previous result and strong normalization.

\begin{thm}[Strong Adequacy]
For $\mathcal C \in \wfconfs{\svconfs}{A}$,
\noindent
	\[
		\lrb{\mathcal C} = \sum_{\mathcal V \in \vconfs}P(\mathcal C \rightarrow_{\ast} \mathcal V)\lrb{\mathcal V},
	\]
	where $P(\mathcal C \rightarrow_{\ast} \mathcal V)$ indicates the overall probability that $\mathcal C$ reduces to $\mathcal V$.
\end{thm}
 \begin{proof}
 The result is a consequence of Theorem~\ref{thm:st-norm} and Theorem~\ref{thm:small-step-sound}.
 \end{proof}
 
 This implies that, for any well-formed configuration $\mathcal C$, its interpretation $\lrb{\mathcal C}$ is an NCPU map (\emph{i.e.}, it is not merely subunital) and therefore it corresponds to a quantum channel in the Heisenberg picture of quantum mechanics.

\section{Conclusion and Perspectives}\label{sec:conclusion}

We described a programming language which has support for both pure state quantum computation
and mixed state quantum computation. We began by describing the pure quantum subsystem (\secref{sec:quantum-control}),
for which we introduced an equational theory and proved its completeness w.r.t. the denotational semantics (\secref{sub:quantum-equational}).
Then, we described the main calculus (\secref{sec:syntax}) which uses a new adaptation of quantum configurations (\secref{sub:operational-c}),
which we used to define the operational semantics of the language. Here, the interaction between the pure quantum subsystem and the
main calculus becomes apparent. We showed that our  
denotational semantics (\secref{sub:denotational}) has a clear and appropriate interpretation as channels in the Heisenberg picture of quantum mechanics
and we showed it is sound and adequate with respect to the operational semantics (\secref{sub:denotational}).

Although our language has support for $\Nat$, an inductive type, an obvious
extension would be to include more general inductive types.
In fact, the part of the categorical model dealing with mixed state computation already supports inductive types and we believe that they can also be added to the syntax and operational semantics in a straightforward manner (see \cite{qpl-fossacs}). The part of the categorical model dealing with pure state computation also supports inductive types, but it is less clear how to add inductive types to the language for pure state computation. 
Furthermore, our language lacks recursion to make full use of these inductive
types. In the pure case, \emph{structural recursion} is
possible~\cite{sabry2018symmetric}; general recursion is, however, a difficult
open problem~\cite{louis-thesis, pablo-thesis}. Adding recursion to the
language for mixed state computation is straightforward, but unfortunately, it
is unclear (to us) how to adapt the categorical model to support general
recursion while keeping lambda abstractions in the language and retaining nice
meta-theoretical properties. If we remove lambda abstractions and instead
consider a first-order language for mixed computation, then the categorical
model supports recursion \cite{qpl-fossacs}. The main problem appears to be
that Currying (via the structure of the Kleisli exponential) is not a
Scott-continuous operation when we equip the category $\NCPSU$ with the
appropriate $\DCPO$-enrichment based on the Löwner order.

Additionally, support for higher-order \emph{pure} quantum computation is
another feature that would be interesting to add in a future work. An approach
only exists in the classical case~\cite{chard:invrec}; whether this generalizes
to the quantum case is an open question.

Another direction for future work is compilation. On the pure side, the most
costly operation to compile will be pattern matching, and could be done
similarly to how Qunity~\cite{qunity, qunity-compile} deals with it. On the mixed side, compilation
of the quantum lambda calculus has been an open problem, although some recent
unpublished work tackles it~\cite{chardonnet2026compiling}.

%% file: syntax-c.tex
%
%
\begin{figure*}[t]
	\begin{tabular}{l l l l}
		(\cvariables) & \multicolumn{3}{l}{$x,y,z$}\\
		(\ctypes)  & $\mathrm{A, B}$ & ::= & $\mathrm{I} \alt \mathrm{A + B} \alt
		\mathrm{A\otimes B} \alt \ ! \mathrm{A} \alt \mathrm{A \multimap B} \alt
		\Nat \alt \mathcal B(\ptype)$\\
		(\cterms)  & $L,M,N$ & ::= & $ \ast \alt x \alt \tinl{M} \alt \tinr{N} \alt \tcase{L}{M}{N} $\\
		& & & $\alt M \otimes N \alt \tletpair{x}{y}{M}{N} \alt \tlift{M} \alt \tforce{M}$ \\
		& & & $\alt \lambda x.M \alt M\, N $\\
		& & & $\alt\tzero \alt \tsucc{M} \alt \tmatch{L}{M}{N}$ \\
		& & & $\alt  \tpure{\pt} \alt \tmeas{M} \alt \tun{\isoterm}{M}$\\
		& & & $\alt\tletb{z}{M}{N} \alt \tletpairb{x}{y}{M}{N}$ \\
		(\cvalues)  & $V, W$ & ::= & $\ast \alt x \alt \tinl{V} \alt \tinr{W} \alt V \otimes W \alt \tlift{M} \alt \lambda x.M$ \\
		& & & $\alt \tzero \alt \tsucc{V}$
	\end{tabular}
\caption{Syntax of the classically controlled system.}
\label{fig:syntax-c}
\end{figure*}

%% file: typing-rules.tex
\begin{figure*}[t]
	\begin{center}
			\bottomAlignProof
			\scalebox{1}{
				\begin{bprooftree}
					\AxiomC{$!\Delta , x : \mathrm{A} , y : \mathrm{B} , \Sigma \vdash M : \mathrm{C}$}
					\UnaryInfC{$!\Delta , y : \mathrm{B} , \textit{x} : \mathrm{A} , \Sigma \vdash M : \mathrm{C}$}
				\end{bprooftree}
				}
				\bottomAlignProof
				\scalebox{1}{
					\begin{bprooftree}
						\AxiomC{}
						\UnaryInfC{$!\Delta \vdash * : \mathrm{I}$}
					\end{bprooftree}
				}
				\bottomAlignProof
				\scalebox{1}{
					\begin{bprooftree}
						\AxiomC{}
						\UnaryInfC{$!\Delta , x : \mathrm{A} \vdash$ $x : \mathrm{A}$}
					\end{bprooftree}
				}
				\\[0.2cm]
				\bottomAlignProof
				\scalebox{1}{
					\begin{bprooftree}
						\AxiomC{$ \Delta \vdash M : \mathrm{A}$}
						\UnaryInfC{$\Delta \vdash \tinl{M} : \mathrm{A + B}$}
					\end{bprooftree}
				}
				\quad
				\bottomAlignProof
				\scalebox{1}{
					\begin{bprooftree}
						\AxiomC{$ \Delta \vdash N : \mathrm{B}$}
						\UnaryInfC{$\Delta \vdash \tinr{N} : \mathrm{A + B}$}
					\end{bprooftree}
				}
				\\[0.2cm]
				\bottomAlignProof
				\scalebox{0.95}{
					\begin{bprooftree}
						\AxiomC{$!\Delta , \Sigma_1 \vdash L : \mathrm{A+B}$}
						\AxiomC{$!\Delta , \Sigma_2 , x : \mathrm{A} \vdash M : \mathrm{C}$}
						\AxiomC{$!\Delta , \Sigma_2 , y : \mathrm{B} \vdash N : \mathrm{C}$}
						\TrinaryInfC{$!\Delta , \Sigma_1 , \Sigma_2 \vdash \tcase{L}{M}{N} : \mathrm{C} $}
					\end{bprooftree}
				}
				\\[0.2cm]
				\scalebox{1}{
					\begin{bprooftree}
						\AxiomC{$!\Delta, \Sigma_1 \vdash M : \mathrm{A}$ \quad $!\Delta, \Sigma_2 \vdash N : \mathrm{B}$}
						\UnaryInfC{$!\Delta, \Sigma_1 , \Sigma_2 \vdash  M \otimes N : \mathrm{A \otimes B}$}
					\end{bprooftree}
				}
				\\[0.2cm]
				\scalebox{1}{
					\begin{bprooftree}
						\AxiomC{$!\Delta, \Sigma_1 \vdash M : \mathrm{A \otimes B}$ \ \ $!\Delta, \Sigma_2, x: \mathrm{A}, \textit{y} :\mathrm{B} \vdash N : \mathrm{C}$}
						\UnaryInfC{$!\Delta, \Sigma_1 , \Sigma_2 \vdash \tletpair{x}{y}{M}{N} : \mathrm{C}$}
					\end{bprooftree}
				}
				\\[0.2cm]
				\scalebox{1}{
					\begin{bprooftree}
						\AxiomC{$!\Delta \vdash M : \mathrm{A}$}
						\UnaryInfC{$!\Delta \vdash\tlift{M} :\ !\mathrm{A}$}
					\end{bprooftree}
				}
				\quad
				\scalebox{1}{
					\begin{bprooftree}
						\AxiomC{$\Delta \vdash M :\ !\mathrm{A}$}
						\UnaryInfC{$\Delta \vdash \tforce{M} : \mathrm{A}$}
					\end{bprooftree}
				}
				\\[0.2cm]
				\quad
				\scalebox{0.95}{
					\begin{bprooftree}
						\AxiomC{$\Delta ,x : \mathrm{A} \vdash M : \mathrm{B}$}
						\UnaryInfC{$\Delta \vdash \lambda x.M : \mathrm{A \multimap B}$}
					\end{bprooftree}
				}
\quad
				\scalebox{1}{
					\begin{bprooftree}
						\AxiomC{$!\Delta , \Sigma_1 \vdash M : \mathrm{A \multimap B}$ \ \ $!\Delta , \Sigma_2 \vdash N : \mathrm{A} $}
						\UnaryInfC{$!\Delta , \Sigma_1 , \Sigma_2 \vdash M\, N : \mathrm{B}$}
					\end{bprooftree}
				}
				\\[0.2cm]
					\scalebox{1}{
						\begin{bprooftree}
							\AxiomC{$\vphantom\Sigma$}
							\UnaryInfC{$!\Delta \vdash \tzero : \Nat$}
						\end{bprooftree}
					}
					\scalebox{1}{
						\begin{bprooftree}
							\AxiomC{$\Delta \vdash M : \Nat$}
							\UnaryInfC{$\Delta \vdash \tsucc{M} : \Nat$}
						\end{bprooftree}
					}
				\\[0.2cm]
					\scalebox{1}{
						\begin{bprooftree}
							\AxiomC{$!\Delta , \Sigma_1 \vdash L : \Nat \ \ !\Delta , \Sigma_2 \vdash M : \mathrm{A}$ \ \ $!\Delta , \Sigma_2 , x: \Nat \vdash N : \mathrm{A}$}
							\UnaryInfC{$!\Delta , \Sigma_1 , \Sigma_2 \vdash \tmatch{L}{M}{N} : \mathrm{A}$}
						\end{bprooftree}
					}
					\scalebox{1}{
						\begin{bprooftree}
							\AxiomC{$ \cdot \entailpure \ps t : \ptype $}
							\UnaryInfC{$!\Delta \vdash \tpure{\ps t} : \mathcal B(\ptype)$}
						\end{bprooftree}
					}
				\\[0.2cm]
					\scalebox{1}{
						\begin{bprooftree}
							\AxiomC{$\Delta \vdash M : \mathcal B(\ptype)$}
							\UnaryInfC{$\Delta \vdash \tmeas{M} : \ov \ptype$}
						\end{bprooftree}
					}
					\scalebox{1}{
						\begin{bprooftree}
							\AxiomC{$ \udash \isoterm : \isotype{\ptype_1}{\ptype_2}$}
							\AxiomC{ $!\Delta , \Sigma_1 \vdash M : \mathcal B (\ptype_1)$}
							\BinaryInfC{$!\Delta , \Sigma_1 \vdash \tun{\isoterm}{M} : \mathcal B (\ptype_2)$}
						\end{bprooftree}
					}
				\\[0.2cm]
				\scalebox{1}{
					\begin{bprooftree}
						\AxiomC{$!\Delta , \Sigma_1 \vdash M : \mathcal B(\ptype_1) \otimes \mathcal B(\ptype_2)$}
						\AxiomC{$!\Delta , \Sigma_2 , z: \mathcal B(\ptype_1 \ps \otimes \ptype_2) \vdash N : \mathrm{A}$}
						\BinaryInfC{$!\Delta , \Sigma_1 , \Sigma_2 \vdash \tletb{z}{M}{N} : \mathrm{A}$}
					\end{bprooftree}
				}
				\\[0.2cm]
					\scalebox{1}{
						\begin{bprooftree}
							\AxiomC{$!\Delta , \Sigma_1 \vdash M : \mathcal B(\ptype_1 \ps \otimes \ptype_2)$}
							\AxiomC{$!\Delta , \Sigma_2 , x: \mathcal B(\ptype_1), y : \mathcal B(\ptype_2) \vdash N : \mathrm{A}$}
							\BinaryInfC{$!\Delta , \Sigma_1 , \Sigma_2 \vdash \tletpairb{x}{y}{M}{N} : \mathrm{A}$}
						\end{bprooftree}
					}
			\end{center}
			\caption{Typing rules for terms.}
			\label{fig:typing-rules-excerpt}
		\end{figure*}

%% file: well-formed.tex
\usetikzlibrary{shapes.geometric}
\begin{figure}{h}
	\begin{center}
		\scalebox{0.5}{
			\begin{quantikz}[background color=gray!10, row sep=0.25cm]
				\gate[7, nwires=4 ,style={tri, fill=white,scale=0.9,xshift=.2cm}]{\rombox{\scalebox{2.5}{\textbf{t}}}\qquad  \qquad} & & \push{\ \ptype_1\ }
				& &\gategroup[7,steps=3,style={dashed,rounded
				corners,fill=gray!10, inner
				xsep=2pt, scale=1.4},background, label style={label
				position=center,anchor=south, yshift=-0.2cm}]{\rombox{\scalebox{2.5}{$\isoterm_{\sigma}$}}} & \permute{3,1,2} & & & & &\gate[7, nwires=4,style={fill=white,scale=1.5,xshift=0.1cm},label style={label
				position=north,anchor=mid, xshift=0.2cm, yshift=-0.2cm}][3.8cm]{\hspace*{0.5cm}\rombox{\scalebox{2.5}{$M : \mathrm{A}$}}}\ \gateinput[2]{\scalebox{1.8}{$x_1 : \mathcal B(\ptype_2 \otimes \ptype_3)$}} \\
				& &\push{\ \ptype_2\ } &  & &   & & & &&\\
				& &\push{\ \ptype_3\ }& &&&&& &&\gateinput{\scalebox{1.8}{$x_2: \mathcal B(\ptype_1)$}}\\
				\setwiretype{n}  & & \wvdots & &  & \wvdots &  & \wvdots & \wvdots & & &&\\
				\setwiretype{n}  & & \wvdots & &  & \wvdots &  & \wvdots & \wvdots & & &&\\
				& &\push{\ \ptype_{n-2}\ } & & & \permute{3,1,2} &&& &&\gateinput[2]{\scalebox{1.8}{$x_{m-1}: \mathcal B(\ptype_{n-1} \otimes \ptype_n)$}}\\
				& &\push{\ \ptype_{n-1}\ } & &&&&& &&\\
				& &\push{\ \ptype_n\ }& &&&&&& &\gateinput{\scalebox{1.8}{$x_m : \mathcal B(\ptype_{n-2})$}}
			\end{quantikz}
			}
			\caption{Graphical representation of a well-formed quantum configuration.}
			\label{fig:well-formed}
	\end{center}
\end{figure}

%% file: unitaries-rearrange.tex
\begin{figure*}[!t]
	\centering
	\begin{subfigure}[b]{0.33\textwidth}
		\[                                                                                                   
		\scalebox{0.68}{                                                                                      
		\begin{quantikz}[background color=gray!10]
			\lstick{$\ps x_1$}&
			&\gategroup[10,steps=3,style={dashed,rounded
			corners,fill=gray!10, inner
			xsep=2pt},background,label style={label
			position=below,anchor=north,yshift=-0.2cm}]{$\isoterm_{\sigma_{gather}}$} & & && \rstick{$\ps x_1$}\\
			&\setwiretype{n} \wvdots & & \wvdots & &\wvdots& \\
			\lstick[3]{$\ps x$} & & & & & &\rstick[6]{$\ps{x \ptimes y}$}\\
			&\setwiretype{n} \wvdots & & \wvdots & &\wvdots &\\
			&&&&&&\\
			\lstick[3]{$\ps y$}&&&&&&\\
			&\setwiretype{n} \wvdots & & \wvdots &   &\wvdots & \\
			&&&&&&\\
			&\setwiretype{n} \wvdots &   & \wvdots &  &\wvdots &\\
			\lstick{$\ps x_n$} &&&&& &\rstick{$\ps x_n$}
		\end{quantikz}
		}
		\]
		\[
			\scalebox{0.6}{
				$u_{\sigma_{gather}} \defeq \{ \mid\ps{x_1 \ptimes \dots \ptimes x \ptimes y \ptimes \dots  \ptimes x_n}$}
			\]
			\[
				\scalebox{0.6}{
					$ \quad \iso \ps{x_1 \ptimes \dots \ptimes (x \ptimes y) \ptimes \dots \ptimes x_n} \}$
				}
			\]
			\label{fig:unitary_li}
	\end{subfigure}
	\begin{subfigure}[b]{0.25\textwidth}
		\[
			\scalebox{0.68}{                                                                                      
			\begin{quantikz}[background color=gray!10]
				\lstick{$\ps x_1$} 
				&\gategroup[4,steps=3,style={dashed,rounded
				corners,fill=gray!10, inner
				xsep=2pt},background,label style={label
				position=below,anchor=north,yshift=-0.2cm}]{$\isoterm_{\sigma_s}$} &\permute{2,1} & & \rstick{$\ps x_s$}\\
				\lstick{$\ps x_s$} &&&& \rstick{$\ps x_1$}\\
				\wvdots &\setwiretype{n}& \wvdots & & \wvdots &\\
				\lstick{$\ps x_n$} &&&& \rstick{$\ps x_n$}
			\end{quantikz}
			}
		\]
		\[
			\scalebox{0.6}{
				$u_{\sigma_{s}} \defeq \{\mid\ps{x_1 \ptimes \cdots \ptimes x_s \ptimes \cdots \ptimes x_n} $}
			\]
			\[
				\scalebox{0.6}{
					$\quad \iso\ps{x_s \ptimes x_1 \ptimes \cdots \ptimes x_n} \}$
				}
			\]
			\\[0.1cm]
			\[
				\scalebox{0.68}{                                                                                      
				\begin{quantikz}[background color=gray!10]
					\lstick{$\ps x$} 
					&\gategroup[2,steps=3,style={dashed,rounded
					corners,fill=gray!10, inner
					xsep=2pt},background,label style={label
					position=below,anchor=north,yshift=-0.2cm}]{$\isoterm_{\sigma_{swap}}$} &\permute{2,1} & & \rstick{$\ps y$}\\
					\lstick{$\ps y$} &&&& \rstick{$\ps x$}
				\end{quantikz}
				}
			\]
			\[
				\scalebox{0.6}{$u_{\sigma_{swap}} \  \defeq \  \{ \mid \ps{x \ptimes y} \iso \ps{y \ptimes x} \}$}
			\]
			\\[0.05cm]
	\end{subfigure}
	\begin{subfigure}[b]{0.3\textwidth}
		\[
			\scalebox{0.68}{                                                                                      
			\begin{quantikz}[background color=gray!10]
				\lstick{$\ps x_1$}&
				&\gategroup[10,steps=3,style={dashed,rounded
				corners,fill=gray!10, inner
				xsep=2pt},background,label style={label
				position=below,anchor=north,yshift=-0.2cm}]{$\isoterm_{\sigma_{divide}}$} & & & &\rstick{$\ps x_1$}\\ 
				&\setwiretype{n} \wvdots & & \wvdots &  &\wvdots &\\
				\lstick[6]{$\ps x \ptimes \ps y$}& & & & & &\rstick[3]{$\ps x$}\\
				&\setwiretype{n} \wvdots &   & \wvdots &  &\wvdots &\\
				&&&&&&\\
				&&&&&&\rstick[3]{$\ps y$}\\
				&\setwiretype{n} \wvdots &  & \wvdots &   &\wvdots &\\
				&&&&&&\\
				&\setwiretype{n} \wvdots &  & \wvdots &  &\wvdots &\\
				\lstick{$\ps x_n$} &&&&& &\rstick{$\ps x_n$}
			\end{quantikz}
			}
		\]
		\[
			\scalebox{0.6}{
				$u_{\sigma_{divide}} \defeq \{ \mid\ps{x_1 \otimes  \cdots \otimes (x \otimes y)\otimes  \cdots\otimes x_n} $}
			\]
			\[
				\scalebox{0.6}{
					$\quad \iso\ps{x_1 \otimes \cdots \otimes x \otimes y \otimes \cdots  \otimes x_n} \}$
				}
			\]
			\label{fig:unitary_i}
	\end{subfigure}
	\label{fig:unitaries}
	\vfill
	\begin{subfigure}{\textwidth}
		\begin{align*}
			(\pt , \isoterm_{\sigma} , \tcase{\tinl{V}}{M}{N}) &\rightarrow_1 ( \pt , \isoterm_{\sigma} , M[V/x])\\
			(\pt , \isoterm_{\sigma} ,  \tcase{\tinr{V}}{M}{N}) & \rightarrow_1 ( \pt , \isoterm_{\sigma} , N[V/y])\\
			(\pt , \isoterm_{\sigma} , \tletpair{x}{y}{V \otimes W}{M}) &\rightarrow_1 ( \pt , \isoterm_{\sigma} , M[V/x,W/y])\\
			(\pt , \isoterm_{\sigma} , \tforce{\tlift{M}}) &\rightarrow_1 ( \pt , \isoterm_{\sigma} , M)\\
			(\pt , \isoterm_{\sigma} , (\lambda x.M)\, V) &\rightarrow_1 ( \pt , \isoterm_{\sigma} , M[V/x])\\
			(\pt , \isoterm_{\sigma} , \tmatch{\tzero}{M}{N}) &\rightarrow_1 (\pt , \isoterm_{\sigma} , M)\\
			(\pt , \isoterm_{\sigma} , \tmatch{\tsucc{V}}{M}{N}) &\rightarrow_1 (\pt , \isoterm_{\sigma} , N[V/x])\\
			(\pt , \isoterm_{\sigma} , \tpure{\pt'}) &\rightarrow_1 (\pt \otimes \pt', \isoterm_{\sigma_{swap}} \circ (\isoterm_{\sigma}\otimes u_{id}) , x)\\
			( {\Sigma}_i p_i {\cdot} \Sigma_j {\ps \alpha_{ij}} {\cdot}  \pb'_{ij} \ptimes {\ps \cdots} \ptimes \pb_i \ptimes {\ps \cdots}, \isoterm_{\sigma_s}  , \tmeas{x}) & \rightarrow_{|p_k|^2} ( \Sigma_j \ps \alpha_{kj} {\cdot} \pb'_{kj} \ptimes {\ps \cdots}, \isoterm_{id} , \ov{\pb_k})\\
			(\pt , \isoterm_{\sigma} , \tun{\isoterm}{z})& \rightarrow_1    ((\isoterm_{\sigma}^* \circ (\isoterm \otimes id) \circ \isoterm_{\sigma})\, \pt , \isoterm_{\sigma} ,  z)\\
			(\pt , \isoterm_{\sigma} , \tletb{z}{x \otimes y}{N}) &  \rightarrow_1  (\pt ,  \isoterm_{\sigma_{gather}} \circ \isoterm_{\sigma} , N)\\
			(\pt , \isoterm_{\sigma}, \tletpairb{x}{y}{z}{N}) & \rightarrow_1       (\pt , \isoterm_{\sigma_{divide}} \circ \isoterm_{\sigma}, N)
		\end{align*}
		\[
			\begin{bprooftree}
				\AxiomC{$(\ps t_1, u_{\sigma_1^{ext}}, M_1) \rightarrow_p (\ps t_2 , u_{\sigma_2^{ext}}, M_2)$}
				\AxiomC{$ dom(\sigma) \cap (dom(\sigma_1) \cup dom(\sigma_2))= \emptyset$}
				\BinaryInfC{$(\ps t_1, u_{\sigma^{ext} \circ \sigma_1^{ext}} , E[M_1]) \rightarrow_p 
				(\ps t_2, u_{\sigma^{ext} \circ \sigma_2^{ext}} , E[M_2])$}
			\end{bprooftree}
		\]
		\label{fig:reduction-rules}
\end{subfigure}
\caption{Reduction rules and their unitary permutations.}
\label{fig:red_uni}
\end{figure*}

%% file: sn_new.tex
In order to prove strong normalization for the relation $\leadsto$, we define
another relation $\brr$ on $\cterms$ of the main calculus, and prove that it strongly
normalizes. We also show that strong normalization of $\brr$ implies strong
normalization of $\leadsto$.

\begin{defi}[Term reduction]
$ \brr\  \in \cterms \times [0,1] \times \cterms$ is defined as follows:
\[
\begin{bprooftree}
\AxiomC{$(\ps v_1, u_{\sigma_1}, M_1) \leadsto_p (\ps v_2, u_{\sigma_2}, M_2)$}
\UnaryInfC{$M_1 \brr_p M_2$}
\end{bprooftree}
\]
\end{defi}

We additionally have a contextual reduction rule for the relation $ \brr $ defined as:
\[
\begin{bprooftree}
\AxiomC{$M \brr_p M'$}
\UnaryInfC{$E[M] \brr_p E[M']$}
\end{bprooftree}
\]

\begin{defi}[Well-typed term]
We say a term $M \in \cterms$ is well-typed of type $A[\ptype_1',\ldots, \ptype_k']$ if
there exist $\ps t, u_{\sigma}$ such that $(\ps t, u_{\sigma}, M)$ is a well-formed
configuration of type $A[\ptype_1',\ldots, \ptype_k']$. We denote such a term by
$M:A[\ptype_1',\ldots, \ptype_k']$.
\end{defi}

Let $\mathcal{SN}_C$ denote the set of well-formed strongly normalizing quantum configurations
under the relation $\leadsto$. Let $\mathcal{SN}_T$ denote the set of strongly normalizing
terms from the main calculus under the relation $\brr$. We denote a reduction path 
$M \brr_{p_0} M_1 \brr_{p_1}\ldots \brr_{p_n} V$ by $M \brr_{\ast} V$.

\begin{lem}
\label{lem:st-conf}
For $M \in \mathcal{SN}_T$, if $\ps v, u_{\sigma}$ are such that $\mathcal C = (\ps v, u_{\sigma}, M)$ is well-formed
then $\mathcal C \in \mathcal{SN}_C$.
\end{lem}

To prove strong normalization of $\brr$, corresponding to each type we define the following sets:
\begin{itemize}
	\item $\mathcal S_T[I] \defeq \{M | M \in \mathcal{SN}_T$ and
		$M \brr_{\ast} V \implies V = \ast:I[\ptype_1',\ldots,\ptype_k'] $ for some $k \}$
	\item $\mathcal S_T[A+B] \defeq \{M | M \in \mathcal{SN}_T$ and
                $M \brr_{\ast} V \implies V = \tinl{V}:A+B[\ptype_1',\ldots,\ptype_k']$ or
                $V = \tinr{W}:A+B[\ptype_1',\ldots,\ptype_k']$
                for some $k\}$
	\item  $ \mathcal S_T[A\otimes B] \defeq \{M | M \in \mathcal{SN}_T$ and
                $M \brr_{\ast} V \implies V = V \otimes W:A\otimes B[\ptype_1',\ldots,\ptype_k']$
                for some $k\}$
	\item $\mathcal  S_T[!A] \defeq \{M | M \in \mathcal{SN}_T$ and
		$M \brr_{\ast} V \implies V = \tlift{M'}:!A[\ptype_1',\ldots,\ptype_k']$
		for some $k$ and $\tforce{V} \in \mathcal S_T[A] \}$
       \item  $\mathcal S_T[\Nat] \defeq \{M | M \in \mathcal{SN}_T$ and
               $M \brr_{\ast} V \implies V = \tzero:\Nat[\ptype_1',\ldots,\ptype_k']$  or
               $V= \tsucc{V}:\Nat[\ptype_1',\ldots,\ptype_k']$
               for some $k \}$
       \item  $\mathcal S_T[A \multimap B] \defeq \{M | M \in \mathcal{SN}_T$
              and $ \forall N  \in \mathcal S_T[A]$ ,
              $MN \in \mathcal S_T[B] \}$
      \item $\mathcal S_T[\mathcal B(\ptype)] \defeq \{M | M \in \mathcal{SN}_T$ and
              $M \brr_{\ast} V \implies  V = x:\mathcal B(\ptype)[\ptype_1',\ldots,\ptype_k']$
              for some $k\}$
\end{itemize}

\begin{defi}[Coherent configurations]
        A quantum configuration $(\ps t, u_{\sigma}, M)$ is said to be \textit{coherent} with respect to a typing context $\Gamma$ with type $A[\ptype'_1, \ldots, \ptype'_k]$
        if the following hold:
        \begin{itemize}
         \item $\cdot \pdash \ps t \colon \ptype_1 \ptimes {\ps \cdots} \ptimes \ptype_n$ can be derived;
       \item $\cdot \udash \isoterm_{\sigma} \colon U(\ptype_1 \ptimes {\ps \cdots} \ptimes \ptype_n , \ptype_1'' \ptimes {\ps \cdots} \ptimes \ptype_{m+k}'')$
             can be derived;
       \item $\Gamma ,
             x_1 : \mathcal B(\ptype_1'') , \ldots , 
               x_m : \mathcal B(\ptype_m'') \vdash M : \mathrm{A}$
             can be derived;
        \item $[\ptype_1'', \ldots, \ptype_{m+k}''] = [\ptype_1'', \ldots, \ptype_m'',\ptype_1', \ldots, \ptype_k']$ as lists of pure quantum types.
        \end{itemize}
\end{defi}

\begin{defi}[Coherent term]
A term $M$ is said to be \textit{coherent} with respect to a typing context $\Gamma$ with type $A[\ptype'_1, \ldots, \ptype'_k]$ if
there exist $\ps t, u_{\sigma}$ such that $(\ps t, u_{\sigma}, M)$ is a coherent configuration with respect to $\Gamma$.
\end{defi}

\begin{rem}
Note that when $\Gamma = \emptyset$ a coherent configuration is just a well-typed configuration. Similarly, in this case
a coherent term is just a well-typed term.
\end{rem}

\begin{defi}[Substitution function]
\label{def:subst-f}
A substitution function is a map $l : \cvariables \rightarrow \cvalues$. For a term $M$, $Ml$ denotes the term obtained after
substituting the free variables in $M$ with their image under $l$.
\end{defi}

For a typing context $\Gamma$,
we define a set $\mathcal S_C[\Gamma] \defeq \{l | dom(l) = dom(\Gamma) \text{ and } x:A \in \Gamma
\implies lx \in S_T[A] \}$.

Now we prove a series of results, which lead to Theorem~\ref{thm:st-norm-strong}.

\begin{thm}[Progress]
\label{thm:progress-t}
For $M:A[\ptype_1',\ldots,\ptype_k']$, there are two possibilities:
\begin{itemize}
\item either $M \in \cvalues$ ,
\item or there exists $M_1 \in A[\ptype_1',\ldots,\ptype_k']$ such that $M \brr_{p_1} M_1$.
 \end{itemize}
Let $\{ M_i \}_{i \in I}$ be the set  of all such distinct (i.e., not $\alpha$-equivalent) terms.
Then, the set $I$ is finite, and $\Sigma_{i \in I}p_i = 1$.
\end{thm}

\begin{thm}[Subject Reduction]
\label{thm:sub-red-t}
For any $M:A[\ptype_1', \ldots, \ptype_{\textit{k}}']$,  if $M \brr_p M'$, for some probability $p \in [0,1]$, then
$M':A[\ptype_1', \ldots, \ptype_{\textit{k}}']$.
\end{thm}

\begin{lem}
\label{lem:st-set}
For a type $A$, if $M \in \mathcal S_T[A]$ then $M \in \mathcal{SN}_T$.
\end{lem}

\begin{lem}
\label{forward-step}
For terms $M, M'$, if $M \in \mathcal S_T[A]$ for some $A$, and $M \brr_p M'$ then $M' \in S_T[A]$.
\end{lem}

\begin{lem}
\label{lem:beta-red}
        If $M[V/x] \in \mathcal S_T[A]$, then $(\lambda x.M)\, V \in \mathcal S_T[A]$.
\end{lem}

\begin{thm}
\label{thm:l-rel}
If $M$ is a coherent term with respect to $\Gamma$ of type $A[\ptype_1',\ldots,\ptype_k']$
then $\forall l \in \mathcal S_C[\Gamma], Ml \in \mathcal S_T[A]$.
\end{thm}
\begin{proof}
 Proof is by induction on the formation rule for the term $M$. We present some important cases here:
        \begin{itemize}
                \item $\lambda x.M$:\\
                     We need to show that $\forall l \in \mathcal S_C[\Gamma], \ Ml \in \mathcal S_T[A \multimap B]$.\\
                     Note that $(\lambda x. M)l = \lambda x. Ml$ upto $\alpha-$equivalence.
                     Suffices to check that for $N \in \mathcal S_T[A]$
		     $(\lambda x.Ml)\, N \in \mathcal S_T[B]$.
                     Since $N \in \mathcal S_T[A]$, all reductions from this term terminate in finitely many steps
                     to a value of the form $V$.
                     For each such reduction, we can apply the contexutal rule for the relation $\brr$ to obtain
                     that $(\lambda x.Ml)\, N$ reduces in each case to a term of the form
                     $(\lambda x.Ml)V$ in finitely many steps.
                     Note that from Lemma~\ref{lem:beta-red}, it suffices to
                     show that $(Ml)[V/x] \in \mathcal S_T[B]$.
                     We know by subject reduction that $(Ml)[V/x]$
                     is a well-typed term. It follows that $l' = l \cup \{x \mapsto V\} \in \mathcal S_C[\Gamma, x: A]$.
                     Hence from induction hypothesis, it follows that $Ml' \in \mathcal S_T[B]$.
                     But $Ml'= M(l\cup \{x \mapsto V \} = Ml[V/x]$. Hence, $Ml[V/x] \in \mathcal S_T[B]$.
                \item $MN$:\\
                      We need to show that for any $l \in \mathcal S_C[\Gamma], (MN)l \in \mathcal S_T[A]$.\\
                      By the induction hypothesis, we know that $Ml \in \mathcal S_T[B \multimap A]$, and
                      $Nl \in \mathcal S_T[B]$. By the definition of $\mathcal S_T[B \multimap A]$ it follows that
                      $(Ml)(Nl) \in \mathcal S_T[A]$. However, note that $(Ml)(Nl) = (M\, N)l$.
                      Thus, we have shown that $(M\, N)l \in \mathcal S_T[A]$ as desired.
                \item $\tpure{\ps t'}$:\\
                      We need to show that for any $l \in \mathcal S_C[\Gamma]$,
                      $\tpure{\ps t'}l \in \mathcal S_T[\mathcal B(\ptype)]$.\\
                      This follows in a straightforward manner since
                      $(\tpure{\ps t'}l = \tpure{\ps t'})$, and from the reduction rules in Figure~\ref{fig:red_uni}, it follows that
		      there are $\ps v, u_{\sigma}$ such that $(\ps v, u_{\sigma}, \tpure{\ps t'})$
                      reduces to a value configuration of the form $(\ps v', u_{\sigma}, x)$. Hence $M \brr_1 x$ and we have that
                      $\tpure{\ps t'} \in \mathcal S_T[\mathcal B(\ptype)]$.
                \item $\tmeas{M}$:\\
                      We need to show that for any $l \in \mathcal S_C[\Gamma]$,
                      $\tmeas{M}l \in \mathcal S_T[\ov{\ptype}]$.\\
                      Note that $\tmeas{M}l = \tmeas{Ml}$. Now, applying
                      the induction hypothesis, we have that
                      $ Ml \in \mathcal S_T[\mathcal B(\ptype)]$. Hence, there exists a $k$ such that $Ml$
                      reduces to a term of the form $x$ in atmost $k$ reduction steps. Applying the contextual reduction rule above
                      we get that $\tmeas{M}$ reduces to a term of the form $\tmeas{x}$
                      in atmost $k$ steps. Finally, applying the reduction rule for measurement from Figure~\ref{fig:red_uni} gives us a
                      value configuration of the desired form. Hence, we have that $\tmeas{M}l \in \mathcal S_T[\ov{\ptype}]$.
                \item $\tun{u}{M}$:\\
                      This case follows by an argument similar to the case for measurement. \qedhere
        \end{itemize}
\end{proof}

\begin{thm}
\label{thm:st-norm-t}
$M:\mathrm{A}[\ptype_1',\ldots,\ptype_k'] \implies M \in \mathcal{SN}_T$.
\end{thm}
\begin{proof}
We know from Theorem~\ref{thm:l-rel} , that for all
$l \in \mathcal S_C[\emptyset], Ml \in \mathcal S_T[A]$. But if $\Gamma= \emptyset, Ml = M$.
Hence $M \in \mathcal S_T[A]$. From Lemma~\ref{lem:st-set} it follows that $M \in \mathcal{SN}_T$.
\end{proof}

\begin{thm}
\label{thm:st-norm-strong}
$\mathcal C \in \wfconfs{\svconfs}{\mathrm{A}[\ptype_1',\ldots,\ptype_k']} \implies \mathcal C \in \mathcal{SN}_C$.
\end{thm}
\begin{proof}
Follows from Lemma~\ref{lem:st-conf} and Theorem~\ref{thm:st-norm-t}.
\end{proof}

%% file: typing-judgements-c.tex
\begin{figure}
\begin{align*}
&\lrb{\Gamma, y:\mathrm{B} , x:\mathrm{A} , \Sigma \vdash M : \mathrm{C}} \eqdef\lrb{\Gamma, x:\mathrm{A} , y:\mathrm{B} , \Sigma \vdash M : \mathrm{C}} \circ (id_{\lrb{\Gamma}} 	\otimes \gamma \otimes id_{\lrb{\Sigma}})\\
&\lrb{!\Delta \vdash * : \mathrm{I}} \eqdef\lrb{!\Delta} \xrightarrow{\diamond_{!\Delta}} \lrb{\mathrm{I}}\\
&\lrb{!\Delta,x: \mathrm{A} \vdash x: \mathrm{A}} \eqdef \lrb{!\Delta} \otimes \lrb{\mathrm{A}} \xrightarrow{\diamond_{!\Delta} \otimes id} \lrb{\mathrm{I}} \otimes \lrb{\mathrm{A}} \cong \lrb{\mathrm{A}}\\
&\lrb{\Delta \vdash \tinl{N} : \mathrm{A + B}} \eqdef i_1 \circ \lrb{\Delta \vdash N : \mathrm{A}}\\
&\lrb{\Delta \vdash \tinr{M} : \mathrm{A + B}} \eqdef i_2 \circ \lrb{\Delta \vdash M : \mathrm{B}}\\
&\lrb{!\Delta, \Sigma_1 , \Sigma_2 \vdash \tcase{L}{M}{N} : \mathrm{C}} \\
&\quad \quad \eqdef[\lrb{!\Delta, \Sigma_2, x : \mathrm{A} \vdash M : \mathrm{C}} , \lrb{!\Delta, \Sigma_2 , y : \mathrm{B} \vdash N : \mathrm{C}}] 	\circ \cong \circ (\lrb{!\Delta, \Sigma_1 \vdash L : \mathrm{A + B}}\otimes id) \circ \textit{split}\\
&\lrb{!\Delta, \Sigma_1 , \Sigma_2 \vdash M \otimes N : \mathrm{A \otimes B}} \eqdef(\lrb{!\Delta, \Sigma_1 \vdash M : \mathrm{A}} \otimes  \lrb{!\Delta, \Sigma_2 \vdash N : \mathrm{B}}) \circ \textit{split}\\
&\lrb{!\Delta, \Sigma_1 , \Sigma_2 \vdash \tletpair{x}{y}{M}{N} : \mathrm{C}}  \\
&\quad \quad \eqdef \lrb{!\Delta, x: \mathrm{A} , \textit{y} : \mathrm{B} , \Sigma_2 \vdash N : \mathrm{C}} 	\circ (\lrb{!\Delta, \Sigma_1 \vdash M : \mathrm{A \otimes B}} 		\otimes id) \circ \textit{split}\\
&\lrb{!\Delta \vdash \tlift{M} :\ !\mathrm{A}} \eqdef\ !\lrb{!\Delta \vdash M : \mathrm{A}} \circ \rho_{!\Delta}\\
&\lrb{\Delta \vdash \tforce{M} : \mathrm{A}} \eqdef\mu \circ \lrb{\Delta \vdash M : !\mathrm{A}}\\
&\lrb{\Delta \vdash \lambda x.M : \mathrm{A \multimap B}} \eqdef\Phi_2 \circ \Phi_1(\lrb{\Delta, x : \mathrm{A} \vdash M : \mathrm{B}})\\
&\lrb{!\Delta, \Sigma_1, \Sigma_2 \vdash M\,N : \mathrm{B}} \eqdef \epsilon_{\mathrm{B,A}}\circ (\lrb{!\Delta, \Sigma_1 \vdash M : \mathrm{A \multimap B}} 
	\otimes \lrb{!\Delta, \Sigma_2 \vdash N:\mathrm{A}}) \circ \textit{split}\\
&\lrb{!\Delta \vdash \tzero : \Nat} \eqdef\lrb{!\Delta} \xrightarrow{\diamond_{!\Delta}} \lrb{I} \xrightarrow{i_0} \lrb{\Nat}\\
&\lrb{\Delta \vdash \tsucc{M}} \eqdef s \circ \lrb{\Delta \vdash M : \Nat}\\
&\lrb{!\Delta, \Sigma_1 , \Sigma_2 \vdash \tmatch{L}{M}{N} : \mathrm{A}} \\
&\quad \quad \eqdef[\lrb{!\Delta, \Sigma_2 \vdash M : \mathrm{A}}, \lrb{!\Delta, \Sigma_2 , x :\Nat \vdash N : \mathrm{A}}] 
	\circ \cong \circ (id \otimes \lrb{!\Delta, \Sigma_1 \vdash L : \Nat}) \circ \textit{split} \\
&\lrb{!\Delta \vdash \tpure{\ps t}: \mathcal B(\ptype)} \eqdef\mathcal B\lrb{\ps t} \circ \diamond_{!\Delta} \\
&\lrb{\Delta \vdash \tmeas{M} : \ov{\ptype}} \eqdef	m^{\ptype} \circ \lrb{\Delta \vdash M : \mathcal B( \ptype)}\\
&\lrb{\Delta \vdash \tun{\isoterm}{M} : \mathcal B(\ptype_2)} \eqdef\mathcal B\lrb{\isoterm} \circ \lrb{\Delta \vdash M : \mathcal B(\ptype_1)}\\
&\lrb{!\Delta , \Sigma_1 , \Sigma_2 \vdash \tletb{z}{M}{N} : \mathrm{A}} \\
&\quad \quad \eqdef\lrb{!\Delta , \Sigma_2 , z: \mathcal B(\ptype_1 \ps \otimes \ptype_2) \vdash N : \mathrm{A}} 
	\circ 
	(\lrb{!\Delta , \Sigma_1 \vdash M : \mathcal B (\ptype_1) \otimes \mathcal B (\ptype_2)}\otimes id) \circ \textit{split}\\
&\lrb{!\Delta , \Sigma_1 , \Sigma_2 \vdash \tletpairb{x}{y}{M}{N} : \mathrm{A}} \\
& \quad \quad \eqdef\lrb{!\Delta , \Sigma_2 , x : \mathcal B(\ptype_1) , y : \mathcal B(\ptype_2) \vdash N : \mathrm{A}} 
	\circ
	(\lrb{!\Delta , \Sigma_1 \vdash M : \mathcal B (\ptype_1 \ps \otimes \ptype_2)}\otimes id) \circ \textit{split}
\end{align*}
\caption{Interpretation of typing judgments.}
\label{fig:typing-judgements}
\end{figure}

%% file: references.bib
@phdthesis{pablo2022unbounded,
	author = {Andrés-Martínez, Pablo},
	title = {Unbounded loops in quantum programs: categories and weak-while loops},
	school = {Laboratory for Foundations of Computer Science, School of Informatics,
		University of Edinburgh},
	doi = {https://doi.org/10.48550/arXiv.2212.05371},
	year = {2022}
}

@inproceedings{childs2003exponential,
	title={Exponential algorithmic speedup by a quantum walk},
	author={Childs, Andrew M and Cleve, Richard and Deotto, Enrico and Farhi, Edward and Gutmann, Sam and Spielman, Daniel A},
	booktitle={Proceedings of the thirty-fifth annual ACM symposium on Theory of computing},
	pages={59--68},
	year={2003}
}

@phdthesis{louis-thesis,
	TITLE = {{The Semantics of Effects: Centrality, Quantum Control and Reversible Recursion}},
	AUTHOR = {Lemonnier, Louis},
	URL = {https://theses.hal.science/tel-04625771},
	NUMBER = {2024UPASG030},
	SCHOOL = {{Universit{\'e} Paris-Saclay}},
	YEAR = {2024},
	MONTH = Jun,
	TYPE = {Theses},
	HAL_ID = {tel-04625771},
	HAL_VERSION = {v1},
}

@article{paykin2017qwire,
	title={QWIRE: a core language for quantum circuits},
	author={Paykin, Jennifer and Rand, Robert and Zdancewic, Steve},
	journal={ACM SIGPLAN Notices},
	volume={52},
	number={1},
	pages={846--858},
	year={2017},
	publisher={ACM New York, NY, USA}
}

@article{cho2014semantics,
	title={Semantics for a Quantum Programming Language by Operator Algebras},
	volume={172},
	ISSN={2075-2180},
	url={http://dx.doi.org/10.4204/EPTCS.172.12},
	DOI={10.4204/eptcs.172.12},
	journal={Electronic Proceedings in Theoretical Computer Science},
	publisher={Open Publishing Association},
	author={Cho, Kenta},
	year={2014},
	month=dec, pages={165–190} 
}

@InProceedings{dave2025control,
	author="Dave, Kinnari
		and Lemonnier, Louis
		and P{\'e}choux, Romain
		and Zamdzhiev, Vladimir",
	editor="Abdulla, Parosh Aziz
		and Kesner, Delia",
	title="Combining quantum and classical control: syntax, semantics and adequacy",
	booktitle="Foundations of Software Science and Computation Structures",
	year="2025",
	publisher="Springer Nature Switzerland",
	address="Cham",
	pages="155--175",
	isbn="978-3-031-90897-2"
}

@inproceedings{sabry2018symmetric,
	author = {Amr Sabry and Beno{\^{\i}}t Valiron and Juliana Kaizer Vizzotto},
	title = {From Symmetric Pattern-Matching to Quantum Control},
	pages = {348--364},
	doi = {10.1007/978-3-319-89366-2\_19},
	booktitle = {FoSSaCS 2018)},
	series = {Lecture Notes in Computer Science},
	volume = {10803},
	publisher = {Springer},
	year = {2018},
}

@article{qunity,
	author       = {Finn Voichick and
		Liyi Li and
			Robert Rand and
			Michael Hicks},
	title        = {Qunity: {A} Unified Language for Quantum and Classical Computing},
	journal      = {Proc. {ACM} Program. Lang.},
	volume       = {7},
	number       = {{POPL}},
	pages        = {921--951},
	year         = {2023},
	url          = {https://doi.org/10.1145/3571225},
	doi          = {10.1145/3571225},
	bibsource    = {dblp computer science bibliography, https://dblp.org}
}

@phdthesis{phd-kostia,
	author       = {Kostia Chardonnet},
	title        = {Towards a Curry-Howard Correspondence for Quantum Computation. (Vers
			une correspondance de Curry-Howard pour le calcul quantique)},
	school       = {Université Paris-Saclay, France},
	year         = {2023},
	url          = {https://tel.archives-ouvertes.fr/tel-03959403},
	bibsource    = {dblp computer science bibliography, https://dblp.org}
}

@article{lineal,
	author       = {Pablo Arrighi and
		Gilles Dowek},
	title        = {Lineal: {A} linear-algebraic Lambda-calculus},
	journal      = {Log. Methods Comput. Sci.},
	volume       = {13},
	number       = {1},
	year         = {2017},
	url          = {https://doi.org/10.23638/LMCS-13(1:8)2017},
	doi          = {10.23638/LMCS-13(1:8)2017},
	bibsource    = {dblp computer science bibliography, https://dblp.org}
}

@article{enriched-presheaf,
	author       = {Takeshi Tsukada and
		Kazuyuki Asada},
	title        = {Enriched Presheaf Model of Quantum {FPC}},
	journal      = {Proc. {ACM} Program. Lang.},
	volume       = {8},
	number       = {{POPL}},
	pages        = {362--392},
	year         = {2024},
	url          = {https://doi.org/10.1145/3632855},
	doi          = {10.1145/3632855},
	bibsource    = {dblp computer science bibliography, https://dblp.org}
}

@inproceedings{popl22,
	author    = {Xiaodong Jia and
		Andre Kornell and
			Bert Lindenhovius and
			Michael W. Mislove and
			Vladimir Zamdzhiev},
	title     = {Semantics for variational Quantum programming},
	booktitle   = {POPL 2022},
	volume    = {6},
	pages     = {1--31},
	year      = {2022},
	publisher = {{ACM}},
	doi       = {10.1145/3498687}
}

@inproceedings{qpl-fossacs,
	author    = {Romain P{\'{e}}choux and
		Simon Perdrix and
			Mathys Rennela and
			Vladimir Zamdzhiev},
	title     = {Quantum Programming with Inductive Datatypes: Causality and Affine
		Type Theory},
	booktitle = {FoSSaCS 2020},
	series    = {Lecture Notes in Computer Science},
	volume    = {12077},
	pages     = {562--581},
	publisher = {Springer},
	year      = {2020},
	doi       = {10.1007/978-3-030-45231-5\_29}
}

@misc{kenta-bram,
	author    = {Kenta Cho and
		Abraham Westerbaan},
	title     = {{V}on {N}eumann {A}lgebras form a {M}odel for the {Q}uantum {L}ambda {C}alculus},
	year      = {2016},
	doi	  = {https://doi.org/10.48550/arXiv.1603.02133}
}

@article{kornell2012quantum,
	author = {Kornell, Andre},
	title = {{Quantum collections}},
	journal = {International Journal of Mathematics},
	volume = {28},
	number = {12},
	year = {2017},
	doi = {10.1142/S0129167X17500859}
}

@phdthesis{marc-phd,
	author       = {Marc de Visme},
	title        = {Quantum Game Semantics. (S{\'{e}}mantique des Jeux Quantique)},
	school       = {University of Lyon, France},
	year         = {2020},
	url          = {https://tel.archives-ouvertes.fr/tel-03045844},
	bibsource    = {dblp computer science bibliography, https://dblp.org}
}

@article{quantum-game-full-abstraction,
	author       = {Pierre Clairambault and
		Marc de Visme},
	title        = {Full abstraction for the quantum lambda-calculus},
	journal      = {Proc. {ACM} Program. Lang.},
	volume       = {4},
	number       = {{POPL}},
	pages        = {63:1--63:28},
	year         = {2020},
	url          = {https://doi.org/10.1145/3371131},
	doi          = {10.1145/3371131},
	bibsource    = {dblp computer science bibliography, https://dblp.org}
}

@article{quantum-game,
	author    = {Pierre Clairambault and
		Marc de Visme and
			Glynn Winskel},
	title     = {{Game semantics for quantum programming}},
	journal   = {{PACMPL}},
	volume    = {3},
	number    = {{POPL}},
	pages     = {32:1--32:29},
	year      = {2019},
	doi       = {10.1145/3290345},
	bibsource = {dblp computer science bibliography, https://dblp.org}
}

@article{semimodules,
	author       = {Alejandro D{\'{\i}}az{-}Caro and
		Octavio Malherbe},
	title        = {Semimodules and the (syntactically-)linear lambda calculus},
	journal      = {CoRR},
	volume       = {abs/2205.02142},
	year         = {2022},
	url          = {https://doi.org/10.48550/arXiv.2205.02142},
	doi          = {10.48550/ARXIV.2205.02142},
	eprinttype    = {arXiv},
	eprint       = {2205.02142},
	bibsource    = {dblp computer science bibliography, https://dblp.org}
}

@inproceedings{quantitative,
	author    = {Michele Pagani and
		Peter Selinger and
			Beno{\^{\i}}t Valiron},
	title     = {Applying quantitative semantics to higher-order quantum computing},
	booktitle = {POPL 2014},
	pages     = {647--658},
	publisher = {{ACM}},
	year      = {2014},
	doi       = {10.1145/2535838.2535879}
}

@book{takesaki,
	title={Theory of Operator Algebras I},
	author={Takesaki, Masamichi},
	isbn={9781461261889},
	series={Encyclopaedia of Mathematical Sciences},
	year={2012},
	publisher={Springer New York}
}

@article{heunen2022axioms,
	title={Axioms for the category of {H}ilbert spaces},
	author={Heunen, Chris and Kornell, Andre},
	journal={Proceedings of the National Academy of Sciences},
	volume={119},
	number={9},
	pages={e2117024119},
	year={2022},
	publisher={National Acad Sciences}
}

@article{selinger2009quantum,
	title={Quantum lambda calculus},
	author={Selinger, Peter and Valiron, Beno{\^{\i}}t},
	journal={Semantic techniques in quantum computation},
	pages={135--172},
	year={2009}
}

@book{heunen2015categories,
	title={Categories for Quantum Theory: an introduction},
	author={Heunen, Chris and Vicary, Jamie},
	year={2019},
	publisher={Oxford University Press}
}

@book{nielsen2001quantum,
	title={Quantum computation and quantum information},
	author={Nielsen, Michael A. and Chuang, Isaac L.},
	year={2010},
	publisher={Cambridge University Press}
}

@article{selinger2004towards,
	title={Towards a quantum programming language},
	author={Selinger, Peter},
	journal={Mathematical Structures in Computer Science},
	volume={14},
	number={4},
	pages={527--586},
	year={2004},
	publisher={Cambridge University Press},
	doi ={https://doi.org/10.1017/S0960129504004256}
}

@article{selinger2006lambda,
	title={A lambda calculus for quantum computation with classical control},
	author={Selinger, Peter and Valiron, Beno{\^\i}t},
	journal={Mathematical Structures in Computer Science},
	volume={16},
	number={3},
	pages={527--552},
	year={2006},
	doi          = {10.1017/S0960129506005238}
}

@inproceedings{diazcaro2019unitary,
	doi = {10.1109/lics.2019.8785834},
	url = {https://doi.org/10.1109%2Flics.2019.8785834},
	year = 2019,
	month = {jun},
	publisher = {{IEEE}
	},
	author = {Alejandro Díaz-Caro and Mauricio Guillermo and Alexandre Miquel and Benoit Valiron},
	title = {Realizability in the Unitary Sphere},
	booktitle = {LICS 2019}
}

@article{diazcaro2022unitary,
	doi = {10.46298/lmcs-18(3:32)2022},
	url = {https://doi.org/10.46298%2Flmcs-18%283%3A32%292022},
	year = 2022,
	volume = {Volume 18, Issue 3},
	author = {Alejandro Díaz-Caro and Octavio Malherbe},
	title = {Quantum Control in the Unitary Sphere: Lambda-S1 and its Categorical Model},
	journal = {Logical Methods in Computer Science}
}

@article{valiron2022semantics,
	title={Semantics of quantum programming languages: Classical control, quantum control},
	author={Valiron, Beno{\^\i}t},
	journal={Journal of Logical and Algebraic Methods in Programming},
	volume={128},
	year={2022},
	publisher={Elsevier},
	doi={https://doi.org/10.1016/j.jlamp.2022.100790}
}

@article{heunen2022information,
    author = {Heunen, Chris and Kaarsgaard, Robin},
    title = {Quantum information effects},
    year = {2022},
    issue_date = {January 2022},
    publisher = {Association for Computing Machinery},
    address = {New York, NY, USA},
    volume = {6},
    number = {POPL},
    url = {https://doi.org/10.1145/3498663},
    doi = {10.1145/3498663},
    journal = {Proc. ACM Program. Lang.},
    month = {jan},
    articleno = {2},
    numpages = {27}
}

@misc{carette2023quantum,
    title={The Quantum Effect: A Recipe for QuantumPi},
    author={Jacques Carette and Chris Heunen and Robin Kaarsgaard and Amr Sabry},
    year={2023},
    eprint={2302.01885},
    archivePrefix={arXiv},
    primaryClass={cs.PL}
}

@article{pablo2022universal,
    title={Universal Properties of Partial Quantum Maps},
    volume={394},
    ISSN={2075-2180},
    url={http://dx.doi.org/10.4204/EPTCS.394.11},
    DOI={10.4204/eptcs.394.11},
    journal={Electronic Proceedings in Theoretical Computer Science},
    publisher={Open Publishing Association},
    author={Andrés-Martínez, Pablo and Heunen, Chris and Kaarsgaard, Robin},
    year={2023},
    month=nov, pages={192–207}
}

@misc{dimeglio2024dagger,
    title={Dagger categories and the complex numbers: Axioms for the category of finite-dimensional {H}ilbert spaces and linear contractions},
    author={Matthew Di Meglio and Chris Heunen},
    year={2024},
    eprint={2401.06584},
    archivePrefix={arXiv},
    primaryClass={math.CT},
    url={https://arxiv.org/abs/2401.06584},
}

@article{kaarsgaard2017join,
    title =         "Join inverse categories and reversible recursion",
    journal =       "J. Log. Algebraic Methods Program.",
    volume =        87,
    pages =         "33--50",
    year =          2017,
    issn =          "2352-2208",
    doi =           "10.1016/j.jlamp.2016.08.003",
    author =        "Robin Kaarsgaard and Holger Bock Axelsen and Robert Gl{\"u}ck",
}

@misc{heunen2022contractions,
    author = {Chris Heunen and Andre Kornell and Nesta van der Schaff},
    title = {Axioms for the category of {H}ilbert spaces and linear contractions},
    year = {2022},
    eprint = {2211.02688},
    archivePrefix = {arXiv}
}

@article{quipper-rc,
	title = {Proto-Quipper with reversing and control},
	journal = {Journal of Logical and Algebraic Methods in Programming},
	volume = {151},
	pages = {101156},
	year = {2026},
	issn = {2352-2208}, 
	doi = {https://doi.org/10.1016/j.jlamp.2026.101156},
	url = {https://www.sciencedirect.com/science/article/pii/S2352220826000489}, 
	author = {Peng Fu and Kohei Kishida and Neil J. Ross and Peter Selinger},
}

@InProceedings{chard:invrec,
    author =    {Kostia Chardonnet and Louis Lemonnier and Benoît Valiron},
    title = {{Semantics for a Turing-Complete Reversible Programming Language with Inductive Types}},
    booktitle = {Formal Structures for Computation and Deduction (FSCD 2024)},
    pages = {19:1--19:19},
    series =    {Leibniz International Proceedings in Informatics (LIPIcs)},
    year =  {2024},
    volume =    {299},
    doi =       {10.4230/LIPIcs.FSCD.2024.19},
}

@phdthesis{pablo-thesis,
    author = {Andrés-Martínez, Pablo},
    title = {Unbounded loops in quantum programs: categories and weak-while loops},
    school = {Laboratory for Foundations of Computer Science, School of Informatics,
        University of Edinburgh},
    year = {2022}
}

@inproceedings{qunity-compile,
	author       = {Mints, Mikhail and Voichick, Finn and Lampropoulos, Leonidas and Rand, Robert},
	title        = {Compositional Quantum Control Flow with Efficient Compilation in Qunity},
	booktitle    = {Proceedings of the 40th ACM SIGPLAN Conference on Object-Oriented Programming, Systems, Languages, and Applications},
	series       = {Proceedings of the ACM on Programming Languages},
	volume       = {9},
	number       = {OOPSLA2},
	pages        = {166--192},
	year         = {2025},
	publisher    = {Association for Computing Machinery},
	doi          = {10.1145/3763056},
}

@article{chardonnet2026compiling,
	author       = {Chardonnet, Kostia and Dal Lago, Ugo and Hoshino, Naohiko and Pistone, Paolo},
	title        = {Compiling Quantum Lambda-Terms into Circuits via the Geometry of Interaction},
	year         = {2026},
	eprint       = {arXiv:2602.17482},
	url          = {https://arxiv.org/abs/2602.17482},
}

@article{phase,
    author       = {Heunen, Chris and Lemonnier, Louis and McNally, Christopher and Rice, Alex},
    title        = {Quantum circuits are just a phase},
    journal      = {Proceedings of the ACM on Programming Languages},
    volume       = {10},
    number       = {POPL},
    articleno    = {89},
    pages        = {2586--2613},
    year         = {2026},
    publisher    = {Association for Computing Machinery},
    doi          = {10.1145/3776731},
    url          = {https://doi.org/10.1145/3776731}
}
